\documentclass[10pt,aps,twocolumn,prd,superscriptaddress,nofootinbib]{revtex4-2}

\usepackage{amsmath, amssymb, amsthm, mathrsfs}
\usepackage[mathscr]{euscript}
\let \eucal \mathscr
\usepackage{bm, dsfont}
\usepackage[usenames,dvipsnames]{color}
\usepackage{colortbl}
\usepackage{graphicx}
\usepackage{natbib}
\usepackage[colorlinks=true,linkcolor=blue,citecolor=blue,urlcolor=blue]{hyperref}
\usepackage{nameref}
\usepackage{hypcap}
\usepackage{cleveref}
\usepackage{verbatim,  psfrag}
\usepackage[normalem]{ulem}
\usepackage{physics}	
\usepackage{nicefrac,xfrac}
\usepackage{relsize, scalerel}

\usepackage{enumitem}
\usepackage{multirow}
\usepackage{graphicx}
\usepackage{wrapfig}
\usepackage{amsthm}
\newenvironment{proofidea}{%
  \par\noindent\textit{Proof sketch.}\hspace*{1em}%
}{\hfill$\square$\par}

\newcommand{\id}{\mathds{1}}
\newcommand{\ii}{\mathrm{i}}

\newcommand{\cA}{\mathcal{A}}
\newcommand{\cB}{\mathcal{B}}
\newcommand{\cC}{\mathcal{C}}
\newcommand{\cE}{\mathcal{E}}

\newcommand{\cH}{\mathcal{H}}
\newcommand{\cI}{\mathcal{I}}

\newcommand{\cM}{\mathcal{M}}
\newcommand{\cP}{\mathcal{P}}
\newcommand{\cS}{\mathcal{S}}

\newcommand{\cL}{\mathcal{L}}
\newcommand{\cU}{\mathcal{U}}

\newcommand{\cX}{\mathcal{X}}
\newcommand{\cT}{\mathcal{T}}

\newcommand{\be}{\begin{equation}}
\newcommand{\ee}{\end{equation}}

\newcommand{\cN}{\mathcal{N}}

\newcommand{\ChJa}{\text{Choi-Jamio{\l}kowski }}

\DeclareMathOperator*{\bigtimes}{\scalebox{1.5}{$\times$}}

\usepackage{pifont}
\newcommand{\xmark}{\ding{55}}%

\usepackage{tikz}

\newcommand{\proc}[1]{\mathbf{#1}}

\newcommand{\fcj}{F_\text{CJ}}

\newcommand{\channelF}{\mathsf{F}}
\newcommand{\channelD}{\mathsf{D}}
\newcommand{\fid}[2]{{\tr \sqrt{\sqrt{#1} #2 \sqrt{#1}}}}

\newcommand{\tightto}{\!\to\!}
\renewcommand{\d}{\mathrm{d}}
\newcommand{\qfi}{\mathrm{QFI}}
\newcommand{\mqfi}{\mathsf{QFI}}
\newtheorem{theorem}{Theorem}

\newtheorem{corollary}{Corollary}[theorem]
\newtheorem{lemma}[theorem]{Lemma}

\newtheorem{proposition}[theorem]{Proposition}

\newcommand{\yg}[1]{{\color{blue}{#1}}}

\newcommand{\eqnref}[1]{Eq.~(\ref{#1})}
\newcommand{\eqnsref}[2]{Eqs.~(\ref{#1}) and (\ref{#2})}

\newcommand{\secref}[1]{Sec.~\ref{#1}}
\newcommand{\appref}[1]{App.~\ref{#1}}

\begin{document}

\title{Cloning Quantum Channels}

\author{Pavel Sekatski}
\affiliation{Department of Applied Physics, University of Geneva, Switzerland}
\author{Yelena Guryanova}
\affiliation{QuantumBasel, Schorenweg 44b, 4144 Arlesheim, Switzerland}
\affiliation{Center for Quantum Computing and Quantum Coherence (QC2), University of Basel, Petersplatz 1, Basel, 4001, Switzerland}
\author{Naga Bhavya Teja Kothakonda}
\affiliation{F\'isica Te\'orica: Informaci\'o i Fen\`omens Qu\'antics, Departament de Fisica, Universitat Aut\`onoma de Barcelona, E-08193, Bellaterra
(Barcelona), Spain} 
\author{Michalis Skotiniotis}
\affiliation{Departamento de Electromagnetismo y F\'isica de la Materia, Universidad de Granada, 18010 Granada, Spain}
\affiliation{Institute Carlos I for Theoretical and Computational Physics, Universidad de Granada, 18010 Granada, Spain}
\begin{abstract}

We consider the problem of deterministically cloning quantum channels with respect to the best attainable rate and the highest quality, so-called \textit{optimal cloning}. We demonstrate that cloning quantum states is, in-fact, equivalent to cloning the trash-and-replace channel and therefore the former is a special case of the more general problem. By appealing to higher-order quantum operations (quantum processes) we construct a unified framework to deal with the most general cloning tasks and establish necessary conditions for a family of channels to exhibit super-replication—a quadratic cloning rate with vanishing error. We find that noisy phase-gate channels satisfy these conditions, and we construct the explicit super-replicating process for the task. Conversely, we find that the criteria are not met by the full set of noisy unitary gates; classical noise channels; or amplitude damping channels, whose respective cloning rates are at most linear. In this paradigm, we not only derive new results, but also refigure known ones.
We 
derive a strong converse for state cloning,  and
for unitary channels we construct an alternative super-replication process to that of D\"ur {\it et al.}~\cite{Dur2015} and Chiribella {\it et al.}~\cite{Chiribella2015} based on a measure-and-prepare process, which allows us to establish a direct connection between optimal channel cloning and Bayesian channel estimation. Finally we give an SDP  algorithm to search for optimal cloning processes and study the advantage of coherent vs measure-and-prepare protocols on concrete examples.

\end{abstract}
\maketitle

\section{Introduction}

The \emph{no-cloning theorem}~\cite{Park1970,Wootters1982,Dieks1982} is one of the cornerstones of quantum 
information theory, fundamentally distinguishing it from its classical counterpart. Given $N$ copies of a quantum 
system prepared in an arbitrary unknown 
state it is impossible to generate $M>N$ identical copies. 
The no-cloning theorem is deeply connected to the no-signaling 
principle~\cite{Herbert1982,Gisin1998,Ghosh1999,Gedik2013,Sekatski2015} and underpins several fundamental 
primitives in quantum cryptography~\cite{Cerf2002} and computing~\cite{Duan1998}.

Whilst exact cloning is impossible, one may still ask what is the best that can be achieved quantum mechanically. 
Given access to $N$ identical quantum resources, an optimal cloning map produces an output approximating $M>N$ 
ideal copies of the same resource. The performance of such a map depends crucially on how the quality of the 
clones is assessed---typically through suitable distance measures or operationally motivated quantifiers such as 
the fidelity.  For quantum states, optimal cloning maps have been derived in various 
settings~\cite{Buzek1996,Gisin1997,Werner1998, Bruss1998,Duan1998, Pati1999,Chefles1999, Bruss2000, DAriano2003}.

A closely related task to cloning is \emph{replication}.  A replication map 
takes $N$ quantum resources and produces an output that approximates $M>N$ ideal 
copies with a bounded error
$\epsilon\ll1$. In this setting the central quantity of 
interest is the \emph{replication rate}: the scaling of $M$ in the large-$N$ 
limit. Deterministic replication maps for quantum states were shown to achieve a 
linear rate~\cite{Chiribella2013}, whereas for certain subsets of states, for instance
so-called clock states $\{\ket{\psi_t}=e^{\ii t H}\ket{\psi} |t\in \mathds{R}\}$, a 
quadratic replication rate can be achieved using probabilistic replication 
maps~\cite{Chiribella2013}. A map exhibiting a quadratic replication rate is 
called \emph{super-replicating}.

Beyond quantum states, cloning and replication concepts were, more recently, extended to the case of unitary 
gates~\cite{Chiribella2008a,Dur2015, Chiribella2015}. Optimal probabilistic $1\rightarrow 2$ cloning of unitary 
gates was established in~\cite{Chiribella2008a} and, unlike the case of quantum states, deterministic 
super-replication of unitary gates in arbitrary dimensions was demonstrated in~\cite{Dur2015, Chiribella2015}.  
In this work we extend cloning and replication to the most general operations allowed by quantum theory: 
completely positive trace-preserving (CPTP) maps, or quantum channels for short.     

The paper is organised as follows.  After establishing some necessary notation, \secref{sec:background} reviews 
the requisite mathematical background for cloning and replicating of quantum channels, introducing the relevant 
distance measures that we will use throughout the remainder of the work. The section concludes by showing that state 
cloning is equivalent to cloning a corresponding family of trash-and-replace channels 
(Proposition~\ref{prop:cloning-trash}), providing a unifying framework for cloning and replication of all quantum 
resources (states, gates, and channels). \secref{sec:upper_bound} connects quantum channel cloning with 
binary channel discrimination and metrology, which allows us to establish general bounds on channel 
cloning (Proposition~\ref{prop:cloning_bound_discrimination}) as well as necessary conditions for 
super-replication (Proposition~\ref{prop:necessary_conds_superrep}). In turn, these results give us the tools
to prove 
that deterministic replication is limited to a linear rate for three classes of resources:
quantum states (Corollary~\ref{cor:state_rep}); 
unitary gates under the diamond fidelity (Corollary~\ref{cor:unitaries_diamond}); classical noise channels 
(Corollary~\ref{cor:classical noise}). We also show that super-replication of the set of all unitary gates of 
dimension $d$ vanishes, even in the presence of minimal noise (Corollary~\ref{cor:no robustness}), restricting 
deterministic processes to linear scaling.

\secref{sec: clon proc} deals with the construction of optimal cloning and replication processes for 
particular families
of quantum channels. In \secref{sec:sdp} we show that the problem of finding the optimal 
cloning/replicating process can be approximated using a semidefinite program (SDP). In \secref{sec:measureandprep} 
we consider measure-and-prepare processes and establish a quantitative connection with Bayesian channel 
estimation. Somewhat surprisingly we find that measure-and-prepare processes are capable of deterministically 
super-replicating all unitary gates acting on qubits, in stark contrast to the {\emph coherent} processes 
in~\cite{Dur2015,Chiribella2015}. \secref{sec:examples} focuses on cloning and replication of specific families 
of qubit channels, including noisy phase gates (\secref{sec:noisy_unitaries}) (with noise acting before or 
after the application of the gate); Pauli-noise channels (\secref{sec:BF}); and amplitude-damping channels 
(\secref{sec:AD}). We summarize and conclude in Section~\ref{sec:concl}.

\section{Background}\label{sec:background}

In this section we review the necessary background behind cloning and replication of 
quantum states and channels.  In \secref{sec:Notation} we provide a brief review 
of the notation that we will use throughout this work, and in \secref{sec:Cloning} we 
formulate the task of cloning in its most general mathematical form, 
introducing the relevant figures of merit that will be used.

\begin{figure*}
    \centering
    \includegraphics[width=0.75\linewidth]{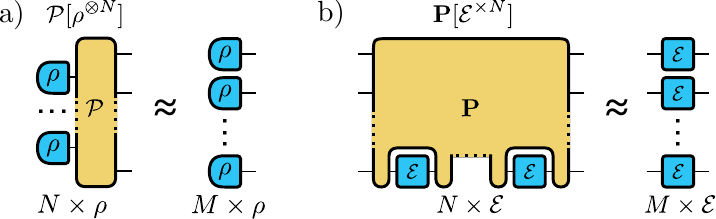}
    \caption{{\bf a)} States: the optimal cloning map  $\cP$ for quantum states receives $N$ copies of an unknown 
    quantum state $\rho$ and produces an output state $\cP[\rho^{\otimes N}]$ that is maximally close to 
    $\rho^{\otimes M}$ (see Eq.~\ref{eq:wc_fid}).  {\bf b)} Channels: The optimal cloning process $\proc{P}$ for 
    quantum channels uses $N$ copies (queries) of an unknown channel $\cE$ to produce a global channel $\proc{ P}
    [\cE^{\times N}]$ which is maximally close to $\cE^{\otimes M}$  (see Eq.~\ref{eq:cloneChan}). Notice that 
    cloning of channels is not equivalent to cloning of their \ChJa states, nor their unitary dilations.}
    \label{fig:enter-label}
\end{figure*}

\subsection{Mathematical Preliminaries and Notation}\label{sec:Notation}

Let us begin by establishing some important definitions and notation.
Quantities (particularly distances) written in sans-serif font will denote the optimized value of the function they pertain to. 
The set of all linear operators acting on a Hilbert space $\cH$ is denoted by $\cL(\cH)$. Density operators 
in $\cL(\cH)$ are written using Greek letters; uppercase Greek letters denote rank-one density operators, 
i.e., $\Psi=\ketbra{\psi}{\psi}, \ket\psi\in\cH$, while lowercase Greek letters will denote general density 
operators. Positive operator valued measures (POVMs) are denoted as $\{{\rm E}_\ell\, \vert {\rm E}_\ell \geq 0\, 
,\sum_\ell {\rm E}_\ell=\id\}$. Quantum channels are completely positive and trace-preserving (CPTP) linear maps 
that take operators from $\cL(\cH_1)$ to $\cL(\cH_2)$, and map density operators to density operators.  
Quantum channels are indicated  with calligraphic upper case Latin characters, i.e., 
\(\cE:\cL(\cH_1)\tightto \cL(\cH_2)\), apart from the identity channel 
which is written
${\rm id}$. So as to avoid 
notation clutter we shall use the short hand $\cL(\cH_1\tightto\cH_2):=\cL(\cH_1)\tightto\cL(\cH_2)$.  
Quantum instruments, denoted as $\{\cI_\ell \}$, are collections of CP maps such that $\sum_\ell \cI_\ell$ is 
also TP. 

Central to the following discussion are superchannels~\cite{Chiribella2008,Chiribella2008c, Chiribella2009}, also 
known as processes or higher-order operations 
(terms we will use interchangeably). These are linear maps 
that take quantum channels from $\cL(\cH_1\tightto\cH_2)$ to $\cL(\cH_{\rm in}\tightto\cH_{\rm out})$. 
A process has three properties: it is trace preserving preserving (TPP) (transforms trace preserving maps to trace
preserving maps); it completely positive preserving (CPP)
(transforms completely positive maps to completely positive maps); it is completely completely positive preserving (CCPP) (the process remains CPP if the channels that it acts on are trivially extended).

We shall denote quantum processes using 
boldface uppercase Latin characters $\proc{P}:\cL(\cH_1\tightto\cH_2)\tightto\cL(\cH_{\rm in}\tightto
\cH_{\rm out})$.  A specific case $({\rm \dim}(\cH_{\rm in})=1)$ are state producing process that map CPTP maps 
to states and will be denoted with the uppercase Latin character 
$\proc{S}:\cL(\cH_1\tightto\cH_2)\tightto\cL(\cH_{\rm out})$. We shall also consider processes which map channels 
to quantum instruments, denoted as $\{\proc{M}_\ell \}$, where each 
$\proc{M}_\ell:\cL(\cH_1\tightto\cH_2)\tightto \cL(\cH_{\rm in}\tightto\cH_{\rm out})$ maps CPTP maps to CP maps 
and $\proc{M}=\sum_\ell \proc{M}_\ell$ returns a TP map. Finally, a specific case  of these superchannels 
$({\rm \dim}(\cH_{\rm in})={\rm \dim}(\cH_{\rm out})=1)$ are measurement processes (also known as channel POVMs), 
$\{\proc{E}_\ell\, \vert \proc{E}_\ell :\cL(\cH_1\tightto\cH_2)\tightto \mathds{R}\}$ which map channels to real 
numbers---the probability of the classical output $\ell$\footnote{All measurement processes can be e.g. realized 
$ \proc{E}_{\ell}[\cdot] = {\rm tr}\, {\rm E}_{\ell}\, \proc{S}[\cdot]$ as a state producing process $\proc{S}$ 
followed by a measurement of the output system represented by a POVM $\{{\rm E}_{\ell}\}$.}.

Particular care must be taken when considering processes that take $N>1$ channels as inputs $\proc{P}: \cL(\cH_1 \tightto
\cH_2)^{\otimes N}\tightto\cL(\cH_{\rm in }\tightto\cH_{\rm out})$. Such 
processes can be divided into two distinct classes: causally ordered and non-causally ordered. The former is comprised of  
processes that operate on $N$ channels in sequence or in parallel and admit a circuit 
implementation~\cite{Chiribella2008c, Gour2019}, while the latter is a more general class and consists of processes that operate on the 
$N$ channels in an indefinite causal order (see \cite{Oreshkov2012,Chiribella2013a} for examples). For more 
details on the characterization of quantum processes we refer the reader to the recent 
review~\cite{Milz2024characterising}.

For both classes we employ the notation $\proc{P}[\cE^{\times N}]$ 
to indicate that the process may use the $N$ input channels in an arbitrary fashion, i.e., 
with
physical access 
to $N$ queries of the channel $\cE$ and not simply a single query of the channel 
$\cE^{\otimes N}$.

Finally, for the purposes of this paper, we shall assign particular understanding to \textit{measure-and-prepare} processes. 
We define these to be a measurement process $ \{\proc{E}_\ell\}$ followed by the implementation of a channel $\cP_\ell$ conditional on the observed classical outcome $\ell$, i.e., $\proc{P}_{\rm M\&P}[\cE^{\times N}] = \sum_\ell \proc{E}_\ell[\cE^{\times N}] \,\cP_\ell.$ In our work, processes that are not measure-and-prepare are termed {\it coherent}. See e.g.~\cite{chen2020entanglement} for a more fine-grained characterization.

\subsection{Quantum Cloning and Replication}\label{sec:Cloning}

Consider a family of quantum states $\eucal{S}\subset\cL(\cH)$.  An $N\tightto M$ quantum cloning map takes as 
input $N$ copies of some state $\rho\in\eucal{S}$ and produces $M>N$ copies.  Hence, an $N\tightto M$ cloning map 
is described by a CPTP map $\cP: \cL(\cH^{\otimes N}\tightto\cH^{\otimes M})$, illustrated in 
Fig.~\ref{fig:enter-label}a.  The performance of the cloning map strongly depends on how one chooses to quantify 
the quality of the copies.  By and large the most widely used figure of merit is the \emph{fidelity}, which 
for two quantum states, $\rho,\sigma\in\cL(\cH)$, is defined~\cite{NielsenChuang00}
	\be\label{eq:fidhalf}
     		F(\rho, \sigma):=\tr |\sqrt{\rho}\sqrt{\sigma}|=\fid{\rho}{\,\sigma}\, .
	\ee

We can quantify the performance of a cloning map based on the \emph{global fidelity} between the 
output $\cP[\rho^{\otimes N}]$ and the ideal target state $\rho^{\otimes M}$~\cite{Werner1998, 
Chefles1999,Bruss2000}, or based on the \emph{per-copy fidelity} between any one of the outputs 
$\sigma_i=\tr_{\neg i}\, \cP[\rho^{\otimes N}]$ and $\rho$~\cite{Buzek1996,Gisin1997,Bruss1998,Bruss2000}.  
Cloning maps with high global fidelity automatically yield high per-copy fidelity (but not vice versa), 
and for this reason we focus on cloning tasks where the figure of merit is the global fidelity from hereon in.


The precise prescription for the optimal cloning map depends strongly on the figure of merit used.  For 
$N\tightto M$ state cloning, optimal cloning maps have been constructed for the \emph{worst case} 
fidelity~\cite{Werner1998} 
	\be
		\channelF^{\eucal{S}}(N,M) : = \inf_{\rho\in\eucal{S}}\, 
		F\left(\cP\left[\rho^{\otimes N}\right], \rho^{\otimes M}
		\right)  \, ,
	\label{eq:wc_fid}
	\ee
as well as for the \emph{average fidelity}~\cite{Gisin1997, Bruss2000, DAriano2003} 
$\max_{\cP}\int_{\eucal{S}} \, p(\rho)\,F\left(\cP\left[\rho^{\otimes N}\right], \rho^{\otimes M}\right)\, 
\mathrm{d}\rho$ where $p(\rho)$ denotes any prior knowledge regarding the state $\rho\in\eucal{S}$ to be cloned. 
Observe that the average fidelity is always greater or equal to the worst case one. Furthermore, in the former, 
the optimal map depends on the choice of the distribution $p(\rho)$ therefore, in this work, we focus on the 
worst case fidelity. Whatever figure of merit is chosen the task of designing the optimal cloning map reduces to 
finding the CPTP map $\cP\in\cL(\cH^{\otimes N}\tightto\cH^{\otimes M})$ that optimizes the figure of merit.  
Notice that the optimization problem depends only on $N$, $M$, and $\eucal{S}$. Hereafter we shall suppress the 
dependence on $N$ and $M$ and simply write $\channelF$ to denote the optimized cloning fidelity.  Optimal quantum 
cloning maps have been constructed for the set of all pure states of dimension 
$d$~\cite{Buzek1996,Gisin1997,Werner1998}, as well as for the so-called phase covariant set of pure states for 
dimensions $d=2,3$~\cite{Bruss2000,DAriano2003} (the set of states that produce a uniform probability 
distribution when measured in the computational basis). 

Restricting the set of input states that the map has to clone is akin to relaxing the problem. Another 
relaxation of the cloning task is to allow for \emph{probabilistic} cloning maps. Such maps are known to produce 
higher fidelity copies than their deterministic counterparts, albeit at the cost of producing no clones with a 
finite probability~\cite{Duan1998, Pati1999, Chefles1999, Chiribella2013}. The more general case of cloning mixed 
states (also known as \emph{broadcasting}) is less studied, but nevertheless some results have been 
obtained~\cite{Barnum1996, DAriano2005, Chiribella2007}; for further reading one may consult reviews on 
cloning here~\cite{Scarani2005,Luo2010,Fan2014}.

Aside from cloning states, it is also of interest to consider whether it is possible to clone dynamical quantum 
resources such as unitary gates or quantum channels. Consider a family of quantum channels $\eucal{C}$  and a 
cloning \textit{process} which attempts to transform $N$ copies of $\cE\in \eucal{C}$ into $M>N$ approximate 
copies. This transformation is solicited by a process $\proc{P}:  \cL(\cH^{\otimes N}\tightto\cH^{\otimes 
N})\tightto\cL(\cH^{\otimes M}\tightto\cH^{\otimes M})$ that uses $N$ copies of $\cE$ to produce a CPTP map 
$\proc{P}[\cE^{\times N}]$ acting on $M$ quantum systems (see Fig.~\ref{fig:enter-label}). Recall that such 
processes can be causally or non-causally ordered and,  
regardless of the classification,
can be thought of as a quantum algorithm 
whose inputs and outputs consist of $M$ quantum systems, and which uses $N$ copies of the unknown channel as 
oracle queries.

Just as for state cloning, the optimal superchannel for cloning quantum channels depends strongly on the figure 
of merit one uses.  There is no clear cut notion of distance between two quantum channels with most measures 
defined indirectly via the effects such channels have on quantum states. Here
we adopt two measures; the first is the Choi-Jamio{\l}kowski fidelity which, for two 
channels $\cA,\cB\in\cL(\cH_1\tightto\cH_2)$, is defined~\cite{Raginsky2001}
 	\begin{equation}
    		F_{\rm CJ} (\cA, \cB) := F (\mathtt{CJ}[\cA], \mathtt{CJ}[\cB]),
	\label{eq:fidcj}
	\end{equation} 
where ${\mathtt{CJ}}[\cA] := ({\rm id} \otimes \cA) \left [\ketbra{\Phi^+}\right]$ with $\ket{\Phi^+}=
\nicefrac{1}{\sqrt{d}}\, \sum_{i=0}^{d-1} \ket{ii}$. Observe that if either of the channels is unitary, 
\eqnref{eq:fidcj} corresponds to the \emph{entanglement} (or process) fidelity~\cite{Schumacher1996,Chuang1997}, 
and that in this case there is a one-to-one relationship between the process fidelity and the \emph{average} 
fidelity~\cite{Horodecki1999, Nielsen2002}. The second figure of merit we shall use is the diamond (or worst-case) fidelity 
	\begin{equation}
		F_{\diamond} (\cA, \cB) := \min_{\Psi} F\big(({\rm id}\otimes \cA)[ \Psi],( {\rm id} \otimes \cB) 
        [\Psi]]\big),
	\label{eq: f diamond}
	\end{equation}
where the minimization is taken over all pure states $\Psi\in\cL(\cH_{1'}\otimes\cH_{1})$ of the extended system. 
By definition it holds that $\fcj (\cA, \cB)\geq F_{\diamond} (\cA, \cB)$. Analogous to the case of state 
cloning, given $(N, M)$ the optimal superchannel is the one that optimizes the chosen figure of merit  
	\begin{equation}
    		\channelF^{\eucal{C}}_\bullet(N,M) = \underset{\proc{P}}{\max}\, \underset{\cE \in\eucal{C}}
            {\inf} \, F_\bullet\big({\proc{P}}
		[\cE^{\times N}], \cE^{\otimes M} \big).
	\label{eq:cloneChan}
	\end{equation}
Note the immediate property $\channelF^{\eucal{C}}_{\rm CJ}\geq \channelF^{\eucal{C}}_{\diamond}$. From now on we 
will drop the explicit dependence on $N$ and $M$ from \eqnref{eq:cloneChan} to avoid notational clutter. 

With these definitions in place the following proposition establishes a connection between cloning quantum states 
and cloning of the corresponding \emph{trash-and-replace} channel. 
\begin{proposition}[Equivalence of cloning states and cloning trash-and-replace channels] 
\label{prop:cloning-trash}
Let $\eucal{S}\subset \cL(\cH)$ be a set of states  and $\eucal{C} =\{\mathcal{T}_\rho:\cL(\cH\tightto\cH)\, 
\vert \mathcal{T}_\rho[\,\cdot\,] 
= \rho,\, \rho \in \eucal{S}\}$ the set of corresponding trash-and-replace channels.  It holds that 
	\be
    		F^{\eucal{C}}_{\rm CJ} =F_{\diamond}^\eucal{C} = F^\eucal{S} \,,
	\label{eq:prop1}
	\ee    
for which any $N\tightto M$ cloning process $\cP:\cL(\cH^{\otimes N}\tightto\cH^{\otimes M})$ and cloning 
superchannel $\proc{P}:\cL(\cH^{\otimes N}\tightto\cH^{\otimes N})\tightto\cL(\cH^{\otimes M}\tightto\cH^{\otimes 
M})$ are related by $\proc{P}[\mathcal{T}_\rho ^{\times N}]=\mathcal{T}_{\cP[\rho^{\otimes N}]}$.
\end{proposition}

\begin{proofidea} The proof consists of two parts (see~\appref{app:states_chan_proof} for details). First, Lemma~\ref{lem: process dec} demonstrates that the three affine sets resulting from the contraction of a parallel, sequential or non-causally ordered process with any number of trash-and-replace channels are, in fact, identical and recover the full set of CPTP maps $\cE_{\proc{P}}[\,\cdot \otimes \rho^{\otimes N}]:=\proc{P}[\cT_\rho^{\times N}]
[\,\cdot\,]:\cL(\cH^{\otimes M}\tightto\cH^{\otimes M})$. Second, we show that for any process $\proc{P}$ 
that tries to clone $N$ trash-and-replace channels ($\cT_\rho^{\times N}$) there exists a corresponding state cloning map $\cP$.  
Conversely, for any state cloning map there exists a cloning process 
that attempts to clone the trash-and-replace 
channel. The proof is completed by showing that equivalence of all fidelities for the particular choice $\proc{P}
[\mathcal{T}_\rho ^{\times N}]=\mathcal{T}_{\cP[\rho^{\otimes N}]}$. We emphasize that the result holds for all processes: parallel, sequential and non-causally ordered and that the later two classes offer no advantage in any task performed on trash-and-
replace channels.
\end{proofidea}
\vspace{12pt}

A closely related task to cloning is \emph{replication}.  A replication protocol uses $N$ copies of a quantum 
resource to produce the maximum number $M>N$ copies with fidelity at least $1-\epsilon$, $\epsilon>0$. 
Formally, a replication process for quantum channels corresponds to the following optimization problem:
	\begin{align}\label{eq:super_replication}
 		\begin{split}
			\mathsf{M}^{\eucal{C}}_\bullet(N,\epsilon) = \, 
            \underset{\proc{P}}{\textrm{max}} \quad & M   \\ 
			\textrm{ s.t. } &\;
    			\underset{\cE \in\eucal{C}}{\textrm{inf}} \, 
                \channelF^{\eucal{C}}_\bullet \big(\proc{P}[\cE^{\times N}], 
                \cE^{\otimes M} \big) \geq 1-\epsilon\, .
		\end{split}
	\end{align}
Specifically, what is of interest in a replication process is the \emph{replication rate}, 
	\begin{equation}\label{eq:replication_rate}
		\mathsf{R}^{\eucal{C}}_\bullet(\epsilon)= \sup \, r \quad 
		\textrm{ s.t.} \quad \lim_{N\tightto \infty} 
        \frac{\mathsf{M}^{\eucal{C}}_\bullet(N,\epsilon)}{N^r}>0\, ,
	\end{equation}
which quantifies how the number of high-fidelity clones scales in the large $N$ limit. If, for any set of quantum 
resources, a replication protocol exists whose rate is $\mathsf{R}(\epsilon)=2,\,\forall\epsilon>0$ 
then we say that these resources can be \emph{super-replicated}.

For the set of pure states $\eucal{S}$, it follows from the expression of the optimal universal cloning 
map~\cite{Werner1998}, that for \emph{all states} the replication rate is at most linear, i.e., 
$\mathsf{R}(\epsilon)=1,\,\forall\epsilon$. In~\cite{Chiribella2013} this was shown to be the maximal rate for 
deterministic cloning of quantum states, and in section \secref{sec:nosuperrep_states} we will 
present an alternative derivation. In contrast, it is known that the set of phase covariant states can be 
\emph{probabilistically} replicated with a rate $\mathsf{R}(\epsilon)=2, \,\forall\epsilon$ with the probability 
of successful replication decreasing exponentially with the number of replicated copies 
$M$~\cite{Chiribella2013}.  Quite remarkably, the set of unitary gates forming a representation of the 
unitary groups $\mathrm{U}(1)$ and $\mathrm{SU}(2)$ can be deterministically replicated at a rate 
$\mathsf{R}(\epsilon)=2,\,\forall\epsilon>0$~\cite{Dur2015, Chiribella2015}.

Here we consider deterministic cloning and replication of a general set of quantum channels beyond the case of 
trash-and-replace and unitary channels. In particular, for continuously parametrized sets of channels we study 
the asymptotic replication rate, and ask in what instances super-replication is possible.  Note that whilst the 
fidelity is as good a quantifier of performance as any other, it is not a distance  as it fails to satisfy the 
triangle inequality. However, a proper distance measure can be obtained from the fidelity quite simply as 
	\be
    		D_{\bullet} := \arccos F_\bullet \, \in \, [0,\nicefrac{\pi}
            {2}]\, .
	\label{eq:distance}
	\ee
In particular, for states $D(\rho,\sigma)=\arccos F(\rho,\sigma)$ (Eq.~\eqref{eq:fidhalf}) is the Bures (or 
quantum) angle~\cite{bures1969extension,uhlmann1976transition}. We stress that the distance $D_{\diamond}
(\cA,\cB)$ between two channels as defined above is not to be confused with the \emph{diamond distance} 
\be\label{eq: dd}
d_\diamond(\cA,\cB):=\frac{1}{2}\|\cA-\cB\|_\diamond
\ee
induced by the diamond norm, although, the two are related by the Fuchs-van der Graaf 
inequality~\cite{Fuchs1999}, $1-\cos D_\diamond\leq  d_{\diamond}\leq \sin D_\diamond$.  Equivalently to 
\eqnref{eq:cloneChan}, we denote the optimized (minimal) cloning distance as $\channelD^{\eucal{C}}_{\bullet}$ or 
$\mathsf{d}^{\eucal{C}}_{\diamond}$.

\section{\label{sec:upper_bound}Upper bounds on cloning and replication rates}

Trying to solve the optimization problem in~\eqnref{eq:cloneChan} directly is, in general, not an easy task (see 
Section ~\ref{sec: clon proc}). In this section we derive upper bounds on the channel cloning rates and 
investigate when super-replication is possible by exploring a connection to two closely related tasks: quantum 
channel discrimination (\secref{sec:channel_discrimination}) and quantum metrology 
(\secref{sec:metrology}). In \secref{sec:nosuperrep_states} we use these bounds to prove that deterministic super-
replication of any continuous set of quantum states is impossible, whilst in \secref{sec:no_superrep_unitaries} 
we show that super-replication of the full set of quantum gates becomes impossible in the presence of any
noise, however weak. 

\subsection{\label{sec:channel_discrimination}Bounds from binary channel discrimination}

Let $\cE_0,\, \cE_1\in\eucal{C}$ be two channels that we wish to replicate and $\proc{P}:\cL(\cH^{\otimes 
N}\tightto\cH^{\otimes N})\tightto\cL(\cH^{\otimes M}\tightto\cH^{\otimes M})$ be the optimal cloning 
superchannel for $\eucal{C}$. Since $D_{\bullet}(\proc{P}[\cE_i^{\times N}], \cE_i^{\otimes M})\leq 
\channelD_{\bullet}^{\eucal{C}}$ must hold by definition, using the triangle inequality
one immediately obtains
the following \emph{geometric} lower bound on the optimal cloning distance $\channelD_\bullet^{\eucal{C}}$  
	\be
		\channelD^{\eucal{C}}_\bullet\geq \frac{1}{2}\left(D_{\bullet}(\cE_0^{\otimes M}, \cE_1^{\otimes M}) 
        - D_\bullet\left(\proc{P}[\cE_0^{\times N}], \proc{P}
		[\cE_1^{\times N}]\right)\right).
	\label{eq:triangle_ineq}
	\ee
Observe that the second term in \eqnref{eq:triangle_ineq} can be upper bounded by the {\it optimal channel 
discrimination distance}
	\begin{align}\nonumber
		 D_\bullet\left(\proc{P}[\cE_0^{\times N}], \proc{P}[\cE_1^{\times N}]\right) &\leq \sup_{\proc{S}} 
         D(\proc{S}[\cE_0^{\times N}], \proc{S} [\cE_1^{\times N}]) \\
        &:=\channelD^{(N)}\left(\cE_0,\cE_1\right)\, , 
	\label{eq:optimized_discrimination_distance}
	\end{align}
where $\proc{S}:\cL(\cH^{\otimes N}\tightto\cH^{\otimes N})\tightto\cL(\cH_{\rm out})$ is any state 
producing-process that maps $N$ copies of the channel $\cE_i$ to a quantum state. Indeed, it follows that 
$D_{\rm CJ}\left(\proc{P}[\cE_0^{\times N}], \proc{P}[\cE_1^{\times N}]\right)\leq D_{\diamond}\left(\proc{P}
[\cE_0^{\times N}], \proc{P}[\cE_1^{\times N}]\right)=D\left(\proc{S}[\cE_0^{\times N}], \proc{S}[\cE_1^{\times 
N}]\right)$, where the last equality holds when choosing the state producing process $\proc{S}[\cE_i^{\times N}]:=
({\rm id}\otimes \proc{P}[\cE_i^{\times N}])[ \Psi]$ (see Eq.~\eqref{eq: f diamond}).

The maximal distance $\channelD^{(N)}\left(\cE_0,\cE_1\right)$ has a clear operational meaning.  
It quantifies the binary channel discrimination task, albeit with the customary diamond distance\footnote{Note 
that in the single copy case the optimization over state producing process $\proc{S}$ reduces to applying the 
channel onto half of a bipartite state. Then one recovers the standard definitions
$ \mathsf{d}_{\diamond}^{(1)}(\cE_0,\cE_1) = d_{\diamond}(\cE_0,\cE_1)=\frac{1}{2}\sup_\Psi \|
(\text{id}\otimes(\cE_0-\cE_1))[\Psi]\|$ and $\mathsf{D}_{\diamond}^{(1)}(\cE_0,\cE_1) =   D_{\diamond}
(\cE_0,\cE_1)=\sup_\Psi  D\big(({\rm id}\otimes \cE_0)[ \Psi],( {\rm id} \otimes \cE_1) [\Psi]\big)$,
where the later is equivalent to Eq.~\eqref{eq: f diamond}. Furthermore, note that 
$\channelD^{(N)}\left(\cE_0,\cE_1\right)\geq D_\diamond(\cE_0^{\otimes N},\cE_1^{\otimes N}) = 
\sup_{\proc{S}^{\rm par}} D(\proc{S}^{\rm par}[\cE_0^{\times N}], \proc{S}^{\rm par} [\cE_1^{\times N}])$ where 
the optimization is taken over {\it parallel} state-producing processes. Hence, in general the quantities 
$D(\proc{S}[\cE_0^{\times N}], \proc{S}[\cE_1^{\times N}])$ and $D_\diamond(\cE_0^{\otimes N},\cE_1^{\otimes N})$ 
are incomparable.}
\be\label{eq: dist diamon norm}
 \mathsf{d}_{\diamond}^{(N)}(\cE_0,\cE_1):=\frac{1}{2}\, \sup_{\proc{S}}
\Vert\proc{S}[\cE_0^{\times N}]-\proc{S}[\cE_1^{\times N}]\Vert,
\ee
replaced by the Bures angle of Eq.~\eqref{eq:distance}. Moreover, by the Fuchs-van der Graaf 
inequality~\cite{Fuchs1999} the two are related as
	\be
		1-\cos\channelD^{(N)}\left(\cE_0,\cE_1\right)\leq  \mathsf{d}^{(N)}_{\diamond}(\cE_0,\cE_1)\leq \sin \channelD^{(N)}(\cE_0,\cE_1)\,. 
	\label{eq:Fuchs_vanderGraaf}
	\ee

Combining Eqs.~\eqref{eq:triangle_ineq} and \eqref{eq:optimized_discrimination_distance} gives the following 
proposition concerning the optimal cloning of a set of channels. 
\begin{proposition}[Optimal distance for cloning and replication]
\label{prop:cloning_bound_discrimination}
Consider the optimal $N\tightto M$ cloning of a set of channels $\eucal{C}$.  The optimal cloning distance as 
defined by \eqnref{eq:distance} satisfies 
	\be
		\channelD^{\eucal{C}}_\bullet\geq \frac{1}{2}\left(D_{\bullet}(\cE_0^{\otimes M}, \cE_1^{\otimes M}) 
        - \channelD^{(N)}\left(\cE_0,\cE_1\right)\right),
	\label{eq:optimal_cloning_distance}
	\ee
for all $\cE_0,\cE_1\in\eucal{C}$, where $\channelD^{(N)}\left(\cE_0,\cE_1\right)$ is the optimal $N$-copy 
channel discrimination distance for the channels $\cE_0$ and $ \cE_1$ defined in 
Eq.~\eqref{eq:optimized_discrimination_distance}.
\end{proposition}

Proposition~\ref{prop:cloning_bound_discrimination} can be used to provide insightful bounds on the replication 
rates for states and unitaries.  The following two corollaries, whose proofs can be found in 
\appref{app:states_and_Us}, \yg{we} show that deterministic super-replication of states, as well as of unitary operators 
(under the diamond fidelity figure of merit), at any rate larger than one is impossible.

\subsubsection{\label{sec:nosuperrep_states}No deterministic super-linear replication of states.}

\begin{corollary}[Strong converse on the replication rate for states]
\label{cor:state_rep}
Let $\eucal{S}\subset\cL(H)$ be any continuous set of states.  For any $\epsilon>0$, and $\mathsf{R}^{\eucal{S}}
(\epsilon)=1+\delta, \, \delta>0$ the optimal replication distance satisfies $\mathsf{D}^{\eucal{S}}\geq 
\nicefrac{\pi}{4}$, and the corresponding optimal replication fidelity is $\channelF^{\eucal S}\leq \frac{1}
{\sqrt{2}}$. In addition, for linear replication rate $M=(1+\lambda)N$, the asymptotic replication distance is 
bounded by 
\begin{align}\label{ref: bound state lin}
\mathsf{D}^{\eucal{S}} \geq \tilde{A}\left(\frac{1}{1+\lambda}\right)
&:= \max_{0\leq x\leq 1} \frac{\arccos x^{1+\lambda} - \arccos x}{2}, 
\end{align}
depicted in Fig.~\ref{fig:A func}.
\end{corollary}

\begin{proofidea} The Bures angle between two states $D(\rho,\sigma)$ is monotonic, i.e., cannot be increased by 
applying a common channel on the states. With this observation one can see that for state cloning the equivalent 
of Proposition~\ref{prop:cloning_bound_discrimination} simply reads 
\begin{align}\label{eq: prop 2 states}
\mathsf{D}^{\eucal{S}} \geq \frac{1}{2}\left(D(\rho_0^{\otimes M},\rho_1^{\otimes M}) - D(\rho_0^{\otimes 
N},\rho_1^{\otimes N})\right)\,.
\end{align}
To obtain the corollary for $M=N^{1+\delta}$ with $\delta>0$ it is then sufficient to chose states $\rho_0$ and 
$\rho_1$ such that in the large $N$ limit $D(\rho_0^{\otimes M},\rho_1^{\otimes M})\to \nicefrac{\pi}{2}$ while 
$D(\rho_0^{\otimes N},\rho_1^{\otimes N})\to 0$. This can always be done for a continuous set $\eucal{S}$.
In the case of asymptotic linear replication $M=(1+\lambda) N$, the two states can be chosen such that 
$x=F(\rho_0^{\otimes N},\rho_1^{\otimes N})=F(\rho_0,\rho_1)^N$ assumes any value in the interval $[0,1]$. 
Formally maximizing the right hand side of Eq.~\eqref{eq: prop 2 states} with respect to $x$ completes the proof.
\end{proofidea}
\vspace{12pt}

We note that the impossibility of super-linear replication of quantum states was already established 
in~\cite{Chiribella2013}. Corollary~\ref{cor:state_rep} serves as an illustrative example of the ideas for the 
simple case of states, and also gives a quantitative bound, \eqnref{ref: bound state lin}, on the cloning 
distance in the linear regime.

\begin{figure}
    \centering
    \includegraphics[width=0.95\linewidth]{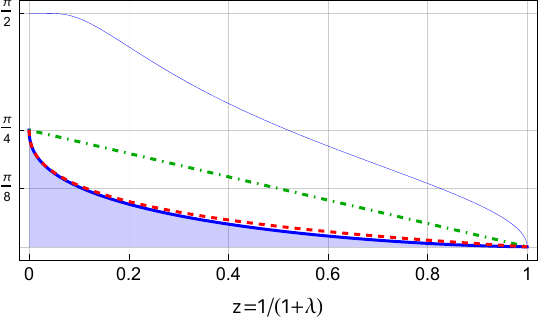}
    \caption{Bounds on the optimal replication distances $\mathsf{D}$ as functions of $z=\frac{1}{1+\lambda}$. 
    The thick blue line depicts the lower bound $A(z)$ on $\mathsf{D}^{\eucal{C}}_{\rm CJ}$ defined in 
    Proposition \ref{prop:necessary_conds_superrep}, and appearing in Corollary~\ref{cor:classical noise} and 
    Sections~\ref{sec:BF},\ref{sec:AD}. The red dashed line gives the lower bound $\tilde A(z)$ on 
    $\mathsf{D}^{\eucal{S}}$ defined in Corollary~\ref{cor:state_rep}. The green dash-dotted line is the simple 
    lower bound  $\mathsf{D}^{\eucal{C}}_\diamond \geq \frac{\pi}{4}\frac{\lambda}{1+\lambda}$ appearing in 
    Corollary~\ref{cor:unitaries_diamond}. Finally, the thin blue line is the asymptotic cloning distance of the 
    measure-and-prepare process for the set $\eucal{P}$ of Pauli-noise channels, discussed in Section~\ref{sec:BF}, 
    i.e. it is an upper bound on $\mathsf{D}^{\eucal{P}}_{\rm CJ}$.}
    \label{fig:A func}
\end{figure}

\subsubsection{\label{sec:no_superrep_unitaries}No deterministic super-linear replication of unitaries with 
respect to the diamond fidelity.}

\begin{corollary}[No super-replication of unitaries under the diamond fidelity]
\label{cor:unitaries_diamond}
Let $\eucal{C}=\lbrace \cU_\lambda\,\vert\, \lambda\in\mathbb{R}\rbrace$ be a continuous set of unitary channels 
and consider $N\tightto M$ cloning under a diamond fidelity figure of merit.  Then for $\epsilon<1-\frac{1}
{\sqrt{2}}$, $\mathsf{R}^{\eucal{C}}(\epsilon)=1$. In addition, for the linear rate $M=(1+\lambda)N$ the asymptotic 
optimal cloning distances are bounded by
\be
\mathsf{D}_\diamond^{\eucal{C}} \geq \frac{\pi}{4} \frac{\lambda}{1+\lambda}, \quad 
\mathsf{d}_\diamond^{\eucal{C}} \geq 1-\cos \left(\frac{\pi}{4} \frac{\lambda}{1+\lambda}\right).
\ee
\end{corollary}

\begin{proofidea} Optimal discrimination of unitaries has been well studied~\cite{DAriano2001a}. In particular, 
for close enough unitaries we know that 
\begin{align}
    \channelD^{(N)}\left(\cU_0,\cU_1\right)=  \mathsf{D}_{\diamond}(\cU_0^{\otimes N},\cU_1^{\otimes N})  = 
    \frac{N \,\Theta(U_0 U_1^\dag)}{2},
\end{align}
where $\Theta(V)$ is the maximal angular difference between the phases of the complex eigenvalues of $V$. 
Proposition~\ref{prop:cloning_bound_discrimination} then directly implies $\mathsf{D}_\diamond^{\eucal{C}} \geq 
\frac{ \Theta(U_0 U_1^\dag )}{4}\left( M-N\right)$. Finally, for a continuous set $\eucal{C}$ one can always take 
two unitaries such that 
\be\label{eq: bound unitary diamond}
\mathsf{D}_\diamond^{\eucal{C}} \geq \frac{\pi}{4}\left(1-\frac{N}{M}\right) =\frac{\pi}{4}\frac{\lambda}
{\lambda +1}
\ee
for $M=(1+\lambda)N$. By the Fuchs-van der Graaf inequality~\cite{Fuchs1999} this immediately translates to the 
bound $\mathsf{d}^{\eucal{C}}_{\diamond}\geq 1- \cos (\frac{\pi}{4}\left(1-\frac{N}{M}\right))$ for the optimal 
diamond distance.
\end{proofidea}
\vspace{12pt}

Note that Corollary~\ref{cor:unitaries_diamond} does not contradict the super-replication results 
of~\cite{Dur2015, Chiribella2015} as for the latter the figure of merit is the Choi-Jami\l kowski fidelity. 
However, it does display a sharp contrast between the two figures of merit. We will come back to it at the end of 
Section~\ref{sec: unitary cloning review}, when discussing the processes of~\cite{Dur2015, Chiribella2015} in 
more detail.

\subsection{\label{sec:metrology} Bounds from channel estimation}

Whilst the geometric bound in \eqnref{eq:optimal_cloning_distance} holds in general, its usefulness relies on our 
ability to compute the optimal channel discrimination distance $\mathsf{D}^{(N)}(\cE_0,\cE_1)$.  For states and 
unitary channels the optimal discrimination distance was easy to compute, however for more general 
families of channels this is no longer the case.  In this section, we use powerful techniques from the theory of 
quantum channel estimation~\cite{Kurdzialek2023} to provide a rigorous upper bound on the optimal discrimination 
distance for any family of continuously parametrized channels $\eucal{C}$.  We then provide necessary conditions 
for when a set of channels $\eucal{C}$ can be deterministically super-replicated and identify general families of 
channels for which super-replication is impossible. 

The key insight for establishing our general upper bound comes from the \emph{quantum Fisher information} 
(QFI)~\cite{Braunstein1994}.  Let $\eucal{S}=\{\rho_x\, \vert\, x\in\mathbb{R}\}$ be a smooth one-parameter 
family of quantum states.  The QFI is proportional to the square of the susceptibility of the distance (Bures 
angle) defined in \eqnref{eq:distance} 
	\be
	\qfi(\rho_x):=4  \left(\frac{D(\rho_x,\rho_{x+\d x})}{\d x}\right)^2 = 8 \frac{1-F(\rho_x,\rho_{x+\d x})}
    {\dd x^2}\, .
	\label{eq:QFI}
	\ee
Now consider a smooth curve $\{\cE_x:\cL(\cH\tightto\cH)\,\vert\, x\in[a,b]\subset\mathds{R}\}$ connecting two 
quantum channels $\cE_{a}$ with $\cE_{b}$, and let $\proc{S}$ be the state-producing process  that maximizes the 
channel discrimination distance $\mathsf{D}^{(N)}(\cE_{a},\cE_{b})=D\left(\proc{S}[\cE_{a}^{\times N}], \proc{S}
[\cE_{b}^{\times N}]\right)$. By the triangle inequality we have 
	\be
		\begin{split}
			\mathsf{D}^{(N)}(\cE_{a},\cE_{b})&\leq \int_{a}^{b}\, D\left(\proc{S}[\cE_x^{\times N}] , \proc{S}
            [\cE_{x+\d x}^{\times N}]\right)\dd x\\
			& = \frac{1}{2}\int_{a}^{b} \sqrt{\qfi\left(\proc{S}[\cE_x^{\times N}]\right)} \,\dd x \\
		&\leq \frac{1}{2}\int_{a}^{b}\sqrt{\max_{\proc{S}}\qfi\left(\proc{S}[\cE_x^{\times N}]\right)} \, \dd 
        x \\
			&:=\frac{1}{2}\int_{a}^{b}\sqrt{\mqfi^{(N)}\left(\cE_x\right)}\, \dd x 
		\end{split}
	\label{eq:triangle_metrology}
	\ee 
where the inequality in the third line follows from relaxing the assumption that $\proc{S}$ maximizes the 
discrimination distance between $\cE_{a}, \cE_{b}$ to one where $\proc{S}$ maximizes the discrimination distance 
between $\cE_x$ and $\cE_{x+\d x}$, and we have defined the maximized QFI in the last line.  The latter 
quantifies the precision with which one can estimate the localized parameter $x$ of a quantum channel, when 
having access to $N$ copies of it, via an estimation process (typically restricted to be 
causal)~\cite{Demkowicz-Dobrzanski2014,zhou2021asymptotic, Kurdzialek2023}. Denoting by $\{K_n(x):\cL(\cH\tightto\cH)\, \vert \, 
n\in\mathbb{Z}\}$ the set of Kraus operators of the channel $\cE_x$, i.e., $\cE_x[\,\cdot\,]=\sum_n K_n(x)\, 
\cdot\, K_n^\dagger(x)$, and defining the operators 
	\be
		\begin{split}
			\alpha(x)&:=\sum_n\, \dot K_n^\dag(x) \dot K_n(x)\\
			\beta(x)&:=\sum_n \dot K_n^\dag(x) K_n(x)\, , 
		\end{split}
	\label{eq:alphabetas} 
	\ee
with $\dot K_n(x):=\frac{\d K^\dagger_n(x)}{\d x}$, the maximized QFI under any causally ordered state-producing 
process $\proc{S}$ in \eqnref{eq:triangle_metrology} can be upper bounded as~\cite{Kurdzialek2023}
	\be
		\begin{split}
			\mqfi^{(N)}(\cE_x)&\leq 4 f_N(\cE_x),\quad \mathrm{where}\\
			f_N(\cE_x)&:=\min\, N\sqrt{\Vert\alpha(x)\Vert}\\
            & \qquad \times\,\left((N-1)\Vert\beta(x)\Vert +\sqrt{\Vert\alpha(x)\Vert}\right)\, ,
		\end{split}
	\label{eq:bound_channel_metro}
	\ee
where the minimization is over all Kraus representations $\{K_n(x):\cL(\cH\tightto\cH)\, \vert \, 
n\in\mathbb{Z}\}$ and their derivatives. We formalize the above analysis into the following proposition (see also \cite{uhlmann1992metric,taddei2013quantum,albarelli2022probe}).

\begin{proposition}[Bound on optimal channel discrimination from channel estimation] 
\label{prop:metrology_optimal_distance} Let $\cE_x$ with $x\in[a,b]\subset \mathds{R}$ be a smooth curve in the 
set of channels. For all causally ordered state-producing processes $\proc{S}$, the optimal $N$-copy channel 
discrimination distance in Eq.~\eqref{eq:optimized_discrimination_distance} satisfies 
\begin{equation}\label{eq: angle to QFI main}
\channelD^{(N)}\left(\cE_{a},\cE_{b}\right)\, \leq \int_{a}^{b}  \sqrt{{f}_N(\cE_x)}\, \dd x\, ,
\end{equation}
where $f_N(\cE_x)$ in given in \eqnref{eq:bound_channel_metro}. Consequently, the optimal diamond distance in
\eqnref{eq: dist diamon norm} satisfies 
 \begin{equation}
 \mathsf{d}_{\diamond}^{(N)}(\cE_{a},\cE_{b})\leq \sin \left( \int_{a}^{b}  \, \sqrt{{ f}_N(\cE_x)}\,\dd x\right)\, .
 \end{equation}
 \label{thm: channel disc}
\end{proposition}

Proposition~\ref{prop:metrology_optimal_distance} provides an easily computable upper bound on the $N$-copy 
channel discrimination task in terms of the single-copy quantity $f_N(\cE_x)$ and as such we believe that it is of 
independent interest beyond the cloning task considered here.   
The same proposition can also be used in conjunction with 
Proposition~\ref{prop:cloning_bound_discrimination} to imply the following bound on the optimal cloning distance 
for all causally ordered processes
\be
		\channelD^{\eucal{C}}_\bullet\geq \frac{1}{2}\left(D_{\bullet}(\cE_{a}^{\otimes M}, \cE_{b}^{\otimes M}) - \int_{a}^{b}\sqrt{f_N(\cE_x)}\, \d x\right)\, .
	\label{eq:optimal_distance_metro}
\ee
In turn, this expression can be used to establish necessary conditions for channel super-replication.  The 
following proposition, whose proof can be found in \appref{app:necessary_conds_superrep}, establishes precisely 
these conditions.  

\begin{proposition}\label{prop:necessary_conds_superrep}
Let $\eucal{C}$ be a continuous set of channels such that $\{\cE_{x}\, \vert\, x 
\in [0,\delta)\}\subset\eucal{C}$ forms a smooth curve\footnote{Regularity conditions are discussed 
\appref{app:necessary_conds_superrep}.} inside \(\eucal{C}\) for some \(\delta>0\). For \(\cE_x\) let \(\alpha(x),\,\beta(x)\) be given as in~\eqnref{eq:alphabetas}.
	\begin{enumerate}[label=(\roman*)]
		\item 
		If there exists a Kraus decomposition of \(\cE_x\) such that $\beta(x) = 0$, then the optimal 
        replication rate is linear
			\be
				\mathsf{R}_{\rm CJ}^{\eucal{C}}(\epsilon)=1 \quad \text{for}\quad \epsilon < 1-\frac{1}{\sqrt 
                2}\, .
			\ee 
		Furthermore, for $M\geq(1+\lambda)N$ the asymptotic optimal cloning  distance satisfies 
			\begin{align} \label{eq: rate linear}
				\mathsf{D}_{\rm{CJ}}^{\eucal{C}} \geq  A\left(\frac{4\, \Vert\alpha(x)\Vert}{(1+\lambda)\, 
                {\rm{QFI}}(\mathtt{CJ}[\cE_{a}]) }\right)
			\end{align}
		\item	
		If there does not exist a Kraus decomposition such that \(\beta(x)= 0 \), then for $M\geq (1+\lambda) 
        N^2$ the asymptotically optimal cloning distance satisfies
			\begin{align}
				\mathsf{D}_{\rm CJ}^{\eucal{C}} \geq A\left(\frac{4\, \sqrt{\Vert\alpha(x)\Vert} 
                \Vert\beta(x)\Vert}{(1+\lambda)\, {\rm{QFI}}(\mathtt{CJ}[\cE_{a}]) }\right).
			\end{align}
		\end{enumerate}
Here \(A(z)\in[0,\nicefrac{\pi}{4}]\), depicted in Fig.~\ref{fig:A func}, is the monotonically decreasing 
function 
\begin{equation}
A(z) :=  \begin{cases}
\frac{\arccos \zeta(z) - \sqrt{2 z \ln(\frac{1}{\zeta(z)})}}{2} & 0\leq z < 1\\
0 & z\geq 1
\end{cases}
\label{eq:A_function}
\end{equation}
of a positive real variable $z\geq 0$ with 
	\be
		\zeta(z) := \sqrt{\frac{-z}{W_{-1}( - z \exp(-z))}}\, ,
	\ee 
where \(W_{-1}(z)\) is the lower branch of the Lambert \(W\) function.	
\end{proposition}

It is worth noting that the condition $(i)$ ($\beta(x)$ can be set to zero) is equivalent to $\ii \beta(x) \in {\rm span}_{\mathds{H}}(K_i^\dag(x) K_j(x) \,\forall i,j)$~\cite{zhou2021asymptotic}, independent of the Kraus decomposition choice.

\begin{proofidea} The proof consists of expanding the bound in \eqnref{eq:optimal_distance_metro} in the large 
$N$ limit. In this limit, the bound is only meaningful if the channels $\cE_{a}$ and $\cE_{b}$ become 
infinitesimally close to a channel $\cE_x$, such that the bound is dominated by the {\it local} properties of 
the channel around $\cE_x$. More specifically, $D_{\rm CJ}(\cE_{a}^{\otimes M}, \cE_{b}^{\otimes M})$ is governed 
by ${{\rm{QFI}}(\mathtt{CJ}[\cE_{x}])}$ and $\int_{a}^{b}\sqrt{f_N(\cE_x)}\, \d x$ is governed by 
$\sqrt{f_N(\cE_x)}$. In ~\appref{app:necessary_conds_superrep} we work out the exact relation between these 
quantities.
\end{proofidea}
\vspace{12pt}
Note that Proposition~\ref{prop:necessary_conds_superrep} is not constructive: in 
order for super-replication to be possible it is necessary for \(\beta(x)\neq 0\),
but that the explicit construction of a process]
is not given. Despite this, the necessary condition for super-replication 
identifies a wide range of channels for which the phenomenon is not possible; the two corollaries that follow identify two specific classes: classical noise channels and the full set of noisy unitary gates.

\subsubsection{No super-linear replication of ``classical noise" channels}
\begin{corollary}[No super-replication for classical noise channels]
    \label{cor:classical noise} 
Let $\eucal{C}$ be a continuous set of channels such that $\{\cE_{x}\, \vert\, x \in 
[0,\delta)\}\subset\eucal{C}$ with 
\begin{align}
\cE_x[\,\cdot \,] &= \sum_{k} K_k(x) \,\cdot \,K_k^\dag(x)\\
 K_k (x)&=  a_k(x) G_k \qquad a_k(x)\in \mathds{R}, \,G_k \in L(\cH)\, ,
\end{align}
forms a smooth curve inside \(\eucal{C}\) for some $\delta>0$. This set cannot be super-replicated.
\end{corollary}

\begin{proof} It is straightforward to see that
\begin{align}
\beta(x) &= \sum_{k} \dot K_k^\dag(x)\, K_k(x)=\sum_{k} \dot a_k(x) a_k(x) G_k^\dag G_k \\
&=  \frac{1}{2}\frac{\dd}{ \dd x} \sum_k (a^2_k(x))\,  G_k^\dag G_k =  \frac{1}{2}\frac{\dd}{ \dd x} \id =0\, .
\end{align}
The proof is completed by using Proposition~\ref{prop:necessary_conds_superrep}.
\end{proof}

For channels with the additional properties $G_k^\dag G_k =\id$ and $\tr G_k G_\ell =0$ for $k\neq \ell$, such as 
Pauli noise channels one can obtain a simple bound for the replication error of any process that achieves a linear 
replication rate.  In this case $p_k(x):=a_k^2(x)$ defines a probability distribution and 
\be
\frac{\mqfi^{(N)}(\cE_{x})}{N} \leq \ 4\|\alpha(x)\| = \sum_k \frac{(\dot p_k(x))^2}{ p_k(x)}\, ,
\ee
where the last expression is the Fisher information of the distribution $p_k(x)$.  Moreover, as $\bra{\Phi^+} 
(G_k\otimes \id)(G_\ell\otimes \id) \ket{\Phi^+}=\delta_{k\ell}$ it follows that ${\rm QFI}({\mathtt{CJ}}
[\cE_x])= \sum_k \frac{(\dot p_k)^2}{ p_k}$ is given by the same expression. Hence, by \eqnref{eq: rate linear}, 
and with  $M = (1+\lambda)N$ it follows that the replication error is 
\be\label{eq: pauli noise}
 \mathsf{D}_{\rm CJ}^{\eucal{C}} \geq  A\left(\frac{1}{1+\lambda}\right)\,.
\ee

\subsubsection{\label{sec:unitary_robustness}Super-linear replication of the full set of unitaries has zero 
robustness.}
\begin{corollary}[No super-replication of the full set of noisy quantum gates]
     \label{cor:no robustness}  
Let $\cN[\,\cdot \,]=\sum_k L_k \cdot L_k^\dag$  be a non-unitary channel. The 
set of noisy unitary gates
\be
\eucal{C} =\{\cN\circ\cU |\, \cU\in{\rm SU}(d)\}
\ee
can not be super-replicated.
\end{corollary}

\begin{proofidea}
For any Hermitian operator $H$ not proportional to the identity we are free to chose the following curve 
\be
\{\cE_{x}\, \vert\, x \in [0,\delta)\}\subset\eucal{C} \quad \text{with} \quad \cE_x[\,\cdot\,] = \cN[ e^{\ii H 
x}\,\cdot\, e^{-\ii H x}]
\ee
inside our set. The channels $\cE_x$ admit the natural Kraus representation $K_k(x)= L_k e^{\ii H x}$. The proof 
consists of showing that for any $\cN$ it is possible to chose $H$ such that $\beta(x)=0$ in 
\eqnref{eq:alphabetas}, enforcing the no-super-replication condition (see \appref{app: zero robustness} for the 
detailed proof.)
\end{proofidea}
\vspace{12pt}

We stress that it is crucial that $\eucal{C}$ contains {\it all} noisy unitaries 
in order for the proof of Corollary~\ref{cor:no robustness} to go through.  As we 
will see in the next section there exist continuous, restricted sets of noisy unitary gates which \textit{can} be super-replicated.

Corollary~\ref{cor:no robustness} should be contrasted with the results of ~\cite{Chiribella2015} in which it was 
shown that the set ${\rm SU}(d)$ of all unitary gates in dimension $d$ can be super-replicated. 
Corollary~\ref{cor:no robustness} shows that such super-replication is not robust.  Indeed, consider a 
superchannel $\proc{P}$ that tries to replicate $N$ uses of a noisy quantum gate into $M$ copies of its noiseless 
counterpart, i.e, 
\be\label{eq: noisy superrep}
\lim_{N\to \infty} D_{\rm CJ}(\proc{P}[(\cN\circ\cU)^{\times N}],\cU^{\otimes M})\leq \epsilon\, .
\ee
By acting on all the $M$ output systems with the fixed noise channel $\cN$ one can define the process $\proc{P}'$ 
such that 
\begin{align}\nonumber
\label{eq:robustness_inequa}
 D_{\rm CJ}(\proc{P}'[(\cN\circ\cU)^{\times N}],(\cN\circ\cU)^{\otimes M})& \\
 \leq D_{\rm CJ}(\proc{P}[(\cN\circ\cU)^{\times N}],\cU^{\otimes M})&
\end{align}
which follows from the monotonicity of the Bures angle.  But by Corollary~\ref{cor:no robustness} the left-hand 
side of \eqnref{eq:robustness_inequa} cannot be made arbitrarily small and thus, super-replicating $N$ noisy 
gates into  $M>N$ perfect copies of gates is not possible.

\section{\label{sec: clon proc}Cloning processes}
In this section we study the performance of various cloning and replication processes. First we show that the 
task of finding the optimal cloning processes for a given family of channels, can be approximated with a 
semidefinite program (SDP) (Sec.~\ref{sec:sdp}). 
Although the SDP is the most general numerical formulation of the problem, it suffers from computational limitations that severely restrict the application of the approach beyond moderate values of 
$N$ and $M$.

In Sec.~\ref{sec:measureandprep} we embark on a more instructive understanding of optimal cloning for large $N, M$ values by honing in on a particular class of measure-and-prepare processes, which are fairly straightforward to analyze.  

We highlight the natural connection 
between
these processes and the Bayesian channel estimation task, and establish quantitative bounds relating optimal 
(measure-and-prepare) cloning fidelity with several score functions used in Bayesian channel estimation. With 
this in hand, in Secs.~\ref{sec:superrepsu2} and ~\ref{sec:superrepphase} we show that super-replication of all 
qubit unitaries $\eucal{C} \simeq \mathrm{SU}(2)$ and all qubit phase gates $\eucal{C}\simeq\mathrm{U}(1)$ can be 
realized with measure-and-prepare processes. This provides an alternative to the coherent super-replication 
process of Refs.~\cite{Dur2015,Chiribella2015} and shows that super-replication of these families of channels is 
a direct consequence of being able to estimate them with an error that decays quadratically with $N$.  Note that this shows the measure-and-prepare protocols give the optimal {\it scaling} of $M$ with $N$, but 
not the exact prefactor.  Away from asymptotics there are other known results:
for 
$1\tightto 2$ cloning of unitary gates, it was shown that coherent processes are strictly optimal, outperforming 
any processes which attempt to learn the action of the unitary~\cite{Chiribella2008a}.\\

Before discussing more involved cloning processes let us quickly introduce a naive, but pretty general cloning 
strategy that we call the {\it dummy} process. It consists of applying the given channels $\cE \in \eucal{C}$ on 
the first $N$ systems and applying a fixed dummy channel $\cE_{\rm dum}$ on the remaining $M-N$ systems. The 
optimal Choi-Jamio\l{}kowski cloning fidelity for the dummy processes is given by 
\be\label{eq:F dummy}
\mathsf{F}_{\rm CJ}^\eucal{C|\rm dum} := \left( \max_{\cE_{\rm dum}} \min_{\cE \in \eucal{C}} \fcj (\cE,
\cE_{\rm dum} )\right)^{M-N}\, .
\ee
Note that the minmax problem---equivalent to finding the best dummy channel---is 
independent of $N$ and $M$ (see the examples in Secs.~\ref{sec:BF} and \ref{sec:AD}). In particular, by choosing 
the trash-and-replace channel $\cE_{\rm dum} = \cT_{\nicefrac{\id}{d}}$ and noting that $F_{\rm CJ}
(\cT_{\nicefrac{\id}{d}},\cE)= \frac{1}{d} \tr \sqrt{\mathtt{CJ}[\cE]}\geq \nicefrac{1}{d}$ is saturated by 
unitaries $\cE=\cU$, we obtain the following expression
\be\label{eq: dummy bound}
\mathsf{F}_{\rm CJ}^\eucal{C|\rm dum}\geq \mathsf{F}_{\rm CJ}^\eucal{\rm CPTP|\rm dum} =(\nicefrac{1}{d})^{M-N}
\ee
as a baseline for future comparisons. Since the cloning fidelity decays exponentially, the dummy process is only 
good when $M-N =\mathcal{O}(1)$. Indeed, even in the linear regime $M=(1+\lambda)N$, its asymptotic replication 
rate is zero for any $\lambda>0$.

\subsection{\label{sec:sdp} Optimal processes via Semidefinite Programming}

Numerically solving the worst case cloning fidelity optimization in Eq.~\eqref{eq:cloneChan} presents several challenges. $(i)$~Formulating the domain constraining the variable $\proc{P}$ is \textit{a priori} not easy. $(ii)$~The figure of merit $F_{\rm CJ}\big({\proc{P}}[\cE^{\times N}], \cE^{\otimes M} \big)$ is nonlinear. $(iii)$~It is a minmax problem on a continuous set.
We now discuss how these challenges can be overcome in order to approximate $\mathsf{F}_{\rm CJ}^\eucal{C}$ numerically. We will keep the discussion at a conceptual level with all technical details deferred to \appref{app:sdp}.

$(i)$ The sets of all (parallel, sequential or non-causal) processes ${\proc{P}}$  can be represented \cite{Milz2024characterising} by a matrix variable subject to positive semidefinite constraints, such that $\mathtt{CJ}[{\proc{P}}[\cE^{\times N}]]$ is linear in the variable. 

$(ii)$ 
The fidelity between any two density matrices $\rho$ and $\sigma$ can be expressed as an SDP~\cite{Skrzypczyk2023} by introducing a slack variable $Y$
\begin{align} 
  F(\rho,\sigma)
  =\; \max_Y &\quad \tr \tfrac{1}{2}(Y+Y^\dag)\\ \nonumber
     \text{subject to}&\quad
		\begin{pmatrix}
					 \rho& Y\\
					Y^\dag & \sigma
				\end{pmatrix} \ge 0.
\end{align}

$(iii)$ To approximate $\mathsf{F}_{\rm CJ}^\eucal{C}$ we can choose a discrete net $\eucal{N}:=\{\cE_i\}_{i=1}^H$ 
inside the continuous set $\eucal{N} \subset \eucal C$, such that $\channelF^{\eucal{N}}_{\rm CJ}\approx 
\channelF^{\eucal{C}}_{\rm CJ}$ approximates the optimal cloning fidelity from above. In turn, a cloning process 
$\proc{P}$ achieves $\channelF^{\eucal{N}}_{\rm CJ}\geq x$ if it fulfills the fidelity constraints 
$F_{\rm}\big({\proc{P}}[\cE_i^{\times N}], \cE_i^{\otimes M} \big)\geq x\;\forall\,i\in\{1,\ldots,H\}$, which 
changes the optimization problem above to a \textit{feasibility} problem. 
Following the above discussion, such a process can be found by the SDP
\begin{align}  \label{eq:pseudoopt}
  \textrm{ find} &\quad \proc{P},\{Y_i\}_{i=1}^H\\  
  \nonumber
     \textrm{subject  to}&\quad
		{\rm tr} \tfrac{1}{2}(Y_i+Y_i^\dag) \geq  x \\  
        \nonumber
        & \quad \left(\begin{array}{cc}
					\mathtt{CJ}[{\proc{P}}[\cE_i^{\times N}]] & Y_i\\
					Y_i^\dag & \mathtt{CJ}[\cE_i^{\otimes M}]
				\end{array}\right) \geq  0.\\  
                \nonumber
        &\quad \textrm{ process constraints  on  }  \proc{P} 
\end{align}
where we introduced $H$ slack variables $Y_i$. Conversely, showing that the program~\eqref{eq:pseudoopt} is infeasible implies $\mathsf{F}_{\rm CJ}^\eucal{C}\leq \mathsf{F}_{\rm CJ}^\eucal{N}<x$ for the considered class of processes. It is also worth mentioning that once the optimal cloning process $\proc{P}$ for $\eucal{N}$ is found, one can easily improve the approximation of the worst case fidelity it achieves on the whole set $\eucal{C}$, e.g. by increasing the net size.

Finally, we note that the SDP approach can also be used to bound the best cloning fidelity achievable by the measure-and-prepare processes introduced in Section~\ref{sec:measureandprep}. This is done by imposing a positive partial transpose constraint on the matrix variable representing the process, as explained in \appref{app:sdp}.

\subsubsection{Numerical results for 1 to 2}
\label{sec:numerics}
\begin{table}[t!]
 \centering
\begin{tabular}{|c|c|c|c|c|c|c|c|}
\hline
\textbf{Channel}&\textbf{P} & \textbf{Interval} & $H$ & {$x$} & \textbf{Soln} & $\mathsf{F}_{\rm CJ}$ & ${\rm avg} \, F_{\rm CJ}$ \\ 
\hline
\rule{0pt}{2ex}
\multirow{6}{*}{$\cE_\gamma$} & \multirow{4}{*}{All} & [0,1] & 2 & 1.0 & ${\checkmark}$ & 0.904 &0.935
\\ \cline{3-8}
\rule{0pt}{2ex}
& &\multirow{3}{*}{[0.01, 0.96]} & \multirow{3}{*}{21} & 0.92 & $\checkmark$ & 0.926 & 0.954
\\ \cline{5-8} 
\rule{0pt}{2ex}
&  & &  & 0.93 & ${\color{ForestGreen} \checkmark}$ & {\color{ForestGreen} 0.932} & 0.97
\\ \cline{5-8} 
\rule{0pt}{2ex}
&  & &  & 0.94 & {\color{red} \xmark} & -- & -- 
\\
\cline{2-8}
\rule{0pt}{2ex}
&\multirow{2}{*}{M\&P} & \multirow{2}{*}{[0.05, 0.966]} & \multirow{2}{*}{21} & 0.92 & ${\color{ForestGreen} \checkmark}$& { \color{ForestGreen} 0.925} &0.945
\\ \cline{5-8}
\rule{0pt}{2ex}
 & & & & 0.93 &  {\color{red} \xmark} & -- & -- 
\\ \hline
\rule{0pt}{2ex}
\multirow{2}{*}{$\cX_p$} & \multirow{2}{*}{All}  & \multirow{2}{*}{[0,1]} & \multirow{2}{*}{21} & 0.92 & ${\color{ForestGreen} \checkmark}$ & {\color{ForestGreen} 0.92} &0.947
\\  \cline{5-8}
\rule{0pt}{2ex}
  &  & &  & 0.93 &  {\color{red} \xmark} & -- & -- 
\\ \hline
\end{tabular}
\caption{The results of running the feasibility SDP in \eqref{eq:sdpfeas} for the qubit amplitude-damping channel $\cE_\gamma$ in \eqnref{eq: AD Kraus} and the bit-flip channel $\cX_p$ in \eqnref{eq: bit-flip} for 
$1\rightarrow 2$ cloning. 
For the former, we run the program to find the most general amplitude damping cloning process (All), as well as one which is restricted by the PPT criterion to bound the measure-and-prepare processes (M\&P). 
For the net $\cN$, $H$ points were chosen from the indicated interval and a minimal 
fidelity of $x$ was set for the feasibility. In some cases we were able to obtain better results (on the entire 
domain) when restricting the interval. The average fidelity for the amplitude damping (respectively bit-flip) 
channel is taken over $101$ points uniformly distributed in $[0, 1]$ according to ${\rm avg}\,F_{\rm CJ}  = 
\tfrac{1}{|\eucal{N}|}\sum_i{F}_{\rm CJ}(\proc{P}^*(\cE_i), \cE_i^{\otimes 2})$, where $\proc{P}^*$ is the 
optimal process found by the computer. 
}
\label{tab:num}
\end{table}

In Table~\ref{tab:num} we present the numerics for $1\rightarrow 2$ cloning for the amplitude-damping channel 
$\cE_{\gamma}$ (Eq.~\eqref{eq: AD Kraus}) as well as the bit-flip channel $\cX_p$ (\eqnref{eq: bit-flip}) 
performed in {\tt Matlab} using the {\tt SeDuMi} solver. The exact formulation of the programs can be found in 
App.~\ref{app:sdp}, as well as plots of the interesting cases. Note that for $1\rightarrow 2$ cloning the linear 
constraints for sequential as well as non-causally ordered processes collapse trivially to the parallel case.
For the amplitude-damping channel we find that when the net is restricted to the two extremal points 
$\eucal{N} =\{\cE_0= \rm id, \cE_1 = \cT_{\ketbra{0}{0}}\}$, that the computer is able to recover the optimal 
process that achieves perfect fidelity $1$. This is exactly as expected, as the identity ($\rm id$) and the trash-and-replace ($\cT_{\ketbra{0}{0}}$) channels can be perfectly distinguished by probing them with the state 
$\ket{1}$. 

By going up to $H=21$, we find the approximation $\mathsf{F}_{\rm CJ}^{\eucal{D}}(1, 2) \approx 0.932$ 
for the 
optimal cloning fidelity of AD channels. This value is computed by fixing the optimal process $\proc{P}^*$ found 
by the feasibility program, and empirically minimizing
\be
{\rm F}_{\rm CJ}^{\eucal{D}|\rm SDP}(1, 2)=\min_{0\leq \gamma\leq 1} F_{\rm CJ}(\proc{P}^*
[\cE_\gamma],\cE_\gamma^{\otimes 2})\, .
\ee
Somewhat surprisingly, we saw that a better process is found by the SDP if channels around the extremal points of 
the domain ($\cE_0$ and $\cE_1$) are not included in the net $\eucal{N}$. This seems to do with the rank change of the \ChJa states at the boundaries.

By restricting the SDP further using the PPT criterion (one additional SDP constraint on the process), we are able to bound the performance of all measure-and-prepare cloning processes. These perform slightly worse than the most general process $\mathsf{F}_{\rm CJ}^{\eucal{D}|\text{M\&P}}(1, 2) \approx 0.925$ 
and are in exact agreement with the analytic bound derived in Section~\ref{ad: compare}. Moreover, the PPT constrained program didn't find a feasible process at $x=0.93$. This shows that for $1\tightto2$ cloning of amplitude-damping channels coherent processes give an advantage over measure-and-prepare ones, complementing the same result for unitary gates established in~\cite{Chiribella2008a}.

\subsection{\label{sec:measureandprep}Measure-and-prepare cloning}

\textit{Measure-and-prepare} processes 
consist of two steps: (i) a measurement process $\{\proc{E}_\ell\}$ that consumes $N$ copies of the input
channel to produce a classical outcome $\ell$ with probability $\proc{E}_\ell[\cE^{\times N}]$, and (ii) a global channel preparation $\widehat \cE_\ell^{(M)} \in \cL(\cH^{\otimes M}\tightto 
\cH^{\otimes M})$ conditioned on 
the value $\ell$. These processes have the general form $\proc{P}_{\rm M\& P}[\cE^{\times 
N}]=\sum_{\ell}  \proc{E}_\ell[\cE^{\times N}] \widehat{\cE}_\ell^{(M)}$ and the associated  optimal cloning fidelity 
\be\label{eq: optimal MP}
\mathsf{F}_{\rm CJ}^\eucal{C|\rm M\&P} :=  \max_{\{\proc{E}_{\ell}, \widehat \cE_\ell \}} \min_{\cE \in 
\eucal{C}} \fcj \left(\proc{P}_{\rm M\& P}[\cE^{\times N}],\cE^{\otimes M} \right)\, .
\ee

While this is the most general form of a measure-and-prepare process,
it is very natural to consider a more 
restricted set of {\it estimate-and-prepare} processes, which first {\it estimate} the channel, using a
process $\{\proc{E}_{\hat \cE}\}$, outputting an estimator $\hat \cE$ (a classical description of the 
corresponding channel) with corresponding probability distribution ${\rm Pr}^{(N)}(\hat \cE|\cE) = 
\proc{E}_{\hat \cE}[\cE^{\times N}]$, and then prepare $M$ copies of the estimated channel
\be
\label{eq:meawsure&prepare}
\proc{P}_{\rm E\& P}[\cE^{\times N}] = \sum_{\hat \cE} {\rm Pr}^{(N)}(\hat \cE|\cE)  \, \hat{\cE}^{\otimes M}\,.
\ee
We will focus on such processes in the rest of this section.
Using the  concavity of the fidelity  we can write 
$\fcj \left(\proc{P}_{\rm E\& P}[\cE^{\times N}],\cE^{\otimes M} \right) \geq \sum_{\hat \cE} 
{\rm Pr}^{(N)}(\hat \cE|\cE) (\fcj(\hat\cE,\cE))^M,$ from which it follows that the optimal measure-and-prepare 
cloning fidelity satisfies
\begin{align}
\label{eq:AF}
\mathsf{F}_{\rm CJ}^\eucal{C|\rm M\&P} \geq
 \max_{\{\proc{E}_{\hat \cE}\}} \min_{\cE \in \eucal{C}} \sum_{\hat \cE} 
{\rm Pr}^{(N)}(\hat \cE|\cE) \, (\fcj(\hat\cE,\cE))^M.
\end{align}
The optimization on the right-hand side defines the optimal worst case 
estimation score for the $M$-dependent score function $(\fcj(\hat\cE,\cE))^M$, 
and is an instance of the channel estimation task studied in Bayesian quantum 
metrology.  

To push the connection with Bayesian estimation further using Jensen's inequality for the expected values 
$\mathds{E}[F_{\rm CJ}^M]\geq\mathds{E}[F_{\rm CJ}^2]^{\nicefrac{M}{2}}\geq \mathds{E}[F_{\rm CJ}]^M$, gives the inequalities
\begin{align}
\mathsf{F}_{\rm CJ}^\eucal{C|\rm M\&P} &\geq \left(
\max_{\{\proc{E}_{\hat \cE}\}} \min_{\cE \in \eucal{C}}\sum_{\hat \cE} 
{\rm Pr}^{(N)}(\hat \cE|\cE) (\fcj(\hat\cE,\cE))^2\right)^{\nicefrac{M}{2}} \label{eq: Jensen2}\\
&\geq     
\left( \max_{\{\proc{E}_{\hat \cE}\}} \min_{\cE \in \eucal{C}}\sum_{\hat \cE} 
{\rm Pr}^{(N)}(\hat \cE|\cE) \fcj(\hat\cE,\cE)\right)^M.
\label{eq:Jensen}
\end{align}
Both score functions $\fcj(\hat\cE,\cE)^2$ and $\fcj(\hat\cE,\cE)$ are commonly used in the literature on channel 
estimation, (see the next sections for examples). Furthermore, let the channel family $\eucal{C}$ be smoothly 
parameterized by a vector of real parameters $\bm \theta \in \cM_\eucal{C}\subset \mathds{R}^n$ from a manifold 
$ \cM_\eucal{C}$,  such that 
\begin{equation}
   F_{\rm CJ}(\cE_{\bm \theta_1},\cE_{\bm \theta_2}) \geq  1- g \|\bm \theta_1-\bm \theta_2\|^2, \quad g>0\, ,
\end{equation}
where $\|\bm \theta_1-\bm \theta_2\|$ is the Euclidean norm.  Using \eqnref{eq:Jensen} with $(1-x^2)^M \geq 1- M 
x^2$ we then obtain 
\begin{align}\label{eq: MSE}\nonumber
&\mathsf{F}_{\rm CJ}^\eucal{C|\rm M\&P} \geq1- g M \,\mathsf{MSE}^\eucal{C} \qquad \text{with}\\
&\mathsf{MSE}^\eucal{C}: = \max_{\{\proc{E}_{\hat \cE}\}} \min_{\cE \in \eucal{C}} \sum_{\hat \cE} 
{\rm Pr}^{(N)}(\hat \cE|\cE) \, \|\hat {\bm \theta} - \bm \theta\|^2\, ,
\end{align}
where we identified $\hat \cE = \cE_{\hat{\bm \theta}}$. Here, $\mathsf{MSE}^\eucal{C}$ is the worst case mean 
squared error, with respect to the parametrization of the channels we introduced.

The inequalities above establish a relation between channel estimation and cloning, dual to the QFI-based 
lower bounds derived in Sec.~\ref{sec:metrology}. Here, we have shown that channel estimation protocols can 
be directly used to define measure-and-prepare cloning processes, with readily available relations of Eqs.~
(\ref{eq:AF}-\ref{eq:Jensen},\ref{eq: MSE})  between the natural figures of merit.

This connection is particularly interesting when considering the asymptotic replication rate. For instance, 
assume that for a family of channels there is a known estimation protocol achieving super-linear mean squared 
error
\be
\mathsf{MSE}^\eucal{C}=  \frac{\gamma}{N^{1+\epsilon}}\, .
\ee
Then this protocol directly induces a measure-and-prepare cloning process with a worst case fidelity 
exceeding
\be\label{eq: F to MSE}
\mathsf{F}^{\eucal{C}|\rm M\&P}_{\rm CJ} \geq 1 - g\, \gamma \frac{M}{N^{1+\epsilon}}\, ,
\ee
i.e., a super-linear replication rate. In the following two sections we show that the super-replication results 
for all qubit gates and qubit phase gates, established in \cite{Chiribella2015, Dur2015} with coherent processes 
(summarized in Sec.~\ref{sec: unitary cloning review}), can also be realized with measure-and-prepare processes.

\subsubsection{\label{sec:superrepsu2}Measure-and-prepare super-replication of all qubit unitries.}

\newcommand{\Uset}{\eucal{U}_2}

The task of estimating an unknown qubit unitary
\be
\Uset =\{\cU[\cdot]:=U \cdot U^{\dag}| U\in {\rm SU}(2)\} 
\ee
has been studied extensively~\cite{hayashi2006parallel, Chiribella2005}.  The optimal strategy probes $N$ copies 
of the gate $\cU$ in parallel using $N$ spin-$1/2$ systems prepared in the state
    \begin{align} \nonumber
        \ket{\psi}&=\beta_{\frac{N}{2}}\ket{\frac{N}{2},\frac{N}{2}}\\
        &+\sum_{j=0(1/2)}^{\frac{N}{2}-1}\frac{\beta_j}{\sqrt{2j+1}} \sum_{m=-
        j}^j\ket{j,m}\otimes\ket{j,\alpha(m)}
    \end{align}
where $j$ denotes the total spin of the $N$ spin-$\nicefrac{1}{2}$ systems, $m$ its $z$-component and 
$\alpha(m)$ is the multiplicity index.  The coefficients $\{\beta_j\}_{j=0(1/2)}^{N/2}$ depend on the score 
function used to quantify the estimation task. The most common score function used in the literature is the 
square of the process fidelity\footnote{We note that in~\cite{hayashi2006parallel} the cost function is the 
infidelity $1-F_{\mathrm{CJ}}(\cU,\hat{\cU})^2$} $F_{\mathrm{CJ}}(\cU,\hat{\cU})^2$ between the true unitary gate 
$\cU$ and the estimate $\hat{\cU}$. The average fidelity of estimation is, in the limit of large $N$,
\be
\label{eq:averageestimationfid}
\int \dd \cU \int \dd \hat{\cU}  \, 
{\rm Pr}^{(N)}(\hat \cU|\cU) \, F_{\mathrm{CJ}}(\cU, \hat\cU)^2 =1- \frac{\pi^2}{N^2}\, ,
\ee     
and is achieved by the covariant POVM with density $\{{\rm E}_{\hat{U}}=\hat{U}^{\otimes N}
\ketbra{e}{e}\hat{U}^{\dagger\otimes N}\}$, whose fiducial element is given by 
\begin{align}
    \ket{e}&=\sqrt{N+1}\ket{\frac{N}{2},\frac{N}{2}}
    \nonumber \\
    &+\sum_{j=0(1/2)}^{\frac{N}{2}-1}\sqrt{2j+1} \sum_{m=-j}^j\ket{j,m}\otimes\ket{j,\alpha(m)}\, .
\label{eq:fidPOVM}
\end{align}
By construction, the protocol is symmetric and yields the same average squared process fidelity for all 
$U\in\mathrm{SU}(2)$. Hence, the worst case squared process fidelity is equal to the average squared process 
fidelity.

Using \eqnref{eq: Jensen2} we find that the cloning process using the estimation strategy described above 
achieves a worst case cloning fidelity 
\be \label{eq:MP of su2}
\mathsf{F}_{\rm CJ}^{\Uset|\rm M\&P} \geq \left(1-\frac{\pi^2}{N^2}\right)^{\nicefrac{M}{2}}\, .
\ee
Thus for $M=N^{2-\delta}$, the worst case process fidelity of such an estimate and prepare strategy tends to 
unity in the limit $N\to\infty$, for any $\delta>0$. We believe that similar estimate and prepare processes are 
capable of super-replication also for gates in $\mathrm{SU}(d)$.

\subsubsection{\label{sec:superrepphase} Measure-and-prepare super-replication of all qubit phase gates.}

We now consider an estimate and prepare strategy for cloning the set  
\newcommand{\Phset}{\eucal{P}_2}
\begin{align}\label{eq: phase set}
\Phset&:=\big\{\cU_\theta\, \vert \, \theta\in[0,2\pi)\big\}\\
\cU_\theta[\,\cdot\,]&= U_\theta [\,\cdot\,] U_\theta^\dag :=e^{\ii\theta \sigma_z}\cdot e^{-\ii\theta \sigma_z}. 
\label{def: phase gate}
\end{align}
The optimal estimation strategy has been determined in~\cite{berry2000optimal}, although here we closely follow 
the notation established in~\cite{bartlett2007reference}.  The optimal estimation strategy consists of probing 
the $N$ unitary gates in parallel using the entangled state 
\be\label{eq: sine state}
        \ket{\Psi_N}=\sqrt{\frac{2}{N+2}}\sum_{k=0}^N \, \sin\left(\frac{(k+1)\pi}{N+2}\right) \ket{{\rm 
        D}_k^N}\, .
\ee
where $\ket{\rm{D}_k^N}$ denote the permutationally symmetric state of $N$ spin-1/2 systems, $k$ of which are in 
an excited state---the so-called Dicke states~\cite{Dicke54}.  The relevant score function is the alignment 
fidelity $f(\hat \theta,\theta) = \cos^2\left((\theta-\hat \theta)/2\right)$. Just as in the case of $\mathrm{SU}
(2)$ estimation, the average alignment fidelity is achieved by the covariant POVM with density $\{{\rm 
E}_{\hat{\theta}}=U_{\hat{\theta}}^{\otimes N}\ketbra{e}{e}U_{\hat{\theta}}^{\dagger\otimes N}\}$, where 
$\ket{e}=\sum_{k=0}^N \ket{{\rm D}_k^N}$, and reads, in the limit of large $N$, 
\begin{align}
   \int \dd \hat\theta \, 
{\rm Pr}^{(N)}(\hat \theta|\theta) \, f(\hat \theta,\theta)\, 
    = 1-\frac{\pi^2}{4 N^2}\, .
\end{align}
The protocol is symmetric by construction, hence the worst case alignment fidelity is equal to the average 
alignment fidelity.

Noting that $f(\hat \theta,\theta)=\frac{1}{2}(1+F_{\mathrm{CJ}})$, and using \eqnref{eq:Jensen}, it follows that 
a cloning process based on the above estimation strategy yields 
\be\label{eq:MP of phase}
\mathsf{F}_{\rm CJ}^{\Phset|\rm M\&P} \geq \left(1-\frac{\pi^2}{2N^2}\right)^{M}, 
\ee
which, for $M=N^{2-\delta}$ and in the limit of large $N$, tends to unity for any $\delta>0$, demonstrating super-replication. Note that the rhs of Eq.~\eqref{eq:MP of phase} is always larger than the rhs of 
\eqnref{eq:MP of su2}, which is consistent with the fact that cloning phase gates is by definition simpler than 
cloning all $\mathrm{SU}(2)$ gates.

\section{\label{sec:examples}  Examples of channel cloning}

In this section we study the possibility of replicating non-unitary channels.  Specifically, in 
Sec.~\ref{sec:noisy_unitaries} we consider noisy phase gates on qubits, where the noise consists of a known Pauli 
$\sigma_x$ channel that occurs either before or after the application of the gate. For both these channels we 
provide an explicit cloning process---utilizing quantum error mitigation techniques---that allows for 
super-replication.  In Secs.~\ref{sec:BF} and~\ref{sec:AD} we consider families of Pauli-noise and amplitude damping 
channels for which, by Proposition~\ref{prop:necessary_conds_superrep}, super-replication is not possible.  
Nevertheless we study the performance of several cloning processes, among which are an error mitigation process, and 
a measure-and-prepare process, and derive bounds on their cloning performance. 

We start by summarizing the coherent parallel cloning processes for unitary gates~\cite{Dur2015,Chiribella2015}.

\subsection{Coherent processes for unitary gates}
\label{sec: unitary cloning review}
For completeness, we now briefly review the coherent cloning process of~\cite{Dur2015} for the qubit phase gates 
$\eucal{P}_2$~(Eq.~\ref{eq: phase set}). We will denote the $M$ input qubits $\bm S=S_1\dots S_M$, their 
computational basis states $\ket{\bm k}_{\bm S}$, where $\bm k$ is bitstrings with Hamming weight $|\bm k|$. 
First, an isometry $\mathcal{V}[\,\cdot_{\bm S
}] = V (\, \cdot_{\bm S} \otimes \ketbra{\bm 0}_{\bm A}) V^\dag$ 
is applied to the $M$ input qubits, by preparing $N$ auxiliary qubits $\bm A 
=A_1\dots A_N$ in the state $\ket{\bm 0}_{\bm A}$ and performing a global 
unitary operation such that $V: \ket{\bm k}_{\bm S} \ket{\bm 0}_{\bm A} 
\mapsto \ket{\bm k}_{\bm S} \otimes \ket{{\rm Unary}(n_{|\bm k|}) }_{\bm A}$ with 
\begin{align}\nonumber
\nonumber
\ket{{\rm Unary}(n)} := 
\begin{cases}
    \ket{1}^{\otimes n}\ket{0}^{\otimes(N- n )} & 0\leq n \leq N\\
    \ket{0}^{\otimes N} & {\rm otherwise}\, 
\end{cases}
\end{align}
and $n_{|\bm k|} := \left\lfloor |\bm k|-\nicefrac{(M-N)}{2}\right\rfloor$.
Then, one applies the $N$ copies of the gate $\cU_{\theta}$ onto 
the auxiliary qubits $\bm A$ before uncomputing the unitary and tracing out the 
auxiliary qubits with the channel $\widetilde{\mathcal{V}}[\,\cdot_{\bm S \bm A}]= \tr_{\bm A}\left( V^\dag \, 
\cdot_{\bm S \bm A} \, V \right)$.
The channel realized by this process $\proc{P}[\cU_\theta^{\times N}] = \widetilde{\mathcal{V}}\circ 
\big(\id_{\bm S}\otimes (\cU_\theta^{\otimes N})_{\bm A}\big) \circ \mathcal{V}$ imprints the correct phases on 
the subspace spanned by all computational basis states $\ket{\bm k}_{\bm S}$ with $\big|\nicefrac{M}{2} -|\bm 
k|\big|\leq \nicefrac{N}{2}$. Crucially, for large $M=N^{2-\delta}$, the state $\ket{\Phi^+}^{\otimes M}$ is 
supported on this subspace with unit probability, guaranteeing $\channelF^{\eucal{P}_2}_{\rm CJ}\to 1$. In 
contrast, the phases outside of this subspace are in general wrong, which is easy to detect by probing the 
channel $\proc{P}[\cU_\theta^{\times N}]$ with a state outside of the subspace and consistent with 
$\channelF^{\eucal{P}_2}_{\diamond}\to 0$ implied by Corollary \ref{cor:unitaries_diamond}. 

The coherent cloning process of~\cite{Chiribella2015} is based on a similar idea. In the case of qubits $\eucal{U}_2$ (spin-$\nicefrac{1}{2}$ systems), it correctly imprints the 
action of the unknown gates in all subspaces with total spin less than $\nicefrac{N}{2}$, and can be realized efficiently with the help of the Schur 
transform circuit~\cite{bacon2006efficient}. Again, the same argument about $\channelF^{\eucal{U}_d}_{\rm CJ}\to 1$ and 
$\channelF^{\eucal{U}_d}_{\diamond}\to 0$ can be made for $M=N^{2-\delta}.$

\subsection{\label{sec:noisy_unitaries} Noisy phase gate channels}
We now consider a specific but highly relevant case of noisy-phase-gate qubit channels. Specifically we shall 
consider the following two families of channels  
\newcommand{\APset}{\eucal{A}}
\newcommand{\BPset}{\eucal{B}}

\begin{align}\nonumber
    \APset &:= \{\cA_\theta = \mathcal{X}_p \circ\cU_\theta | \theta \in [0,2\pi) \}\\
    \BPset &:= \{\cB_\theta = \cU_\theta \circ \mathcal{X}_p| \theta \in [0,2\pi) \}
    \label{eq: C tilde mnoisy}
\end{align}
where $\cU_\theta(\cdot)=e^{\ii\theta\sigma_z}(\cdot)e^{\ii\theta\sigma_z}$ is the unitary phase channel and 
\be\label{eq: bit-flip}
\mathcal{X}_p[\,\cdot\,] = (1-p) \cdot \,+ \, p \, \sigma_x \cdot \sigma_x,
\ee
is a bit-flip channel of known strength $0\leq p\leq 1$.  Notice that for the family of channels $\APset$ 
the bit-flip noise occurs {\it after} the application of the phase gate, whilst for $\BPset$ the noise occurs 
{\it before} the phase gate. 
\begin{figure}[t!]
    \centering
      \includegraphics[width=1\linewidth]{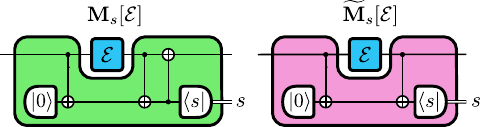}
    \caption{The green (purple) comb defines the error-mitigation (error-detection) process $
    \{\proc{M}_s\}_{s=0}^1$ ($\{ \widetilde{\proc{M}}_s\}_{s=0}^1$). Here, all systems are qubits, the box 
    "$\ket{0}$" prepares the corresponding state, the box "$\bra{s}$" is a measurement in the computational 
    basis, and the representation of the CNOT gates is standard. }
    \label{fig:M}
\end{figure}

A key ingredient in our cloning process for both these families of channels is the error-mitigation process 
$\{\proc{M}_s\}_{s=0}^1$, depicted in green in Fig.~\ref{fig:M}. It is a process that takes a qubit channel $\cE$ 
and returns a quantum instrument $\{\proc{M}_s[\cE]\}$, i.e., two CP maps $\proc{M}_0[\cE]$ and $\proc{M}_1[\cE]$ 
labeled by the classical outcome $s$. It is not difficult to see (App.~\ref{app: EMproc}), that the outcome $s$ 
detects bit flips and hence the CP maps $\proc{M}_s[\cA_\theta]$ (respectively $\proc{M}_s[\cB_\theta]$) have a 
single Kraus operator $K_s$ (respectively $L_s$) given by 
\begin{align}\label{eq:error_mitigation}\nonumber
    K_0 &=L_0 = \sqrt {1-p} \, U_\theta  \\  
     K_1 &= \sqrt{p}\, U_{\theta} \\ \nonumber
    L_1 &= \sqrt{p}\, U_{-\theta}. 
\end{align}
We now separately discuss the sets $\APset$ and $\BPset$.

\subsubsection{Bit-flip noise after the phase gate}
\label{sec: BF after PG}
For the family of channels $\APset$ it is trivial to see from \eqnref{eq:error_mitigation} that the error 
mitigation process corrects the bit flips. After discarding the output $s$ we find
\be
\proc{M}[\cA_\theta]:= \proc{M}_0[\cA_\theta]+\proc{M}_1[\cA_\theta] = \cU_\theta.
\ee
Upon recovering the phase gate $\cU_\theta$ we can now use the coherent process of~\cite{Dur2015} 
(Sec.~\ref{sec: unitary cloning review}), or the measure-and-prepare process 
(Sec.~\ref{sec:superrepphase})), to replicate the unitary gate and append it with the noise map $\mathcal{X}_p$ 
on all of the $M$ output systems. This defines the process 
\be
\proc{P}'[\cA_\theta^{\times N}] := \mathcal{X}_p^{\otimes M}\!\circ \proc{P}[\proc{M}[\cA_\theta]^{\times N}]
\ee
that realizes super-replication of the set $\APset$. The coherent version is illustrated in 
Fig.~\ref{fig: 3processes}a.

\begin{figure*}
    \centering
    \includegraphics[width=0.9\linewidth]{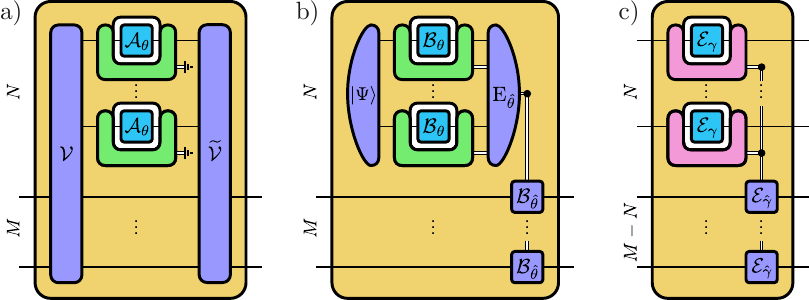}
\caption{ The cloning processes for the following sets of non-unitary qubit channels. {\bf (a)} All phase gates 
followed by fixed bit-flip noise ($\eucal{A}$ in Sec.~\ref{sec: BF after PG}). {\bf (b)} All phase gates preceded 
by fixed bit-flip noise ($\eucal{B}$ in Sec.~\ref{sec: BF before PG}). {\bf (c)} The amplitude-damping channels 
($\eucal{D}$ in Sec.~\ref{sec:coherentprocess}). The processes {\bf(a)} and {\bf(c)} are coherent, while {\bf(b)} 
is measure-and-prepare. The green (purple) comb represent the error-mitigation (error-detection) process 
$\{{\proc{M}}_s\}_{s=0}^1$ ($\{\widetilde{\proc{M}}_s\}_{s=0}^1$) of Fig.~\ref{fig:M}.}
    \label{fig: 3processes}
\end{figure*}

\subsubsection{Bit-flip noise before the phase gate}
\label{sec: BF before PG}
Here the situation is more involved, since depending on the value of $s$ one finds
\be
\proc{M}_0[\cB_\theta] = (1-p)\,  \cU_{\theta} \quad \text{or }\quad \proc{M}_1[\cB_\theta] = p \, 
\cU_{-\theta}.
\ee
We will restrict our attention to measure-and-prepare cloning processes. For a bit string 
$\bm s = (s_1\dots s_N)$ with Hamming weight $|\bm s|= \sum_{i=1}^N s_i$ consider the CP map
\be
\overline{\proc{M}}_{\bm s}[\cB^{\times N}] := \bigotimes_{i=1}^N \proc{M}_{s_i}[\cB_\theta] = p^{|\bm s|}
(1-p)^{N-|\bm s|}\bigotimes_{i=1}^N \cU_{(-1)^{s_i}\theta}.
\ee
Let us introduce the short-hand ${\rm U}_{\bm s} :=\bigotimes_{i=1}^N U_{(-1)^{s_i}\theta}$ and restrict our  
attention to the action of the map on pure $N$-qubit states, $\Psi$, in the symmetric subspace. We have that
\be
{\rm U}_{\bm s}  \ket{\Psi} = 
\begin{cases}
    \bigotimes_{i=1}^{n} U_\theta \ket{\Psi} & |\bm s| \leq \nicefrac{N}{2}\\
    \bigotimes_{i=1}^{n} U_{-\theta} \ket{\Psi} & |\bm s| > \nicefrac{N}{2}
\end{cases},
\ee
with $n=N-2|\bm s|\, \in \{0,\dots,N\}$. Thus, in the symmetric subspace ${\rm U}_{\bm s}$ acts as $n$ copies 
of $U_{\pm \theta}$, where $n$ is known to us via the error-mitigation process.

One could now utilize the estimation strategy described in Sec.~\ref{sec:superrepphase} to optimally estimate the 
value of $\theta$.  Note however that in our case the state of \eqnref{eq: sine state} that maximizes estimation 
precision will depend on $n$ which is not known to us ahead of time. Nevertheless, in the limit $N\to \infty$ the 
binomial distribution concentrates around the mean $\frac{n}{N}\to (1-2p)$, 
which implies that in this limit we can safely set $n=(1-2p)N$ (recall that both $N$ and $p$ are known). Since 
$F_{\rm CJ}(\cB_{\hat{\theta}}, \cB_\theta)= F_{\rm CJ}(\cU_{\hat{\theta}}, \cU_\theta)$ the worst case alignment 
fidelity is given by 
\be
\int \dd \hat\theta \, 
{\rm Pr}^{(n)}(\hat \theta|\theta) \, \fcj(\cU_{\hat \theta}, \cU_\theta)\geq 1-\frac{\pi^2}{2 (1-2p)^2 N^2}\,, 
\ee
assuming that $p\neq1/2$.  By virtue of \eqnref{eq:Jensen} this estimation process can be used to construct an 
$N$ to $M$ measure-and-prepare cloning process that satisfies
\be
\mathsf{F}^{\BPset|\rm M\&P}_{\rm CJ}\geq \left(1- \frac{\pi^2}{2 (1-2p)^2 N^2}\right)^M
\ee
in the large $N$ limit. It is easy to see that this bound is synonymous to super-replication, except for the 
special the case $p=\frac{1}{2}.$ The process is illustrated in Fig.~\ref{fig: 3processes}b.

\subsection{\label{sec:BF}Pauli-noise channels}

\newcommand{\BFset}{\eucal{X}}

Next, let us consider the set of Pauli-noise channels 
\be\label{def: Pauli noise}
 \eucal{P}=\Big\{\cN_{\bm p}[\cdot] = \sum_{j=0}^3 p_j \, \sigma_j \cdot \sigma_i \Big | p_j\geq 0, \sum_j p_j =1\Big\},
\ee
whose \ChJa states are given by 
$\mathtt{CJ}[\cN_{\bm p}]= \sum_{j=0}^3 p_j \,\ketbra{\Phi_i} $ with 
$\ket{\Phi_j}=(\id\otimes \sigma_j)\ket{\Phi^+}$. The Pauli-noise channels are a specific example of ``classical'' 
noise channels, hence by Corollary~\ref{cor:classical noise} they cannot be replicated 
beyond a linear rate. 

The \ChJa fidelity  between two Pauli noise channels reads
\be
F_{\rm CJ}(\cN_{\bm p},\cN_{\bm q})= \sum_{j=0}^3 \sqrt{p_j \,q_j}.
\label{eq:FCJ_bitflip}
\ee
As $\mathtt{CJ}[\cN_{\bm p}]$ is a mixture of fixed states, it is not difficult to see that the optimal dummy 
channel in Eq.~\eqref{eq:F dummy} is $\cE_{\rm dum}=\cN_{\nicefrac{\bm1}{\bm 4}}=\cT_{\nicefrac{\id}{2}}$, leading to the dummy cloning fidelity of $\mathsf{F}_{\rm CJ}^{\eucal{P}|\rm dum} =\left(\nicefrac{1}{2}\right)^{M-N}$, tight with the expression for the full set of channels in Eq.~\eqref{eq: dummy bound}.

Before discussing better replication processes, consider the following proposition 
guaranteeing that optimal cloning of Pauli-noise channels is realized by a measure-and-prepare process.

\begin{proposition}\label{prop: pauli no adv}
For any subset of Pauli-noise channels $\eucal{P}$ in Eq.~\eqref{def: Pauli noise} coherent  processes bring no advantage over measure-and-prepare processes (for any task).
\end{proposition}
\begin{proof} Let $\{\proc{N}_j\}_{j=0}^3$ be the following measure-and-prepare process using a single qubit channel
\be\label{eq: proc N}
\proc{N}_j[\cE] = \bra{\Phi_j} ({\rm id}\otimes\cE)[\Phi^+]\ket{\Phi_j} \, \mathcal{P}_j,
\ee
with the unitary maps $\mathcal{P}_j[\cdot] := \sigma_j \cdot \sigma_j$
It consists of using the input channel to prepare the  state $({\rm id}\otimes\cE)[\Phi^+]$, measuring this state in the Bell basis, and implementing the unitary 
channel $\mathcal{P}_j$. When combined with a Pauli-noise channel the process gives 
\be
\proc{N}_j[\cN_{\bm p}][\cdot] = p_j \, \sigma_j \cdot \sigma_j.
\ee
Hence, upon discarding the classical output the process reproduces the input channel 
$\proc{N}[\cN_{\bm p}]:=\sum_{j} \proc{N}_j[\cN_{\bm p}] = \cN_{\bm p}$. 

Now, let $\proc{P}$ be any process on multiple qubit channels. Combing it with (potentially different) Pauli noise channels one finds 
\begin{align}
    \proc{P}[\bigtimes_k \cN_{\bm p^{(k)}}] =  \proc{P}\big[\bigtimes_k \proc{N}[\cN_{\bm p^{(k)}}]\big] := \proc{P}'[\bigtimes_k \cN_{\bm p^{(k)}}],
\end{align}
where the process $\proc{P}'$, obtained by combining $\proc{P}$ with copies of $\proc{N}$, is  measure-and-prepare. Hence, any possible process on Pauli-noise channels, can be realized in this fashion. The result directly generalizes to any set of channels left invariant by a fixed measure-and-prepare process.
\end{proof}

A convenient property of Pauli noise is that all Kraus operators 
$K_j(\bm p)=\sqrt{p_j}\,\sigma_j$ in Eq.~\eqref{def: Pauli noise} are proportional to 
identity when squared. Hence the probability of the corresponding events---labeled by $j$---is independent of the state on which the channel acts. It follows that the best a measurement process can do with $N$ Pauli-noise channels is to perfectly distinguish all the events labeled by the string $\bm j = (j_1,\dots, j_N)$. In other words, the optimal measurement process is $\{\proc{E}_{\bm j}\}$ such that
\be
\proc{E}_{\bm j}[\cN_{\bm p}^{\times N}]= 
\prod_{k=0}^3 (p_k)^{t_k(\bm j)}  ,
\ee
where $t_k(\bm j)$ is the number of values $j_i=k$ in the string $\bm j$.
This measurement process can be realized via applying  $
\{\proc{N}_j\}$ in Eq.\eqref{eq: proc N} on each Pauli-noise channel in parallel.

Finding the optimal channel estimator $\widehat{\cN}^{(M)}_{\bm j}$ in \eqnref{eq:meawsure&prepare} 
is more challenging. Instead, we now focus on the specific construction
\begin{align}
\proc{P}[\cN_{\bm p}^{\times N}] &= \sum_{\bm j} \proc{E}_{\bm j}[\cN_{\bm p}^{\times N}]\, \widehat{\cN}^{(M)}_{\bm j}  \qquad \text{with}\\
\widehat{\cN}^{(M)}_{\bm j} &:= \left(\bigotimes_{i=1}^N \cP_{j_i}\right)\otimes \cN_{\hat{\bm p}(\bm j)}^{\otimes (M-N)},    
\end{align}
where $\hat {\bm p}(\bm j)$ is an estimator of the parameter \bm $p$, which only depends on the totals $\bm t :=(t_0(\bm j), \dots,t_3(\bm j))$. The idea of this strategy is as follows: upon discarding the classical register $\bm j$ it 
implements the perfect copies of the channel $\cN_{\bm p} = \proc{N}[\cN_{\bm p}]$ on 
the first $N$ systems, while at the same time it estimates the value $\bm p$ and 
implements the estimated channels $\cN_{\hat{\bm  p}}$ on the remaining $M-N$ systems. 

Using the fact that the four states $\mathtt{CJ}[\mathcal{P}_{j}]$ are pure and orthogonal, the fidelity of $\proc{P}[\cN_{\bm p}^{\times N}]$ with $M$ ideal copies $\cN_{\bm p}^{\otimes M}$ of the channel is found to be (see \appref{app: bitflip fid})
\begin{align}\label{eq: opt F coh}
F_{\rm CJ}&(\proc{P}[\cN_{\bm p}^{\times N}], \cN_{\bm p}^{\otimes M})\\
&=\sum_{\bm t } {\rm Pr}_{\rm Mult}(\bm t|N,\bm p) ( F_{\rm CJ}\left( \cN_{\hat {\bm p}(\bm t)} ,\cN_{\bm p}\right))^{M-N}\nonumber,
\end{align}
where ${\rm Pr}_{\rm Mult}(\bm t|N,\bm p)= \binom{N}{t_0,t_1,t_2, t_3} \prod_{k=0}^3 (p_k)^{t_k}$
is the multinational distribution.

In \appref{app: BF} we lower bound the value of the coherent cloning fidelity in 
\eqnref{eq: opt F coh} for the specific choice of the (unbiased) estimator 
$\hat{p}_j(\bm t)=\frac{t_j}{N}$. First, using the inequality $\mathds{E}[\sqrt{X}] \geq 
\sqrt{pN}\left(1-\frac{1-p}{2 p N}\right)$ for binomial random variables 
$X \sim {\rm Bin}(N,p)$ we show that  
${\rm F}_{\rm CJ}^{\eucal{P}} \geq \left(1-\nicefrac{3}{2 N}\right)^{M-N}$.
Second, for a restricted family of channels $\eucal{P}=\{ \mathcal{N}_{\bm p} |p_j \geq \epsilon, \sum_j p_j =1\}$ with some $\epsilon>0$, we also compute the asymptotic value of the optimal worst-case alignment fidelity, leading to the following tighter upper bound on the cloning distance
\be
  \mathsf{D}_{\rm CJ}^{\eucal{P}} \leq \arccos(\exp(-\nicefrac{3}{8}\, \lambda))\, .
  \label{eq:bit_flip_angle}
\ee
Conversely, the lower bound $\mathsf{D}_{\rm CJ}^{\eucal{P}} \geq A\left(\frac{1}{1+\lambda}\right)$ follows form Eq.~\eqref{eq: pauli noise}. Both bounds are depicted in Fig.~\ref{fig:A func} as functions of $z=\frac{1}{1+\lambda}$.  Ultimately, Eq.~\eqref{eq:bit_flip_angle} demonstrates that the constructed cloning process has a vanishing error in the limit $\lambda\to0$ (unlike any estimate-and-prepare process) and a constant assymptotic error at finite $\lambda$ (unlike the dummy process).

When restricting to bit-flip channels $\eucal{X}=\{\cX_p| p\in[\epsilon,1-\epsilon]\}$ in Eq.~\eqref{eq: bit-flip}, i.e. setting $p_2=p_3=0$, the same calculation gives a higher asymptotic fidelity $\mathsf{D}_{\rm CJ}^{\eucal{P}} \leq \arccos(\exp(-\nicefrac{\lambda}{8}))$. Finally in the minimal nontrivial case $1\tightto2$, we then found that the measure-and-prepare protocol with $\widehat{\cX}_0^{(2)}= \cX_{q}^{\otimes 2}$ and $\widehat{\cX}_1^{(2)}= \cX_{1-q}^{\otimes 2}$ for $q\approx 0.0778$ achieves a fidelity $\rm{F}_{\rm CJ}^{\eucal{X}}(1,2)\approx 0.922$, matching the best value found by the SDP in Table~\ref{tab:num}.
\subsection{\label{sec:AD} Amplitude-damping channels}

\newcommand{\adp}{{\gamma}}
\newcommand{\ada}{{\zeta}}
\newcommand{\ADset}{\eucal{D}}
As a last example,  we consider the set
\be
\ADset :=\{\cE_{\adp}| \adp\in[0,1]\}
\ee
of qubit amplitude-damping (AD) channels
\begin{equation}\label{eq: AD Kraus}
\begin{split}
\cE_\adp[\,\cdot\,] &:= K_0(\adp) \cdot K_0^\dag(\adp) +K_1(\adp) \cdot K_1^\dag(\adp),\quad \text{with}\\
K_0(\adp) &= 
\left(\begin{array}{cc}1 & 0\\
0& \sqrt{1-\gamma}
\end{array}\right),\quad
K_1(\adp) = \left(\begin{array}{cc} 0& \sqrt{\gamma}\\
0& 0
\end{array}\right)
\end{split}
\end{equation}
in the computational basis. The Choi-Jamio\l kowski fidelity between any two AD channels 
is 
\begin{align}
F_{\rm CJ}(\cE_\adp,\cE_{\adp'}) 
\!=\! \frac{1}{2}\left(1+\sqrt{\gamma \gamma'}+\sqrt{(1-\gamma)(1- \gamma')}\right).
\label{eq: F AD gam}
\end{align}

For the AD channels the optimization of the dummy process 
in \eqnref{eq:F dummy} is done explicitly in \appref{app: AD dummy}, the optimal  dummy channel is the AD channel $\cE_{\rm dum}=\cE_{\adp=\nicefrac{1}{2}}$ 
leading to the following cloning fidelity
$\mathsf{F}_{\rm CJ}^{\ADset|\rm dum} =\left(\frac{2+\sqrt{2}}{4} \right)^{M-N} 
\!\!\!\approx 0.85^{M-N}.$

In what follows we will first show that super-linear replication  of  AD 
channels is impossible (Sec.~\ref{sec:nosuoperrepAD}), and then study two 
different cloning processes: a measure-and-prepare process 
(Sec.~\ref{sec:estimate&prepare}), and a coherent process 
(Sec.~\ref{sec:coherentprocess}).  We then compare the two (Sec.~\ref{ad: compare}). 

\subsubsection{\label{sec:nosuoperrepAD}AD channels can not be super-replicated}

With the help of Proposition~\ref{prop:necessary_conds_superrep} it is immediate 
to see that AD channels can not be super-replicated. Indeed, with the 
derivatives of the natural Kraus operators
\be\label{eq:AD_channel}
\dot K_0^\dag(\adp)= -\frac{1}{2 \sqrt{1-\gamma }} \ketbra{1}, \quad \dot 
K_1^\dag(\ada)=  \frac{1}{2 \sqrt{\gamma }}\ketbra{1}{0}
\ee
we find that  $\beta(\adp) = \dot K_0^\dag(\adp)  K_0(\adp) + \dot 
K_1^\dag(\adp)  K_1(\adp)  =0$ which rules out the possibility of super-linear 
replication.
In addition, we have 
\be\label{eq: QFI UB AD}
\frac{\mathsf{QFI}^{(N)}(\cE_\adp)}{N}\leq 4 \|\alpha(\adp)\|= \frac{1}{\gamma(1-\gamma)}.
\ee 
For the QFI of the \ChJa states combining Eqs.~\eqref{eq:QFI} and \eqref{eq: F AD gam} we find 
${\rm QFI}(\mathtt{CJ}[\cE_\adp])= \frac{1}{2\gamma(1-\gamma)}$.

Hence, for a linear rate $M=(1+\lambda)N$ the optimal cloning distance is bounded by 
\be\label{eq: D ad nogo}
 \mathsf{D}_{\rm CJ}^{\ADset} \geq  A\left(\frac{2}{1+\lambda}\right).
\ee
This bound resembles the one obtained for Pauli noise channels in \eqnref{eq: 
pauli noise}, but  with an additional factor of $2$. This comes from the fact 
that probing AD channels in parallel through their Choi-Jamio\l kowski states is 
in fact suboptimal by a factor of two from the QFI perspective (it is not 
difficult to see\footnote{Verify that $F(\cE_\adp[\ketbra{1}],\cE_{\adp'}
[\ketbra{1}]) =\sqrt{\gamma \gamma'}+\sqrt{(1-\gamma)(1- \gamma')}$ and use 
Eq.~\eqref{eq:QFI}.} that probing the AD channels with the state 
$\ket{1}^{\otimes N}$ saturates the upper-bound of \eqnref{eq: QFI UB AD}, see 
also~\cite{Fujiwara04}).
As a consequence the bound \eqnref{eq: D ad nogo} remains trivial for all $M\leq 
2N$.

\subsubsection{Estimate-and-prepare cloning of AD channels}
\label{sec:estimate&prepare}

A simple estimation process for the AD channels, maximizing the Fisher information for all $\gamma$, consists of 
probing each copy $\cE_\adp$ with the state $\ket{1}$, and measuring the output system in the computation basis $\{\ket{s\oplus 1}\}_{s=0,1}$, such that the output $s$ corresponds to the application of the Kraus operator $K_s$, and $s=1$ occurs with probability $\gamma$.

For $N$ copies of the channel, the sum of the outcomes $t:=\abs{\bf s}=\sum_{i=1}^{N} s_i$ follows the binomial distribution $ {\rm Pr}_{\rm Bin}(t|N,\adp)$. Now consider estimate-and-prepare cloning processes, that prepare $\{\hat \cE_{\hat \gamma (t)}^{\otimes M}\}_{t =0}^N$ for all possible values $t$. This family of processes achieves the fidelity
\begin{align}\nonumber
     {\rm F}_{\rm CJ}^{\ADset|\rm E\&P} \!\!= \! \min_{\adp \in [0,1]}  F_{\rm CJ}\left( \sum_{t=0}^N  
{\rm Pr}_{\rm Bin}(t|N,\adp) \cE_{\hat \adp(t)}^{\otimes M},\cE_\adp^{\otimes M}\right)
 \end{align}
 with the binomial distribution ${\rm Pr}_{\rm Bin}(t|N,p) := \binom{N}{t} p^t (1-p)^{N-t}$,
 for any choice of the estimator $\hat \gamma(t)$. As the fidelity is nonlinear this expression is very challenging to analyze, even numerically. Instead in Section~\ref{ad: compare} we analyze the lower bound 
 \begin{align}\label{eq: ad mp}
     {\rm F}_{\rm CJ}^{\ADset|\rm E\&P} \geq  \! \min_{\adp \in [0,1]}  \sum_{t=0}^N  
{\rm Pr}_{\rm Bin}(t|N,\adp)  F_{\rm CJ}\left( \cE_{\hat \adp(t)},\cE_\adp\right)^M,
 \end{align}
obtained by concavity of $F_{\rm CJ}$. Before doing so, let us introduce a closely related coherent process.

\subsubsection{\label{sec:coherentprocess}Coherent cloning of AD channels}

Just like for the bit-flip noise the two branches of the AD channel, i.e. the CP maps
\be
\mathcal{K}_{s|\gamma}[\,\cdot\,] := K_{s}(\gamma)\cdot K_s^{\dag}(\gamma) \quad \text{with} \quad s=0,1,
\ee
have a different parity and can be detected with the process $\{\widetilde {\proc{M}}_s\}$ with  $\widetilde {\proc{M}}_0 = \proc{M}_0$ and 
$\widetilde {\proc{M}}_1 = \cX_1 \circ \proc{M}_1$, illustrated in Fig.~\ref{fig:M}. Indeed, one can verify (see \appref{app: EMproc}) that when used on $\cE_\adp$ the process returns $\{\widetilde{\proc{M}}_s[\cE_\gamma]=\mathcal{K}_{s|\gamma}\}$.
This is a quantum instrument allowing one to {\it execute} the unknown channel $\cE_\gamma =\widetilde{\proc{M}}_0[\cE_\gamma]+\widetilde{\proc{M}}_1[\cE_\gamma]$ (upon discarding $s$) while at the same time {\it estimating} the parameter $\gamma$ through the classical output $s$. With this in mind, we now construct a coherent cloning process, illustrated in Fig.~\ref{fig:M}c.

The cloning process $\proc{P}$ first applies the instruments $\{\widetilde{ \proc{M}}_{s_i}[\cE_\adp]\}$ on the first $N$ input qubits.  Depending on the  value $t=\abs{\bm s}=\sum_{i=1}^{N} s_i$ observed, it applies the fixed AD channel $\cE_{\hat{\adp}(t)}$ on the remaining $M-N$ qubits, and   then discards the classical register $\bf s$. 
It follows that the result of this process is the $M$-qubit CPTP map 
\be
\proc{P}[\cE_\adp^{\times N}] =\! \sum_{s_1,\dots,s_N}\!\left( \mathcal{K}_{s_1|\gamma} \otimes\dots\otimes \mathcal{K}_{s_N|\gamma} \otimes (\cE_{\hat{\adp}(\bf s)})^{\otimes (M-N)} \right)
\ee
and it remains to find the optimal estimator $\hat{\gamma}(t)$ for each $N$ and $M$.

To do so we first compute the fidelity achieved by this process for all values of the parameter $\gamma$. Observe that here the maps $\mathcal{K}_{s|\gamma}$  are not applied on $\ket{1}$ but on half of the maximally entangled state $\ket{\Phi^+}=\frac{1}{\sqrt{2}}(\ket{00}+\ket{11})$.  The output $s$ thus takes the value $s= 0$ with probably $\|\id \otimes K_0(\adp)\ket{\Phi^+}\|^2= 1-\frac{\gamma}{2}$ and the value $s=1$ with probability $\| \id \otimes K_1(\adp) \ket{\Phi^+}\|^2= \frac{\gamma}{2}$.
Hence the sum $t=\|\bm s|=\sum_i s_i$ follows the binomial distribution ${\rm Pr}_{\rm Bin}(t|N,\nicefrac{\adp}{2})$.

This allows one (see \appref{app: ad coherent}) to obtain the following expression for the fidelity $ F_{\rm CJ}(\proc{P}[\cE_\adp^{\times N}], \cE_\adp^{\otimes M}) 
 = \sum_{t=0}^N 
{\rm Pr}_{\rm Bin}(t|N,\nicefrac{\adp}{2}) F_{\rm CJ}(\cE_{\hat \gamma(t)},\cE_\gamma)^{M-N}$
and shows that the coherent process achieves
 \begin{align}\label{eq: ad ep}
     {\rm F}_{\rm CJ}^{\ADset|\rm coh} 
     &=  \min_{\gamma \in [0,1]}  \sum_{t=0}^N  
{\rm Pr}_{\rm Bin}(t|N,\nicefrac{\adp}{2}) F_{\rm CJ}(\cE_{\hat \gamma(t)},\cE_\gamma)^{M-N}
\end{align}
for any choice of the estimator $\hat{\gamma}(t)$. We now compare the two processes.

 \subsubsection{Performance of the estimate-and-prepare and coherent processes}
\label{ad: compare}

To evaluate the fidelities ${\rm F}_{\rm CJ}^{\eucal{D}|\bullet}$ acheived by the estimate-and-prepare ($\bullet=\rm E\&P$) and the coherent ($\bullet=\rm coh$) processes, notice that the expressions in Eqs.~(\ref{eq: ad mp},\ref{eq: ad ep}) are very similar. Using Jensen's inequality \eqnref{eq:Jensen} and the expression of the fidelity \eqnref{eq: F AD gam} both are bounded by
\begin{align}
{\rm F}_{\rm CJ}^{\eucal{D}|\bullet}&\geq \min_{\adp \in [0,1]} \left( \sum_{t=0}^N  
{\rm Pr}_{\rm Bin}(t|N,p_\bullet)  F_{\rm CJ}\left( \cE_{\hat \adp(t)},\cE_\adp\right)\right)^{L_\bullet} \nonumber \\
&=\min_{\gamma\in [0,1]} \mathds{E}\left[\frac{ 1+\sqrt{\gamma \, \hat \gamma_\bullet(t)}+\sqrt{(1-\gamma)(1-\hat \gamma_\bullet(t))}}{2}\right]^{L_\bullet},
\end{align}
with a different parameter $p_{\rm E\&P}=\gamma$ and $p_{\rm coh}=\nicefrac{\gamma}{2}$ of the binomial, and a different exponent $L_{\rm M\&P}=M$ and $L_{\rm coh}=M-N$. 

In the large $N$ limit, just as in the case of the Pauli-noise channels, if we restrict the family of AD channels $\eucal{D}=\{ \mathcal{E}_\gamma\,|\, \gamma\in [\epsilon,1-\epsilon]\}$ for some $\epsilon>0$, for the natural choices of estimators  we can compute these expected values to obtain 
\be
{\rm F}_{\rm CJ}^{\ADset|\rm E\& P} \geq  \exp(-\frac{1+\lambda}{16})
\qquad 
{\rm F}_{\rm CJ}^{\ADset|\rm coh} \geq  \exp(-\frac{\lambda}{8})
\label{eq:tighter_dist_AD}
\ee
with $M=(1+\lambda)N$. It is interesting to note that the  bound we obtain for 
the coherent process is only better if $M\leq 2N$. This behavior can be 
understood intuitively---the coherent process is suboptimal from the parameter 
estimation perspective (bad for large $M$) but ideally replicates $N$ channels  
(good for small $M$). In particular, it achieves ${\rm F}_{\rm CJ}^{\ADset|\rm 
coh}=1$ at $N=M$, which, for amplitude-damping channels, we conjecture to be impossible for any measure-and-prepare process.

To study the processes in the moderate $N$ regime we 
optimized the estimators $\hat \gamma(t)$ in Eqs.~(\ref{eq: ad mp},\ref{eq: ad ep}) numerically. To do so we proceed similarly to the SDP approximation of Section \ref{sec:sdp}. We start by replacing the interval $\gamma\in[0,1]$ with a discrete net $\adp\in \{\adp_i \}_{i=1}^H$, and rephrasing the minimization as $H$ feasibility constraints 
\be
\sum_{t=0}^N  
{\rm Pr}_{\rm Bin}(t|N,p_\bullet^{(i)})  F_{\rm CJ}\left( \cE_{\hat \adp(t)},\cE_{\adp_i}\right)^{L_\bullet}\geq x\quad \forall\, \gamma_i,
\ee
with $p^{(i)}_{\rm M\&P}=\gamma_i$ and $p^{(i)}_{\rm coh}=\nicefrac{\gamma_i}{2}$. We then run an empirical constrained optimization of $x$ with respect to $N+1$ real variables $\hat{\adp}(t)\in [0,1]$. The results obtained with {\tt FindMaximum} in {\tt Mathematica} for $H=1000, N=1$ or $5$, and $M\leq 20$ are reported in Fig.~\ref{fig: AD clon}. These numerical results suggest that the coherent process outperforms the measure-and-prepare one when $M$ is close to $N$, i.e. $5\to 6$, $5\to 7$ and the trivial cases $M=N$, but becomes worse when $M$ increases. This comparison is to be taken with a grain of salt, as for the estimate-and-prepare process we only optimized the lower bound \eqnref{eq: ad mp}, so in reality it might compare more favorably with the the coherent one.

For the minimal nontrivial case $1\to 2$ we found the values ${\rm F}_{\rm CJ}^{\ADset|\rm M\&P}(1,2) \approx 0.925$,  ${\rm F}_{\rm CJ}^{\ADset|\rm coh}(1,2) \approx 0.900$ and ${\rm F}_{\rm CJ}^{\ADset|\rm dum}(1,2)\approx 0.854$. Remarkably in Sec.~\ref{sec:numerics} we have seen that the general process found by the SDP approximation achieves a higher value ${\rm F}_{\rm CJ}^{\ADset|\rm SDP}(1,2)\approx 0.932$ (see Table~\ref{tab:num}), outperforming all our hand-crafted processes.

\begin{figure}[t!]
    \centering
    \includegraphics[width=1\linewidth]{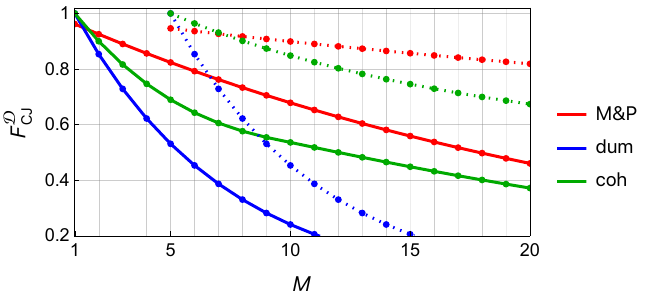}
    \caption{The cloning fidelity ${\rm F}_{\rm CJ}^{\ADset}(N,M)$ for the amplitude-damping channels (Eq.~\ref{eq: AD Kraus}) achieved by the dummy (blue lines), measure-and-prepare (red lines, Sec.~\ref{sec:estimate&prepare}) and the coherent (green lines, Sec.~\ref{sec:coherentprocess}) cloning processes, as function of $M\geq N$. The tfull lines correspond to $N=1$,  and the dotted lines to $N=5$.}
    \label{fig: AD clon}
\end{figure}

\section{\label{sec:concl} Summary and conclusion}
This work extends the notion of cloning and replication beyond quantum states 
and gates to the most general transformations of quantum theory: quantum 
channels.  By demonstrating the equivalence between cloning of states and 
cloning the respective trash-and-replace channels (Proposition \ref{prop:cloning-trash}) we are able to formulate cloning and replication of all quantum 
resources under a unified framework, 
\yg{cast} 
in the language of quantum processes.  
By employing a geometric approach we establish a connection between cloning and the fundamental tasks of binary channel discrimination and channel estimation, which allows us to derive upper bounds on the performance of any cloning and 
replication process (Propositions~\ref{prop:cloning_bound_discrimination} 
and~\ref{prop:metrology_optimal_distance}). Independently of cloning, we 
believe the latter to be of great interest for binary channel discrimination. 

We also establish the necessary conditions for a family of channels to exhibit 
super-replication (Proposition~\ref{prop:necessary_conds_superrep}) and prove 
strong converses on the deterministic replication rates for any continuous 
family of states (Corollary~\ref{cor:state_rep}), quantum gates under the 
diamond distance (Corollary~\ref{cor:unitaries_diamond}), and parametric 
families of ``classical" noise channels (Corollary~\ref{cor:classical noise}).  
For all of the above we provide rigorous upper bounds on the cloning fidelity in 
the limit of large $N$ and $M$. Moreover, while the full set of $d$-dimensional 
unitaries admits super-replication, we show that even infinitesimal noise 
reduces the rate to linear (Corollary~\ref{cor:no robustness}).

Motivated by the converse question, in Section~\ref{sec: clon proc}, we devise 
general techniques to approximate the optimal cloning process for a given set of 
channels. Building on the results 
of~\cite{Milz2024characterising,Skrzypczyk2023} we show in Section~\ref{sec:sdp} 
that the search of the optimal cloning process can be approximated by a semi-
definite program (SDP), and we illustrate this approach with $1\tightto 2$ 
cloning for the bit-flip and amplitude-damping channels. We concluded that, for the amplitude-damping channels, the optimal process found by the SDP outperforms all measure-and-prepare process and also the coherent hand-crafted process. We also study the performance of measure-and-prepare processes (Section~\ref{sec:measureandprep}) which enables us to draw a quantitative connection between the tasks of cloning and Bayesian channel estimation. Somewhat surprisingly, we find that 
super-replication of qubit gates, demonstrated in Refs.~\cite{Dur2015,Chiribella2015} with {\it 
coherent} processes,  can also be realized by measure-and-prepare process Sections 
\ref{sec:superrepsu2} and \ref{sec:superrepphase}), and can 
be viewed as a consequence of the quadratic scaling of precision in the corresponding Bayesian estimation
tasks~\cite{berry2000optimal,Chiribella2005,hayashi2006parallel}.

Finally, in Section~\ref{sec:examples}, we study cloning process for several 
specific families of quantum channels: noisy phase gates---where the noise acts 
either before or after the application of the phase gate---Pauli-noise channels 
and  amplitude-damping channels. We show that for the first two cases  
super-replication is possible, the first instance of a non-unitary channel where this is possible. In contrast, both Pauli-noise and amplitude-damping channels can 
only be replicated at a linear rate. For Pauli-noise channels, we prove that coherent processes offer no advantage over measure-and-prepare processes in {\it any} 
task (Proposition~\ref{prop: pauli no adv}), and discuss the asymptotic performance of such a process. The case of amplitude-damping channels is more subtle, and we focus on  the 
comparisons between two hand-crafted processes. .

Our findings raise several interesting questions for future research. Whilst 
Proposition~\ref{prop:necessary_conds_superrep} establishes necessary 
conditions for super-replication it is worth noting that we do not know whether 
these are also sufficient. In particular, the proposition was derived 
via a construction in the
local neighborhood $\dd \eucal{C}=\{\cE_x| x\in[0,\delta)\}\subset 
\eucal{C}$ of a channel $\cE_x$ and 
by studying the scaling of the (optimal 
$N$-copy) $\mqfi^{(N)}(\cE_x)$. Yet, even when all such subsets $\dd \eucal{C}$ 
can be super-replicated the processes may be different and thus it is not clear if 
super-replication is achievable across the whole set $\eucal{C}$ with a single process. Answering this 
question requires consideration of the global properties of the set. 
Interestingly, a sufficient condition for super-replication is offered by the 
estimate-and-prepare processes discussed in Sec.~\ref{sec:measureandprep}---it is 
possible if the channels $\eucal{C}$ can be {\it estimated} with the worst-case 
error that decays as $N^2$, for an appropriate cost function. Establishing 
closed-form sufficient conditions for such Heisenberg scaling in Bayesian estimation task and for channel super-replication is an interesting open question. 

Another interesting direction for future research, is to rephrase the question 
of replication/cloning at the level of continuous semi-group dynamics. For instance, does there exist a process that, given access to a system evolving under a given parametrized family of Liouvillian dynamics for a total time $t$, simulate its evolution for a time $t'>t$. For Hamiltonian dynamics this can be seen as an instance of the task of transforming Hamiltonian eigenvalues~\cite{odake2025universal}, while nothing seems to be known for the case of open-system dynamics.

\section{Acknowledgements}
We thank Jessica Bavaresco and Marco T\'ulio Quintino for useful discussions. PS acknowledges financial support from the Swiss National Science Foundation NCCR SwissMAP. MS acknowledges support from Ayuda Ramón y Cajal 2021 (RYC2021-032032-I, MICIU/AEI/10.13039/501100011033, ESF+) as well as Project FEDER C-EXP-256-UGR23 Consejería de Universidad, Investigación e Innovación y UE Programa FEDER Andalucía 2021-2027.  N.B.T.K. acknowledges support by the European Space Agency (EISI project 2021-01250-ESA), and by the Spanish MICIN (project PID2022-141283NB-I00) with the support of FEDER funds.

\appendix
\onecolumngrid

\section{Proof of Proposition~\ref{prop:cloning-trash}.}
\label{app:states_chan_proof}

\subsection{A preliminary lemma}

We start by proving the following lemma which is of independent interest from 
cloning.
\begin{figure}
    \centering
    \includegraphics[width=0.85\linewidth]{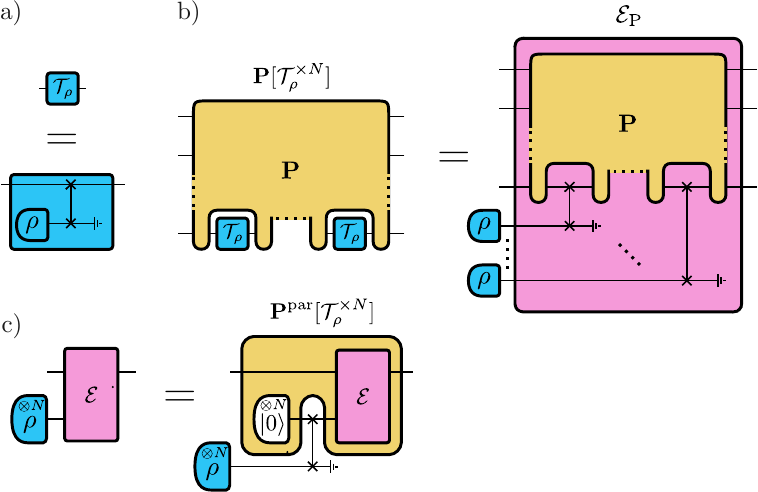}
    \caption{{\bf a)} The trash-and-replace channels $\cT_\rho$ can be realized 
    with a SWAP gate acing on a copy of $\rho$. {\bf b)} Any process $\proc{P}$ 
    (also non-causal) combined with $N$ trash-and-replace channels 
    $\mathcal{T}_\rho$ defines a (fixed) CPTP map $\cE_{\proc P}$ (purple box), 
    acting on the same input systems extended with $N$ copies of $\rho$. 
    {\bf c)} Any such CPTP map $\cE$ can be realized via a \emph{parallel} 
    process $\proc{P}^{\rm par}$, hence upon combining with trash-and-replace 
    channels the hierarchy of processes collapses.}
    \label{fig:lemm5}
\end{figure}

\begin{lemma}\label{lem: process dec}
For any process $\proc{P}$ taking $N$ channels, the CPTP map 
$\proc{P}[\mathcal{T}_\rho^{\otimes N}]$  resulting from using the process with 
$N$ trash-and-replace channels $\mathcal{T}_\rho$ is of the form
\begin{equation}
    \proc{P}[\mathcal{T_\rho}^{\times N}] [\,\cdot\,] = \tr_{\bm S'} \cE_{\rm P}
    [\,\cdot \,\otimes (\rho^{\otimes N})]: \, \cL(\cH^{\otimes M}\tightto 
    \cH^{\otimes M}),\label{eq: lemma 5}
\end{equation}
for some fixed CPTP map $\cE_{\rm P}$ (determined by the process $\proc{ P}$ and 
independent of $\rho$). In addition, one obtains the full set of CPTP maps 
$\cE_{\rm P}$ staring with parallel, sequential or non-causal processes $\bm P$.
\end{lemma}

\begin{proof} 
We first show that any process taking trash-and-replace channels must be of the 
form $\proc{P}[\mathcal{T_\rho}^{\times N}] [\cdot_{\bm S}] = \tr_{\bm S'} 
\cE_{\rm P}[\,\cdot_{\bm S} \otimes (\rho^{\otimes N})_{\bm S'}]$
where $\bm S(\bm S')$ contains $M(N)$ copies of the system $S$. First, we note 
that the trash-and-replace channel acting on a system $S$ can be realized with a 
SWAP gate as follows (see Fig.~\ref{fig:lemm5}a)
\be
\cT_\rho[\cdot_S] = \tr_{S'} {\rm SWAP}_{SS'} [\cdot_S\otimes \rho_{S'}].
\ee
Hence,  we can write  
\be
\proc{P}[\mathcal{T_\rho}^{\times N}] [\cdot_{\bm S}] = \tr_{\bm S'} \proc{P}[ 
{\rm SWAP}^{\times N}] [\cdot_{\bm S}\otimes (\rho^{\otimes N})_{\bm S'}],
\ee
which is shown Fig.~\ref{fig:lemm5}b. 
Here the linear map $\cE_{\proc P} := \tr_{\bm S'} \proc{P}
[{\rm SWAP}^{\times N}]$ (the purple box in Fig.~\ref{fig:lemm5}b) must be
CPTP -- by definition a valid process must give a CPTP map when combined with CPTP  maps (also extended). 

Next, we show that staring with \emph{parallel} processes $\proc{P}^{\rm par}$ 
one can obtain \emph{any} CPTP map $\cE: \cL(\cH^{M+N}\tightto \cH^{M})$ on the 
rhs of Eq.~\eqref{eq: lemma 5}. This can be done by taking the target channel 
$\cE$ and defining the parallel process shown in Fig.~\ref{fig:lemm5}c. 
Formally, this process can be defined as
\be
\proc{P}^{\rm par}[\cC^{\times N}] [\,\cdot_{\bm S}\,] := \cE[\,\cdot_{\bm S}\, 
\otimes (\cC[\ketbra{0}]^{\otimes N})_{\bm S'}].
\ee
One easily verifies that, when applied on trash-and-replace channels, the 
process gives $\proc{P}^{\rm par}[\cT_\rho^{\times N}] [\cdot_{\bm S}] = 
\cE[\cdot_{\bm S} \otimes (\rho^{\otimes N})_{\bm S'}].$ Hence, for trash-and-
replace channels the hierarchy of parallel, sequential and non causally ordered 
processes collapses. This 
completes the proof.
\end{proof}

\subsection{The main proof}
\label{app: TaR is states}

Next we use Lemma~\ref{lem: process dec} to prove 
Proposition~\ref{prop:cloning-trash} in the main text.

\begin{proof}
We readily know that $\mathsf{F}^{\eucal{C}}_{\rm CJ}\geq
\mathsf{F}_{\diamond}^\eucal{C}$, since by construction the CJ fidelity for 
two channels is greater or equal than the diamond fidelity.

First, we show that the optimal state cloning machine 
$\cP:L(\mathcal{H}^{\otimes N})\to L(\mathcal{H}^{\otimes M})$ can be used to 
define a process $\proc{P}_{\cP}$, cloning trash-and-replace channels, which 
satisfies $\mathsf{F}^{\eucal{C}}_{\rm CJ}=\mathsf{F}_{\diamond}^\eucal{C} =  
\mathsf{F}^\eucal{S}$. Consider the process which discards all of its $M$ input 
quantum systems,  calls in parallel the $N$ received trash-and-replace channels 
($\mathcal{T}_\rho$) with some fixed input states $\ket{0}$ and collects the $N$ 
outputs ($\rho^{\otimes N}$). The process then applies the optimal cloning map $\cP$ 
on these $N$ outputs to produce the $M$ output systems in the state $\mathcal{P}[\rho^{\otimes N}]$. 
When fed with $N$ copies of individual trash-and-replace channels the process $\proc{P}_{\cP}$ 
we just described realizes the global trash-and-replace channel 
\begin{align}\nonumber
\proc{P}_{\cP}[\mathcal{T}_\rho^{\times N}]=\mathcal{T}_{\mathcal{P}[\rho^{\otimes N}]} : \quad 
L(\cH^{\otimes M}) &\to L(\cH^{\otimes M}) \\
\cdot\, &\mapsto \mathcal{P}[\rho^{\otimes N}]\, .
\end{align}
We now compute the two channel fidelities of the resulting map with respect to $M$ copies of $\mathcal{T}_\rho$. 
For the CJ fidelity we find
\begin{align}\nonumber
    F_{\rm CJ}(\proc{P}_{\cP}[\mathcal{T}_\rho^{\times N}] ,\mathcal{T}_\rho^{\otimes M} ) &= F(({\rm id}\otimes 
    \mathcal{T}_{\mathcal{P}[\rho^{\otimes N}]})[\Phi^+],({\rm id}\otimes\mathcal{T}_\rho^{\otimes M})[\Phi^+] 
    )\\ 
   &= F\left(\nicefrac{\id}{d} \otimes \mathcal{P}[\rho^{\otimes N}],\nicefrac{\id}{d} \otimes \rho^{\otimes 
   M}\right) \\\nonumber
   &= F\left( \mathcal{P}[\rho^{\otimes N}], \rho^{\otimes M}\right)
\end{align}
where we used $F(\rho\otimes\sigma,\rho'\otimes\sigma')=F(\rho,\rho')F_{\nicefrac{1}{2}}(\sigma,\sigma')$. Note 
that in the above calculation we can replace the state $\Phi^+$ with any state of the extended system, hence an 
identical computation gives $ F_{\diamond}(\proc{P}_{\cP}[\mathcal{T}_\rho^{\times N}] ,\mathcal{T}_\rho^{\otimes 
M}) = F\left( \mathcal{P}[\rho^{\otimes N}], \rho^{\otimes M}\right)$. Since by assumption $\cP$ is the optimal 
cloning machine, and for the process  $\proc{ P}_{\cP}$ the equality $F_{\diamond}(\proc{P}_{\cP}
[\mathcal{T}_\rho^{\times N}],\mathcal{T}_\rho^{\otimes M} ) = F\left( \mathcal{P}[\rho^{\otimes N}], 
\rho^{\otimes M}\right)$ holds for all states $\rho$, we conclude that this cloning process gives a worst-case 
diamond cloning fidelity equal to $\mathsf{F}^\eucal{S}$. Hence it must be that  the optimal channel cloning 
fidelity is at least that good
\begin{align}
\label{eq: CJ,d,st ub}
\mathsf{F}^{\eucal{C}}_{\rm CJ}\geq \mathsf{F}_{\diamond}^\eucal{C} \geq  \mathsf{F}_{\nicefrac{1}{2}}^\eucal{S}.
\end{align}

In the second step we will show that $\mathsf{F}^{\eucal{C}}_{\rm CJ}\leq \mathsf{F}_{\nicefrac{1}{2}}^\eucal{S}$ 
collapsing the chain of inequalities.  For any cloning process $\proc P$ let us compute the CJ fidelity between 
${\proc P}[\mathcal{T_\rho}^{\times N}]$ and $M$ copies of the channel
\begin{align}\nonumber
F_{\rm CJ}\big( {\proc P}[\mathcal{T_\rho}^{\times N}],\mathcal{T_\rho}^{\otimes M}\big) &=
F\big( ({\rm id}_A\otimes{\proc P}[\mathcal{T_\rho}^{\times  N}]_S)[\Phi^+_{\bm A\bm S}],
({\rm id}_A\otimes(\mathcal{T_\rho}^{\otimes M})_S)[\Phi^+_{\bm A\bm S}]\big)\\ \nonumber
&\leq F\big( \tr_{\bm A} ({\rm id}_A\otimes{\proc P}[\mathcal{T_\rho}^{\times N}]_S)
[\Phi^+_{\bm A\bm S}],\tr_{\bm A} ({\rm id}_A\otimes(\mathcal{T_\rho}^{\otimes M})_S)
[\Phi^+_{\bm A\bm S}]\big) \\ \nonumber
& = F\big(  {\proc P}[\mathcal{T_\rho}^{\times N}])[\nicefrac{\id}{d}], \rho^{\otimes M}\big) \\
& = F\big( \tr_{\bm S'} \cE_{\rm P}[\,(\nicefrac{\id}{d})_{\bm S} \otimes (\rho^{\otimes N})_{\bm S'}], 
\rho^{\otimes M}\big),
\end{align}
where we used the monotonicity of fidelity (aka data processing inequality) in the first line, and 
lemma~\ref{lem: process dec} in the last one. Here, let us define the cloning machine (CPTP map) appearing in the 
last line and induced by the process $\proc P$
\begin{align}
\cP_{\rm P}[\,\cdot\,]:=  \tr_{\bm S'} \cE_{\rm P}[\,(\nicefrac{\id}{d})_{\bm S} \otimes 
(\rho^{\otimes N})_{\bm S'}]: L(\cH^{\otimes N}) \to L(\cH^{\otimes M}).
\end{align}
Finally letting $\proc P$ be the optimal cloning process for the CJ fidelity (achieving 
$\mathsf{F}^{\eucal{C}}_{\rm CJ}$) we see that the induced state cloning machine $\cP_{\rm P}$ achieves  
$F\big( \cP_{\rm P}[\rho^{\otimes N}], \rho^{\otimes M}\big)\geq F_{\rm CJ}\big( {\proc P}
[\mathcal{T_\rho}^{\times N}],\mathcal{T_\rho}^{\otimes M}\big) \geq \mathsf{F}^{\eucal{C}}_{\rm CJ}$ for all 
states from the set $\eucal{S}$. Hence, it must be that 
\be
 \mathsf{F}_{\nicefrac{1}{2}}^\eucal{S}\geq \mathsf{F}^{\eucal{C}}_{\rm CJ}
\ee
which together with Eq.~\eqref{eq: CJ,d,st ub} implies 
$\mathsf{F}^{\eucal{C}}_{\rm CJ} =\mathsf{F}_{\diamond}^\eucal{C} = \mathsf{F}_{\nicefrac{1}{2}}^\eucal{S}$ and 
concludes the proof.
\end{proof}

Note that the proof can be readily generalized to the equivalence of optimal cloning of states and trash-and-
replace channels with respect to distances $D_S(\rho,\sigma)$ on states and $D_C(\cA,\cB)$ on channels whenever, 
for all states from $\eucal{S}$ these distances simultaneously satisfy
\begin{align}\nonumber
    D_{C}(\mathcal{T}_{\mathcal{P}[\rho^{\otimes N}]} ,\mathcal{T}_\rho^{\otimes M}) &\leq D_{S} \left( 
    \mathcal{P}[\rho^{\otimes N}], \rho^{\otimes M}\right)\\
D_{C}\big( \tr_{\bm S'} \cE_{\rm P}[\,\cdot \otimes (\rho^{\otimes N})_{\bm S'}], 
\mathcal{T_\rho}^{\otimes M}\big) &\geq D_{S}\big( \tr_{\bm S'} \cE_{\rm P}[\,\sigma_{\bm S} \otimes 
(\rho^{\otimes N})_{\bm S'}], \rho^{\otimes M}\big)\, ,
\end{align}
for some fixed state $\sigma_{\bm S}$. These inequalities can be plugged in the proof to yield 
$\mathsf{D}^{\eucal{C}}\leq \mathsf{D}^\eucal{S}$ and $\mathsf{D}^{\eucal{C}}\geq \mathsf{D}^\eucal{S}$ 
respectively.

\section{\label{app:states_and_Us} Proof of Corollaries~\ref{cor:state_rep} and ~\ref{cor:unitaries_diamond}}	

\begin{proof}(Corollary~\ref{cor:state_rep})
First note that a replication rate of $\mathsf{R}^{\eucal{S}}=1$ can be trivially achieved via a replication 
process corresponding to the identity.  To prove the strong converse note that, by 
Proposition~\ref{prop:cloning-trash}, cloning of states is equivalent to the cloning of the trash-and-replace 
channels.  Thus, substituting $\cE_0=\cT_{\rho_0}, \cE_1=\cT_{\rho_1}$ in \eqnref{eq:triangle_ineq}, and using 
the data processing inequality \(D_\bullet(\proc{P}[\cT_{\rho_0}^{\times N}],\proc{P}
[\cT_{\rho_1}^{\times N}])\leq D_\bullet(\cT_{\rho_0}^{\otimes N},\cT_{\rho_1}^{\otimes N})=
D(\rho_0^{\otimes N},\rho_1^{\otimes N})\) the geometric bound in \eqnref{eq:optimal_cloning_distance} of 
Proposition~\ref{prop:cloning_bound_discrimination} reads 
     \be
		\mathsf{D}^{\eucal{S}}\geq \frac{1}{2}\left(D(\rho_0^{\otimes M}, \rho_1^{\otimes M})
    -D(\rho_0^{\otimes N}, \rho_1^{\otimes N})\right)
        \, .
	\label{app:distance_bound_states}
	\ee
    
As the set $\eucal{S}$ is continuous we can always choose two states $\rho_0, \rho_1\in\eucal{S}$ such that 
	\be
		F(\rho_0,\rho_1)= 1-\frac{1}{N^{1+\delta}}
	\label{app:fid_delta}
	\ee	
for $\delta>0$.  It follows that 
	\be
		\begin{split}
			\lim_{N\tightto\infty}F(\rho_0^{\otimes N}, \rho_1^{\otimes N})&=\lim_{N\tightto\infty}\left(1-
            \frac{1}{N^{1+\delta}}\right)^N\\
			&=\lim_{N\tightto\infty}(1-N^{-\delta})=1\, 
		\end{split}
	\label{app:fid_N}
	\ee	
which implies that $\lim_{N\tightto\infty}D(\rho_0^{\otimes N}, \rho_1^{\otimes N})=0$. On the other hand, 
setting $M=N^R$, with $R=1+\delta$ we obtain 
	\be
		\begin{split}
			\lim_{N\tightto\infty}F(\rho_0^{\otimes M}, \rho_1^{\otimes M})&=\lim_{N\tightto\infty}\left(1-
            \frac{1}{N^{1+\delta}}\right)^{N^R}\\
			&=\lim_{N\tightto\infty}(1-N^{1+R-(1-\delta)})=0\, 
		\end{split}
	\label{fid:M}
	\ee
which implies that $\lim_{N\tightto\infty}D(\rho_0^{\otimes M}, \rho_1^{\otimes M})=\frac{\pi}{2}$.  
It follows that for any $R=1+\delta$ with $\delta>0$ the optimal cloning distance asymptotically satisfies 
$\mathsf{D}^{\eucal{S}}\geq \frac{\pi}{4}$, which in turn means that the optimal cloning fidelity satisfies 
$\channelF^{\eucal{S}}\leq 1/\sqrt{2}$. 
\end{proof}

\begin{proof}(Corollary~\ref{cor:unitaries_diamond})
Again we note that a replication rate of $\mathsf{R}^{\eucal{C}}_{\diamond}=1$ can be trivially achieved by a 
superchannel corresponding to the identity.  To prove the strong converse recall that 
$\mathsf{D}^{(N)}(\cU_0, \cU_1)$ is maximized by strategies that do not make use of entanglement with auxiliary 
systems.  Using the minimization over the corresponding fidelity we obtain 
	\be
		\begin{split}
			\channelF^{(N)}\left(\cU_0, \cU_1\right)& =\min_{\psi\in\cL(\cH^{\otimes N})} 
            F\left(\cU_0^{\otimes N}[\psi],\cU_1^{\otimes N}[\psi]\right)\\  
			\nonumber
			&=\min_{\ket{\psi}\in\cH^{\otimes N}}\abs{\bra{\psi}(U_0^\dagger U_1)^{\otimes N}\ket{\psi}}\, .
		\end{split}
	\label{app:fid_diamond}
	\ee
Writing $U_0^\dagger U_1$ in its diagonal basis 
	\be
		U_0^\dagger U_1 = \sum_j e^{\ii \theta_j} \, \ketbra{j}{j}
	\label{app:U0U1_basis}
	\ee
the state minimizing the fidelity in \eqnref{app:fid_diamond} is the equal superposition  
$\ket{\psi} := \frac{1}{\sqrt{2}}\left(\ket{j}^{\otimes N}+\ket{j'}^{\otimes N}\right)$ of the eigenstates 
$\ket{j}$ and $\ket{j'}$ corresponding to the maximal angular difference
\be
\Theta := \max_{\theta_{j},\theta_{j'}} \big(|\theta_{j'}-\theta_j| \,\,{\rm mod} \,\, \pi \big),
\ee
where we assumed that $N\Theta \leq \pi$. Hence, for the optimal fidelity we find
    \begin{align}
       \channelF^{(N)}\left(\cU_0, \cU_1\right) =\abs{\bra{\psi}(U_0^\dagger U_1)^{\otimes N}\ket{\psi}}
        =
  \bigg|
     \frac{ e^{i N \theta_{j}} + e^{i N \theta_{j'}} 
       }{2}\bigg|= \cos\left(\frac{N\Theta}{2}\right),
    \end{align}
from which it immediately follows that $\channelD^{(N)}\left(\cU_0,\cU_1\right)=\frac{N\Theta}{2}$.  
A similar calculation yields $D_{\diamond}\left(\cU_0^{\otimes M}, \cU_1^{\otimes M}\right)=\frac{M\Theta}{2}$ 
so that $\mathsf{D}^{\eucal{C}}_{\diamond}\geq \frac{\Theta}{4}(M-N)$. As $\eucal{C}$ is a continuous family of 
unitary channels one can always find $\cU_0, \cU_1\in\eucal{C}$ for which $M\Theta=\pi$ (and $N\Theta<\pi$) which 
implies
	\be
		\mathsf{D}^{\eucal{C}}_{\diamond}\geq \frac{\pi}{4}\left(1-\frac{N}{M}\right)
	\label{app:unitary_bound_rep}
	\ee	
Setting $M=N^{R(\epsilon)}$, with $R(\epsilon)=1+\delta$ for $\delta>0$ implies 
$\lim_{N\tightto\infty}\mathsf{D}^{\eucal{C}}_{\diamond}\geq \frac{\pi}{4}$, 
proving the strong converse on the replication rate. Hence for $\mathsf{D}^{\eucal{C}}_{\diamond}<\nicefrac{\pi}
{4}$ the optimal replication rate is $\mathsf{R}^{\eucal{C}}_{\diamond}(\epsilon)=1$ for all 
$\epsilon<1-\nicefrac{1}{\sqrt{2}}$.
\end{proof}

\section{\label{app:necessary_conds_superrep} Proof of Proposition~\ref{prop:necessary_conds_superrep}}
In order to demonstrate Proposition~\ref{prop:necessary_conds_superrep} we start from the following 
bounds
\be
\begin{split}
		\channelD^{\eucal{C}}_{\rm CJ} &\geq \frac{1}{2}\left(D_{\rm CJ}(\cE_a^{\otimes M}, 
        \cE_b^{\otimes M}) - \channelD^{(N)}\left(\cE_a,\cE_b\right)\right) \\
        &\geq \frac{1}{2}\left(D_{\rm CJ}(\cE_a^{\otimes M}, \cE_b^{\otimes M}) - \frac{1}
        {2}\int_{a}^{b}\sqrt{\mqfi^{(N)}\left(\cE_x\right)}\, \dd x \right)\\
        &\geq \frac{1}{2}\left(D_{\rm CJ}(\cE_a^{\otimes M}, \cE_b^{\otimes M}) - \int_{a}^{b}\sqrt{f_N(\cE_x)}\, 
        \d x \right)
	\label{app:optimal_cloning_distance}
    \end{split}
	\ee
established by Proposition~\ref{prop:cloning_bound_discrimination}, Eq.~\eqref{eq:triangle_metrology} and 
Proposition~\ref{prop:metrology_optimal_distance}. Now observe that  
\be
		\begin{split}
			D_{\rm CJ}\left(\cE_a^{\otimes M},\cE_{b}^{\otimes M}\right) &= \arccos 
            F_{\rm CJ}\left(\cE_a^{\otimes M},
			\cE_{b}^{\otimes M}\right)\\
			&=\arccos \big(F_{\rm CJ}\left(\cE_a,\cE_{b}\right)\big)^M\, ,
		\end{split}
	\ee 
where we have made use of the following property $F(\rho^{\otimes M}, \sigma^{\otimes M})=F(\rho, \sigma)^M$.
For the bound \eqnref{app:optimal_cloning_distance} to remain meaningful in the large $M$ limit, the channels 
$\cE_a$ and $\cE_b$ must become infinitesimal close. Thus, setting $b=a+\delta$ with $\delta= \sqrt{\frac{R}{M}}$ 
for some $R\geq 0$ we obtain the expression 
\be
		\begin{split}
			D_{\rm CJ}\left(\cE_a^{\otimes M},\cE_{a+\delta}^{\otimes M}\right)
			&=\arccos \big(F_{\rm CJ}\left(\cE_a,\cE_{a+\delta}\right)\big)^M\\
			&=\arccos \big(1- \frac{\delta^2}{8}{{\rm{QFI}}(\mathtt{CJ}[\cE_{a}])} + \mathcal{O}
            (\delta^3)\big)^M\,\\
            & = \arccos(\exp[-\frac{R}{8} {{\rm{QFI}}(\mathtt{CJ}[\cE_{a}])} ] )
		\end{split}
	\label{app:upper_bound_targets1}
	\ee 
where we used \eqnref{eq:QFI} in going to the second line and took the large $M$ limit to go to the last.

Next, we turn to the second term in \eqnref{app:optimal_cloning_distance}. 
In the limit of small $\delta=\sqrt{\frac{R}{M}}$ and large $N$ we find 
\be\label{app: bound QFI}
\frac{1}{2}\int_{a}^{a+\delta}\sqrt{\mqfi^{(N)}\left(\cE_x\right)}\, \dd x =  \frac{\delta}{2} 
\sqrt{\mqfi^{(N)}\left(\cE_a\right)} = \sqrt{\frac{R \, \mqfi^{(N)}\left(\cE_a\right)}{4 M}}\leq  \sqrt{\frac{R\, 
f_N(\cE_a)}{M}},
\ee
where we assumed that the next order term $\frac{1}{M} \frac{{\dot \mqfi}^{(N)}\left(\cE_a\right)}
{\mqfi^{(N)}\left(\cE_a\right)}$  (respectively $\frac{1}{M} \frac{{\dot f}_N \left(\cE_a\right)}
{f_N\left(\cE_a\right)}$) vanish in the limit (see regularity conditions discussion in 
Sec.~\ref{app sec: regularity}).

Plugging Eqs.~({\ref{app:upper_bound_targets1}, \ref{app: bound QFI}}) into \eqnref{app:optimal_cloning_distance} 
we obtain  the following bounds for the optimal cloning distance 
	\be
		\begin{split}
			\mathsf{D}^{\eucal{C}}_{\mathrm{CJ}}&\geq\frac{1}{2}\left( \arccos\left(\exp\left[-\frac{R}{8}
            {{\rm{QFI}}(\mathtt{CJ}[\cE_{a}])} \right]\right) 
			- \sqrt{\frac{R \, \mqfi^{(N)}\left(\cE_a\right)}{4 M}}\right)\\
            &\geq \frac{1}{2}\left(\arccos\left(\exp\left[-\frac{R}{8}{{\rm{QFI}}(\mathtt{CJ}[\cE_{a}])} 
            \right]\right)- \sqrt{\frac{R\, f_N(\cE_a)}{M}}\right)
		\end{split}
	\label{app:general_upper_bound}
	\ee
valid for all $R\geq  0$. Our next goal is to find the value of the parameter $R$ which maximizes the right hand 
side. We postpone this optimization to Sec.~\ref{app: bound optimization}.  For now we note that 
for all $0<x,y < \infty$
\be
A(\nicefrac{y}{x}) =\frac{1}{2} \max_{R\geq 0} \left(\arccos\left(\exp[- R x ] \right) - \sqrt{ 2  R y} \right), 
\ee
where $A(Z)$ is a monotonically decreasing function of the real variable $0<Z<\infty$  satisfying 
$A(0)=\frac{\pi}{4}$. A closed form expression of the function $A$ is given in \eqnref{app: function A}.  
This result allows one to rewrite the bound in \eqnref{app:general_upper_bound} as
\be
			\mathsf{D}^{\eucal{C}}_{\mathrm{CJ}}\geq A\left(\frac{\mqfi^{(N)}\left(\cE_a\right)}{M \, 
            {\rm{QFI}}(\mathtt{CJ}[\cE_{a}]) }\right) \geq A\left(\frac{4\, f_N(\cE_a)}{M \, {\rm{QFI}}
            (\mathtt{CJ}[\cE_{a}]) }\right).
	\label{app:general_upper_bound final}
\ee

The first inequality lends itself to a very intuitive interpretation---the argument of the monotonically 
decreasing function $A$ is the ratio between the optimal QFI for any state producing process $\proc{S}
[\cE_a^{\times N}]$ (using $N$ copies of the channel) and the QFI of the Choi-Jamio\l kowski state $\mathtt{CJ}
[\cE_a^{\otimes M}]$ (using $M$ copies of the channel in parallel). 

What remains to be done is to distinguish the cases where $\mqfi^{(N)}\left(\cE_a\right)\leq 4 \, f_N(\cE_a)$ may 
scale quadratically with $N$, and where it is bound to a linear scaling dictated by 
Eq.~\eqref{eq:bound_channel_metro} from \cite{Kurdzialek2023}.  Concretely, given a Kraus decomposition of the 
channel achieving $\beta(a)=0$ we get
\begin{equation}
		f_N(\cE_a)= N \|\alpha(a)\|	\implies \mathsf{D}^{\eucal{C}}_{\mathrm{CJ}}\geq  A\left(\frac{4\, N 
        \|\alpha(a)\|}{M \, {\rm{QFI}}(\mathtt{CJ}[\cE_{a}]) }\right).
\end{equation}
In the asymptotic limit the rhs goes to $\nicefrac{\pi}{4}$ if $\frac{N}{M}\to 0$. Instead, for a linear 
replication rate $M=(1+\lambda)N$ we find 
\begin{equation}
		(i)\, \beta(a)=0 \quad \implies \quad  A\left(\frac{4\, \|\alpha(a)\|}{(1+\lambda)\, {\rm{QFI}}
        (\mathtt{CJ}[\cE_{a}]) }\right).
\end{equation}

In contrast, if $\beta(a)\neq 0$ we find 
\begin{equation}
		f_N(\cE_a)= N\sqrt{\Vert\alpha(x)\Vert}
            \left((N-1)\Vert\beta(x)\Vert +\sqrt{\Vert\alpha(x)\Vert}\right)\implies 
            \mathsf{D}^{\eucal{C}}_{\mathrm{CJ}}\geq  A\left(\frac{N\sqrt{\Vert\alpha(x)\Vert}
            \left((N-1)\Vert\beta(x)\Vert +\sqrt{\Vert\alpha(x)\Vert}\right)}{M \, {\rm{QFI}}(\mathtt{CJ}
            [\cE_{a}]) }\right).
\end{equation}
Here the asymptotic limit gives  $\mathsf{D}^{\eucal{C}}_{\mathrm{CJ}}\to \nicefrac{\pi}{4}$ if $\frac{N^2}{M}\to 
0$ and for a quadratic replication rate $M=(1+\lambda)N^2$ one finds 
\begin{equation}
		(ii)\, \beta(a)\neq 0 \quad \implies \quad  A\left(\frac{4\, \sqrt{\Vert\alpha(x)\Vert} 
        \Vert\beta(x)\Vert}{(1+\lambda)\, {\rm{QFI}}(\mathtt{CJ}[\cE_{a}]) }\right).
\end{equation}

\subsection{Regularity conditions}
\label{app sec: regularity}
Let us now come back to the assumptions used to derive the bounds in \eqnref{app: bound QFI}.
For any $\delta>0$ and $R<\infty$, there is a big enough $M_0$ such that  $\sqrt{\frac{R}{M}}< \delta$ for 
$M\leq M_0$. Hence, if $(i)$ $\beta(x)= 0$ for all $x\in[a,a+\delta)$ we can  write for $M=(1+\lambda)N$
\be
\int_{a}^{a+\sqrt{\frac{R}{M}}} \sqrt{f_N(\cE_x)} \dd x = \sqrt{N} \int_{a}^{a+\sqrt{\frac{R}{M}}} 
\sqrt{\|\alpha(x)\|} \dd x \leq  \sqrt{\frac{N R}{M} \|\alpha(a)\|}+ o(1),
\ee
where the last equality holds provided that $ \frac{\dd}{\dd x}\|\alpha(x)\|\leq \infty$  for $x\in[a,a+\delta)$ 
($\|\alpha(x)\|>0$ because the QFI along the curve must be nonzero). If $(ii)$ $\beta(a)\neq 0$ for 
$M=(1+\lambda)N^2$ we can write 
\begin{align}
\int_{a}^{a+\sqrt{\frac{R}{M}}} \sqrt{f_N(\cE_x)} \dd x &= N \int_{a}^{a+\sqrt{\frac{R}{M}}} 
\sqrt{\sqrt{\|\alpha(x)\|} \| \beta(x)\| +\frac{\|\alpha(x)\|}{N}} \dd x \\
&\leq N \int_{a}^{a+\sqrt{\frac{R}{M}}} \sqrt{\sqrt{\|\alpha(x)\|} \| \beta(x)\|} \dd x +o(1) \\
&=  \sqrt{\frac{N R}{M} \sqrt{\|\alpha(x)\|} \| \beta(x)\|} +o(1)
\end{align}
provided that $\frac{\dd}{\dd x} \sqrt{\sqrt{\|\alpha(x)\|} \| \beta(x)\|} \leq \infty$ for $x\in[a,a+\delta)$, 
which follows from $\frac{\dd}{\dd x}\|\alpha(x)\|\leq \infty$ and $\frac{\dd}{\dd x}\sqrt{\|\beta(x)\|}\leq 
\infty$. Therefore, the regularity conditions for the two cases can be summarized as follows
\be
\text{Regularity conditions} = 
\begin{cases}
(i) \quad \beta(x)=0 \qquad \qquad \quad \frac{\dd}{\dd x} \|\alpha(x)\|<\infty & \forall x \in[0,\delta) \\
(ii) \quad \frac{\dd}{\dd x} \sqrt{\|\beta(x)\|}\leq \infty \quad \frac{\dd}{\dd x} \|\alpha(x)\|<\infty & 
\forall x \in[0,\delta)\,
\end{cases}
\ee
where we have set $a=0$ to match the notation in the main text.

\subsection{Optimization of the bound}
\label{app: bound optimization}
In this section for all $0<x,y < \infty$ we solve the following optimization
\be
\mathrm{f}(x,y) :=\max_{R\geq 0} \left(\arccos\left(\exp[- R x] \right) - \sqrt{ 2 R y} \right).
\ee
Define a new variable $\zeta:= \exp\left[-R x \right] \in [0,1]$  which allows us to rewrite the maximization as 
\begin{align}\nonumber
\mathrm{f}(x,y) &= \mathrm{g}(\nicefrac{y}{x}) := \max_{\zeta \in [0,1]} g(\zeta; Z) \\
g(\zeta; Z) &:= \left(\arccos\left(\zeta \right) -  \sqrt{2 Z \ln(\nicefrac{1}{\zeta})} \right)
\end{align}
with $0<Z=\nicefrac{y}{x}<\infty$. Note that the function $g(\zeta; Z)$ assumes the values $-\infty$ 
and $0$ on the boundaries of the domain ($\zeta=0,1$ respectively), and is smooth on the open interval $(0,1)$. 
Therefore its maximum is attained at the boundary or at a local extremum satisfying 
\begin{align}
\frac{\dd}{\dd \zeta} g(\zeta;Z)=0 \quad &\Leftrightarrow \quad
\frac{1}{\sqrt{1-\zeta^2}} = \frac{\sqrt{2 Z}}{2 \zeta  \sqrt{ \log \left(\frac{1}{\zeta }\right)}}
\quad \Leftrightarrow \quad  \frac{Z}{\zeta^2} -Z + \log(Z) = \log(\frac{Z}{\zeta^2}).
\end{align}
Exponentiating both sides of the equation gives
\be
\exp(\frac{Z}{\zeta^2}) \exp(-Z) Z  =\frac{Z}{\zeta^2} \quad \Leftrightarrow \quad 
( -\frac{Z}{\zeta^2}) \exp(-\frac{Z}{\zeta^2}) = - Z \exp(-Z),
\quad \Leftrightarrow \quad 
y \,e^y = - Z \exp(-Z)
\label{app:Lambert_thing}
\ee
after the change of variable $y= -\frac{Z}{\zeta^2}$. Note that the right hand side in 
the \eqnref{app:Lambert_thing} satisfies $-\frac{1}{e}\leq - Z \exp(-Z) \leq 0$, and the equation is known to 
have two real solutions given by the two branches $(k=0,-1)$ of the Lambert W function $W_k$
\be
y= W_0( - Z \exp(-Z))\, , \qquad y= W_{-1}( - Z \exp(-Z))\, .
\ee
Transforming the variables back gives
\be
\zeta = \sqrt{-\frac{Z}{W_0( - Z \exp(-Z))}}\,  \qquad \zeta = \sqrt{-\frac{Z}{W_{-1}( - Z \exp(-Z))}}\, .
\ee
As $\zeta\in[0,1]$ only two solutions are possible
\be
\zeta = 1 \qquad \text{or} \qquad   \zeta = \sqrt{-\frac{Z}{W_{-1}( - Z \exp(-Z))}} \quad \text{if} \quad Z\leq 1.
\ee
The solution $\zeta=1$ corresponds gives the trivial value $g(1;Z)=0$. Hence, the unique nontrivial solution is 
only possible if $Z\leq 1$, i.e.,  
\be
\mathrm{g}(Z) = \begin{cases}
0 & Z \geq 1 \\
\arccos \zeta(Z) - \sqrt{2Z \ln(\frac{1}{\zeta(Z)})} & Z < 1\, .
\end{cases}
\ee 
It follows that the solutions of $A(Z)= \frac{1}{2} \mathrm{g}(Z)$ are given by 
\be\label{app: function A}
A(Z) = \begin{cases}
0 & Z \geq 1 \\
\frac{1}{2}\left(\arccos \zeta(Z) - \sqrt{2Z \ln(\frac{1}{\zeta(Z)})}\right) & Z < 1
\end{cases} \qquad \text{with} \qquad \zeta(Z) := \sqrt{-\frac{Z}{W_{-1}( - Z \exp(-Z))}}\, .
\ee 
To see that $A(Z)$ (or $\mathrm{g}(Z)$) is a monotonically decreasing function of $Z$, note that $g(\zeta;Z)$ is 
monotonically decreasing in $Z$ for all $\zeta$.

\section{Proof of corollary \ref{cor:no robustness}}
\label{app: zero robustness}

Following the beginning of the proof in the main text, consider the curve 
\be
\{\cE_{x}\, \vert\, x \in [0,\delta)\}\subset\eucal{C} \quad \text{with} \quad \cE_x[\cdot] = \cN[ e^{\ii H x} 
\cdot e^{-\ii H x}]
\ee
inside our set of channels. We now show that we can choose $H$ (not proportional to identity) such that 
$\beta(x)$ in Eq.~\eqref{eq:alphabetas} can be set to zero.
For the natural Kraus representation $K_k(x)= L_k e^{\ii H x}$ of the channels $\cE_x$ we find 
$\dot{K}_k(x)= \ii L_k e^{\ii H x} H$ and 
\begin{equation}
\beta(x) = \sum_{k} \dot {K}_k(x)^\dag K_k(x) = - \ii H e^{-\ii H x} \sum_k L_k^\dag L_k  e^{\ii H x} = -\ii H 
\neq 0, 
\end{equation}
which does not give the desired result. However, exploiting the gauge freedom we can chose any other Kraus 
representation of the form
\be
\bm K(x) =   \hat u(x)  \bm L e^{\ii H x}
\ee
where we have collected the $n$ Kraus operators in a column vector $\bm L = (L_1,\dots L_n)$ and introduced an 
$n\times n$ unitary matrix  $\hat u(x)$ that may depend on $x$.  The derivatives are now given by 
\begin{align}
\dot {\bm K}(x) = \ii  \hat u(x) (\hat h \bm L + \bm L H)e^{\ii H x}
\end{align}
where $\hat h = -\ii \hat{u}^\dag(x) \dot{ \hat{u}}(x)$ is a hermitian $n\times n$ matrix. For the operator 
$\beta(x)$ we now find
\begin{align}
\beta(x)= \dot {\bm K}^\dag \bm K = -\ii e^{-\ii H x}(H +\bm L^\dag \hat h \bm L) e^{\ii H x} \, .
\end{align}
To finish the proof we need to chose $H$ and $\hat h$ such that the last expression becomes zero, i.e. such that 
$H = -\bm L^\dag \hat h \bm L$. Since we are free to chose any $H$ non-proportional to identity it remains to 
show that one can always select $\hat h$ such that 
\be
G =\bm L^\dag \hat h \bm L
\ee
is not a multiple of identity.

To do so, note that there must be two Kraus operators, say $L_1$ and $L_2$, which are linearly independent since 
otherwise $\cN$ is unitary.  Then we can set $\hat h_{ij}=0$ for $i,j\neq 1,2$, and  choose any
\be
G \in {\rm span}_{\mathds{R}}\{L_1^\dag L_1, L_2^\dag L_2, L_1^\dag L_2 +L_2^\dag L_1,\ii (L_1^\dag L_2 -L_2^\dag 
L_1)\}.
\ee
We now consider two cases separately. If either $L_1^\dag L_1$ and $ L_2^\dag L_2$ are not multiples of identity 
the proof can be completed by choosing $G=L_1^\dag L_1$ or $L_2^\dag L_2$. If both $L_1^\dag L_1$ and 
$L_2^\dag L_2$ are multiples of identity then they are of the form $L_i = a_i U_i$ with $a_1,a_2\neq 0$ and 
$U_1\neq e^{\ii \varphi} U_2$ since they must be linearly independent. Hence we have $L_1^\dag L_2 = a_1^* a_2 
U_1^\dag U_2= a V$ with $a\neq 0$ and unitary $V=U_1^\dag U_2$ linearly independent from $\id$. We thus have 
\be
G \in {\rm span}_{\mathds{R}}\{\id, V+V^\dag,\ii (V -V^\dag)\}.
\ee
As $V =\sum_{i=1}^d e^{\ii \lambda_i}\ketbra{i}$ it must have at least two distinct complex eigenvalues 
$e^{\ii \lambda_0}$ and $e^{\ii \lambda_1}$. In the two-dimensional subspace corresponding to these eigenvalues 
we have
\begin{align}\nonumber
V+V^\dag=
 2\left(\begin{array}{cc}
 \cos( \lambda_0) & \\
 &\cos( \lambda_1) \\
 \end{array}\right)\\
 \ii( V- V^\dag) =
 - 2 \left(\begin{array}{cc}
 \sin( \lambda_0) & \\
 &\sin( \lambda_1) \\
 \end{array}\right)\,
 \end{align}
which can not be both proportional to identity. Hence one can always chose $G$ non-
proportional to identity, and the proof is complete.\hfill$\qedsymbol$

\section{The error-mitigation process $\{\proc{M}_s\}$ of Fig.~\ref{fig:M}}
\label{app: EMproc}

In this appendix we discuss the error mitigation process $\{\proc{M}_s\}_{s=0}^1$ shown in Fig.~\ref{fig:M}. To 
characterize the process it is sufficient to understand how it transforms all extremal CP maps $\mathcal{K}
[\cdot] = K \cdot K^\dag$. Hence we compute the CP maps $\proc{M}_0[\mathcal{K}]$ and $\proc{M}_1[\mathcal{K}]$
for all $\mathcal{K}$. Following the quantum circuit in Fig.~\ref{fig:M} it is easy to see that these CP maps 
are of the form $\proc{M}_s[\mathcal{K}] = K_s \cdot K_s^\dag$, with Kraus operators given by
\begin{align}
    K_s = \bra{s}_A {\rm CNOT}_{A\to S}{\rm CNOT}_{S\to A} (K\otimes \id)_{SA}{\rm CNOT}_{S\to A} \ket{0}_A,
\end{align}
where $S$ denotes the above (system) qubit and $A$ denotes the bottom (auxillilary) qubit. Straightforward matrix 
multiplication gives
\be
{\rm CNOT}_{A\to S}{\rm CNOT}_{S\to A} (K\otimes \id)_{SA}{\rm CNOT}_{S\to A} \ket{0}_A
= 
     \left(\begin{array}{cc}
    \bra{0}K \ket{0}& \\
    & \bra{1}K \ket{1}
    \end{array}\right)_S \ket{0}_A +
    \left(\begin{array}{cc}
    \bra{1}K \ket{0}& \\
    & \bra{0}K \ket{1}
    \end{array}\right)_S \ket{1}_A
\ee
where the computational basis is used for the qubit $S$.  After the measurement of the auxiliary qubit we find 
$\proc{M}_s[\mathcal{K}][\cdot] = K_s \cdot K_s^\dag$ with 
\be
 K_0=\left(\begin{array}{cc}
    \bra{0}K \ket{0}& \\
    & \bra{1}K \ket{1}
    \end{array}\right)\\ \quad \text{and} \quad
    K_1 =  \left(\begin{array}{cc}
    \bra{1}K \ket{0}& \\
    & \bra{0}K \ket{1}
    \end{array}\right).
\ee

Finally, we apply this expression for the channels $\cA_\theta = \mathcal{X}_p \circ\cU_\theta, \cB_\theta = 
\cU_\theta \circ \mathcal{X}_p$ in \eqnref{eq: C tilde mnoisy}) and the amplitude-damping 
channel $\cE_\gamma$ in Eq.~\eqref{eq: AD Kraus} to obtain 
\begin{align}
    \proc{M}_0[\cA_\theta] &= (1-p)\, \cU_\theta  \qquad\qquad  \qquad \qquad \qquad \qquad \qquad 
    \proc{M}_1[\cA_\theta] = p \,\cU_\theta\\
    \proc{M}_0[\cB_\theta] &= (1-p)  \, \cU_\theta \qquad \qquad \qquad \qquad \qquad \qquad  \qquad 
    \proc{M}_1[\cB_\theta] = p \, \cU_{-\theta}\\
     \proc{M}_0[\cE_\gamma][\cdot] &=  \left(\begin{array}{cc}
    1& \\
    & \sqrt{1-\gamma}
    \end{array}\right)\cdot \left(\begin{array}{cc}
    1& \\
    & \sqrt{1-\gamma}
    \end{array}\right) \quad  
    \,\, \proc{M}_1[\cE_\gamma][\cdot] =  \left(\begin{array}{cc}
    0& \\
    & \gamma  
    \end{array}\right)\cdot \left(\begin{array}{cc}
    0& \\
    & \gamma
    \end{array}\right).
\end{align}

\section{Cloning Pauli-noise channels}
\subsection{Derivation of Eq.~\eqref{eq: opt F coh}}
\label{app: bitflip fid}

In this section we make explicit the computation of the fidelity $F_{\rm CJ}(\proc{P}[\cN_{\bm p}^{\times N}], \cN_{\bm p}^{\otimes M})$ for the measure-and-prepare cloning process of the Pauli-noise channels, 
$\proc{P}[\cN_{\bm p}^{\times N}]= \sum_{\bm j} \proc{E}_{\bm j}[\cN_{\bm p}^{\times N}]\, \widehat{\cN}^{(M)}_{\bm j}$, with 
\be
\widehat{\cN}^{(M)}_{\bm j} = \left(\bigotimes_{i=1}^N \cP_{j_i}\right)\otimes 
\widehat{\mathcal N}_{\bm j}^{(M-N)},    
\ee
$\cP_{j}[\cdot] =\sigma_j \cdot \sigma_j$ and $Q(\bm j|N):=\proc{E}_{\bm j}[\cN_{\bm p}^{\times N}]= \prod_{k=0}^{3}(p_k)^{t_k(\bm j)}$.
Using $\cN_{\bm p}= \sum_{j=0}^3 p_j \cP_j$  we can write the fidelity as 
\begin{align} \nonumber
F_{\rm CJ}(\proc{P}[\cN_{\bm p}^{\times N}], \cN_{\bm p}^{\otimes M}) &= F_{\rm CJ}\left(\sum_{\bm j}\! Q(\bm j|N) 
\left(\bigotimes_{i=1}^N \cP_{j_i}\right) \otimes \widehat{\mathcal N}_{\bm j}^{(M-N)}, (\sum_{j=0}^3 p_j \cP_j)^{\otimes M} \right) \\ 
& =
F\left(\sum_{\bm j}\! Q(\bm j|N) \left(\bigotimes_{i=1}^N \mathtt{CJ}[\cP_{j_i}]\right) \otimes 
\mathtt{CJ}[\widehat{\cN}_{\bm j}^{(M-N)}],\sum_{\bm j'} Q(\bm j'|M) \left(\bigotimes_{i=1}^M 
\mathtt{CJ}[\cP_{j_i'}]\right) \right)\, .
\end{align}
Using $F(\rho,\sigma)= \tr |\sqrt \rho \sqrt \sigma|$, 
$\sqrt{\mathtt{CJ}[\cP_{j}]}\sqrt{\mathtt{CJ}[\cP_{j'}]} =\delta_{jj'} \, \mathtt{CJ}[\cP_{j}]$ 
and $\sqrt{ \sum_{j=0}^3 p_j \mathtt{CJ}[\cP_{j}] }= \sum_{j=0}^3 \sqrt{p_j} \, \mathtt{CJ}[\cP_{j}]$ 
we obtain 
\begin{align}
F_{\rm CJ}(\proc{P}[\cN_{\bm p}^{\times N}], \cN_{\bm p}^{\otimes M})
& = \tr \left|\sum_{j_i,j_i'} \sqrt{Q(\bm j|N)Q(\bm j'|M)} \bigotimes_{i=1}^N \sqrt{\mathtt{CJ}
[\mathcal{P}_{j_i}]}\sqrt{ \mathtt{CJ}[\mathcal{P}_{j_i'}]} \otimes 
\sqrt{\mathtt{CJ}[\widehat{\mathcal N}_{\bm j}^{(M-N)}]}
\sqrt{\bigotimes_{i=N+1}^M  \mathtt{CJ}[\mathcal{P}_{j_i'}]}\right|
\\
& = \tr \left| \sum_{j_1,\dots j_N}  Q(\bm j|N) \bigotimes_{i=1}^N \mathtt{CJ}[\mathcal{P}_{j_i}]   \otimes
\sqrt{\mathtt{CJ}[\widehat{\mathcal N}_{\bm j}^{(M-N)}]}
 \sqrt{ (\sum_{j=0}^3 p_j \mathtt{CJ}[\cP_{j}])^{\otimes(M-N)} }
\right|
\\
&= \sum_{\bm t } {\rm Pr}_{\rm Mult}(\bm t|N,\bm p)  F_{\rm CJ}\left( \widehat{\mathcal N}_{\bm j}^{(M-N)},\cN_{\bm p}^{\otimes (M-N)}\right).
\end{align}

For the choice $\widehat{\mathcal N}_{\bm j}^{(M-N)}=\cN_{\hat {\bm p}(\bm t)}^{\otimes (M-N)}$ we obtain 
\begin{align}
F_{\rm CJ}(\proc{P}[\cN_{\bm p}^{\times N}], \cN_{\bm p}^{\otimes M})
& = \sum_{\bm t} {\rm Pr}_{\rm Mult}(\bm t|N,\bm p) \, F_{\rm CJ}(\cN_{\hat {\bm p}(\bm t)} ,\cN_{\bm p})^{M-N}\, .
\end{align}

\subsection{Computing the asymptotic cloning fidelity of the coherent process}
\label{app: BF}
In this section we lower bound the cloning fidelity of the Pauli-noise channels
(Section~\ref{sec:BF})
\begin{align}
     {\rm F}_{\rm CJ}^{\eucal{C}|\rm M\&P} =  \min_{\bm p } \sum_{\bm t} {\rm Pr}_{\rm Mult}(\bm t|N,\bm p) \, F_{\rm CJ}(\cN_{\hat {\bm p}(\bm t)} ,\cN_{\bm p})^{M-N}
\end{align}
for the specific choice of the estimator $\hat{p}_k(\bm t)= \frac{t_k}{N}$. To do so we will lower bound the final 
fidelity for all values of the parameter $\bm p$, using two approaches.

First, consider the following lower bound given by Jensen's inequality
\begin{align}
F_{\rm CJ}(\proc{P}[\cN_{\bm p}^{\times N}], \cN_{\bm p}^{\otimes M}) 
&\geq \left(\sum_{\bm t} {\rm Pr}_{\rm Mult}(\bm t|N,\bm p) \, F_{\rm CJ}(\cN_{\hat {\bm p}(\bm t)} ,\cN_{\bm p}) \right)^{M-N} \\ \label{eq: bound BF F 1}
&= \mathds{E}\left[\sum_{k=0}^3 \sqrt{p_k \, \hat p_k(\bm t)}\right]^{M-N},
\end{align}
where we used Eq.~\eqref{eq:FCJ_bitflip} for the single copy fidelity. Here the four random variables $N \hat p_k(\bm t)\sim 
{\rm Bin}(N,p_k)$  are binomially distributed. To lower bound the last 
expression we can thus use the following inequality for a binomially distributed random variable $X\sim {\rm Bin}
(N,p)$
\be
 \mathds{E}[\sqrt{X}] \geq \sqrt{pN}\left(1-\frac{1-p}{2 p N}\right),
\label{eq: bin bound}
\ee
which is derived in \appref{app: binom r.v.} following \cite{mathoverflow2013}. Plugging \eqnref{eq: bin bound} 
into \eqnref{eq: bound BF F 1} one immediately obtains
\begin{align}
F_{\rm CJ}(\proc{P}[\cN_{\bm p}^{\times N}], \cN_{\bm p}^{\otimes M}) 
&\geq \left(1- \sum_{k=0}^3 \frac{1-p_k}{2 N}\right)^{M-N}=  \left(1-  \frac{3}{2 N}\right)^{M-N} \, \forall \bm p \quad \implies \quad \mathsf{F}_{\rm CJ}^{\eucal{P}} \geq \left(1-\frac{3}{2 N}\right)^{M-N}. 
\end{align}

Finally, considering a smaller family of Pauli-noise channels $\eucal{P}=\{\cN_{\bm p}| \sum p_j =1, p_j \geq \epsilon\}$ for 
some $\epsilon>0$, we obtain a tight expression for the asymptotic worst-case alignment fidelity
\be
\sum_{\bm t} {\rm Pr}_{\rm Mult}(\bm t|N,\bm p) \, F_{\rm CJ}(\cN_{\hat {\bm p}(\bm t)} ,\cN_{\bm p})   = 1 - \frac{3}{8 N}\, .
\ee
To do so we note that in the large $N$ limit the distribution of $\hat {\bm p}(\bm t)$ concentrates around the mean value 
$\mathds{E}[\hat {\bm p}(\bm t)]= \bm p$ and we can expand $F_{\rm CJ}(\cN_{\hat {\bm p}},\cX_{\bm p}) = 1 -\sum_k \frac{(\hat p_k-p_k)^2}{8 p_k} + 
\mathcal{O}(\|\hat {\bm p} -\bm p\|^3)$. Then
\begin{align}
\sum_{\bm t} {\rm Pr}_{\rm Mult}(\bm t|N,\bm p) \, F_{\rm CJ}(\cN_{\hat {\bm p}(\bm t)} ,\cN_{\bm p})&= 
1 -\sum_{k=0}^3   \frac{\sum_{\bm t} {\rm Pr}_{\rm Mult}(\bm t|N,\bm p)(\hat p_k N-p_k N)^2}{8 N^2 p_k}\\
&= 1 -\sum_{k=0}^3 \frac{(1-p_k)}{8 N} = 1-\frac{3}{8 N},
\label{app:second_moment}
\end{align}
where the sum in the first line is nothing more than the second moment of the binomial distribution ${\rm Bin}(N,p_k)$, given by $N p_k(1-p_k)$. In the asymptotic limit with $M=(1+\lambda)N$ we thus find for any restricted set of Pauli-channels the cloning fidelity satisfies
\be
\mathsf{F}_{\rm CJ}^{\eucal{P}} \geq \exp(-\frac{3}{8} \lambda).
\ee

\subsection{Square root of a binomial random variable}
\label{app: binom r.v.}

For completeness, in this section we summarize the demonstration of the bound 
\be \label{eq: h1}
 \mathds{E}[\sqrt{X}] \geq \sqrt{p N}\left(1-\frac{1-p}{2 p N}\right) 
\ee
for a binomial random variabel $X\sim {\rm Bin}(N,p)$, given in the discussion \cite{mathoverflow2013} on 
Mathoverflow. First, for $x\geq 0$ consider the inequality
\be
\sqrt x \geq 1- \frac{(1-x)}{2}-\frac{(1-x)^2}{2} \Longleftrightarrow 2\sqrt{x} \geq (3-x) x \Longleftrightarrow 
2- 3(\sqrt x) + (\sqrt x)^3 \geq  0 \Longleftrightarrow
(\sqrt x -1)^2 (\sqrt x +2) \geq 0.\ee
With its help, for any nonnegative random variable we can write 
\begin{align}
  \frac{\mathds{E}[\sqrt X]}{\sqrt{\mathds{E}[X]}} =\mathds{E}\left[\sqrt{\frac{X}{ \mathds{E}[X]}}\right] \geq 
  \mathds{E}\left[1 - 
   \frac{(1-X/ \mathds{E}[X])}{2}-\frac{(1-X/ \mathds{E}[X])^2}{2}\right]= 1-\frac{\mathds{E}[X^2]-\mathds{E}
   [X]^2}{2 \, \mathds{E}[X]^2}
\end{align}
Plugging in $\mathds{E}[X]= p N$ and $\mathds{E}[X^2]= p  \left(p N^2-p  N+N\right)$ gives 
\eqnref{eq: h1}.

\section{Cloning of amplitude-damping channels}

\subsection{The dummy process for amplitude-damping channels}
\label{app: AD dummy}
Recall that the Kraus operators of the amplitude-damping channels $\cE_\adp, \,\adp\in[0,1]$ are given by
\begin{equation}
    		K_0(\adp) = \begin{pmatrix} 
				1 & 0 \\ 
				0 & \sqrt{1-\adp}
			   \end{pmatrix}, \quad 
    		K_1(\adp) = \begin{pmatrix} 
				0 & \sqrt{\adp} \\ 
				0 & 0
			     \end{pmatrix}\, 
	\end{equation}
in the computational basis. In order to derive the best dummy cloning process we need to find the channel 
$\cE_{\rm dum}$ maximizing 
\be
\min_{\adp\in[0,1]} F_{\rm CJ} (\cE_\adp, \cE_{\rm dum}).
\ee
For simplicity, let us first relax this problem by only minimizing the Choi fidelity with repect to the extremal 
amplitude-damping channels $\cE_0 = \text{id}$ and $\cE_1$ (we show later that this actually leads to the optimal 
solution of the full problem), i.e.
\be
\min \left\{ F_{\rm CJ} (\text{id}, \cE_{\rm dum}), F_{\rm CJ} (\cE_1, \cE_{\rm dum})\right\}.
\ee

Consider the  Choi-Jamio\l{}kowski states associated to these three channels. For the target channels we have
\be
\mathtt{CJ}[\text{id}] =
\left( \begin{array}{cccc}
\frac{1}{2} & & & \frac{1}{2} \\
&0&&\\
&&0&\\
\frac{1}{2} & & & \frac{1}{2}
\end{array}
\right) \qquad
\mathtt{CJ}[\cE_1] =
\left( \begin{array}{cccc}
\frac{1}{2} & & & \\
&0&&\\
&&\frac{1}{2}&\\
 & & & 0
\end{array}
\right) \, .
\ee
The most general  Choi-Jamio\l{}kowski state that can be associated to the dummy channel is of the form 
\be
\mathtt{CJ}[\cE_{\rm dum}] =
\left( \begin{array}{cccc}
\frac{p}{2} & \dots& \frac{b}{2}& \frac{a}{2} \\
\dots&\frac{1-p}{2}&\dots&\dots\\
\frac{b^*}{2}&\dots&\frac{1-q}{2}&\dots\\
\frac{a^*}{2} &\dots &\dots & \frac{q}{2}
\end{array}\right),
\ee
with $0\leq p,q\leq 1$ and $|a|,|b| \leq 1$, where we only specified the elements that will be important for the 
following discussion. Note that the state must satisfy $\tr_S\mathtt{CJ}[\cE_{\rm dum}] = \frac{1}{2}\id$ which 
is readily guaranteed by the chosen parametrization. In addition, it must be  positive semi-definite $\mathtt{CJ}
[\cE_{\rm dum}]\geq 0$ which in particular implies $|a|^2\leq p q$ and $|b|^2\leq p(1-q)$. Since both target 
states are real, we can assume that $a$ and $b$ are also real without loss of generality.

Let us now compute the target fidelities, starting with the identity channel. We find
\begin{align}
F_{\rm CJ} ({\rm id}, \cE_{\rm dum}) &= \tr |\sqrt{\mathtt{CJ}[{\rm id}]} \sqrt{ \mathtt{CJ}[\cE_{\rm dum}] }| = 
\frac{1}{2} \sqrt{p+q+2 a} \leq \frac{1}{2} \sqrt{p+q+2 \sqrt{p q}}
\end{align}
where the last bound is saturated by setting $a=\sqrt{p q}$. For $\cE_1$ we find 
\begin{align}
F_{\rm CJ} (\cE_1, \cE_{\rm dum}) &= \tr |\sqrt{ \mathtt{CJ}[\cE_{1}] } \sqrt{ \mathtt{CJ}[\cE_{\rm dum}] }| \\
&= \frac{\sqrt{p+1-q-\sqrt{4 b^2+(p+1-q)^2}}+\sqrt{p-1-q+\sqrt{4 b^2+(p+1-q)^2}}}{2 \sqrt{2}} \\
&\leq \frac{\sqrt{p+1-q}}{2},
\end{align}
where the last inequality is tight for $b=0$. Note that the upper-bounds can be attained simultaneously, without 
compromising the positivity of $\mathtt{CJ}[\cE_{\rm dum}]$. Hence, this is the optimal choice for the parameters 
$a$ and $b$, and it remains to optimize over $p$ and $q$. As both expressions we computed are increasing with 
$p$, we set $p=1$ as the optimal choice. This leads to 
\begin{align}
     F_{\rm CJ} ({\rm id}, \cE_{\rm dum}) &= \frac{1}{2}\sqrt{1+2\sqrt{q}+q}\\
     F_{\rm CJ} (\cE_1, \cE_{\rm dum}) &= \frac{1}{2} \left(1+ \sqrt{1-q}\right)\, 
\end{align}
and their minimum occurs at $q=\frac{1}{2}$
\be
F_{\rm CJ} ({\rm id}, \cE_{\rm dum}) = F_{\rm CJ} (\cE_1, \cE_{\rm dum}) =\frac{2+\sqrt 2}{4}.
\ee
We have thus shown that the optimal  Choi-Jamio\l{}kowski states for the relaxed problem is
\be
\mathtt{CJ}[\cE_{\rm dum}] =
\left( \begin{array}{cccc}
\frac{1}{2} & & & \frac{1}{2 \sqrt 2} \\
&0&&\\
&&\frac{1}{4}&\\
\frac{1}{2 \sqrt 2} & & & \frac{1}{4}
\end{array}\right),
\ee
with the dummy channel corresponding to an AD channel with $\gamma=1/2$. 
Finally, it is straightforward to check that
\be
F_{\rm CJ} (\cE_\adp, \cE_{\rm dum}) \geq F_{\rm CJ} ({\rm id}, \cE_{\rm dum}) = F_{\rm CJ} (\cE_1, 
\cE_{\rm dum}) =\frac{2+\sqrt 2}{4}, 
\ee
which implies  that the dummy channel we constructed is also optimal for the original problem, and 
\be
\max_{\cE_{\rm dum}} \min_{\adp \in [0,1]} F_{\rm CJ} (\cE_\adp, \cE_{\rm dum})=\frac{2+\sqrt 2}{4}.
\ee

\subsection{Computing the fidelity of the coherent process for cloning AD channels}
\label{app: ad coherent}

Here we compute the cloning fidelity $F_{\rm CJ}(\proc{P}[\cE_\adp^{\times N}], \cE_\adp^{M})$ for the coherent 
AD cloning process discussed in Sec.~\ref{sec:coherentprocess}. In the main text we have established that the 
process returns the $M$-qubit CPTP map
\begin{align}\nonumber
\proc{P}[\cE_\adp^{\times N}] &=\! \sum_{s_1,\dots,s_N}\!\left( \mathcal{K}_{s_1|\gamma} \otimes\dots\otimes 
\mathcal{K}_{s_N|\gamma} \otimes (\cE_{\hat{\adp}(\abs{\bf s})})^{\otimes (M-N)} \right) \\
& = \! \sum_{s_1,\dots,s_M}\!\left( \mathcal{K}_{s_1|\gamma} \otimes\dots\otimes \mathcal{K}_{s_N|\gamma} \otimes 
\mathcal{K}_{s_{N+1}|\hat{\gamma}(\abs{\bf s})} \otimes\dots\otimes \mathcal{K}_{s_M|\hat{\gamma}
(\abs{\bf s})}\right)
\end{align}
where $\abs{\bf s}=\sum_{i=1}^N s_i$. Henceforth we shall write $\hat \gamma(\abs{\bf s})=\hat \gamma$ for short. 
Analogously the ideal channel $\cE_\adp^{\otimes M}$ can be expanded as 
\begin{align}
\cE_\adp^{\otimes M}
& = \! \sum_{s_1,\dots,s_M}\!\left( \mathcal{K}_{s_1|\gamma} \otimes\dots\otimes \mathcal{K}_{s_N|\gamma} \otimes 
\mathcal{K}_{s_{N+1}|\gamma} \otimes\dots\otimes \mathcal{K}_{s_M|\gamma}\right)\,.
\end{align}
For both channels we find that the corresponding Choi-Jamio{\l}kowski states read
\begin{align}\nonumber
    \mathtt{CJ}\big[\proc{P}[\cE_\gamma^{\times N}]\big] &= \sum_{s_1,\dots,s_M}\!\left( \Psi_{s_1|\gamma} 
    \otimes\dots\otimes \Psi_{s_N|\gamma} \otimes \Psi_{s_{N+1}|\hat \gamma} \otimes\dots\otimes \Psi_{s_M|\hat 
    \gamma}\right)\\
    \mathtt{CJ}\big[\cE_\gamma^{\otimes M}\big]&= \sum_{s_1,\dots,s_M}\!\left( \Psi_{s_1|\gamma} 
    \otimes\dots\otimes \Psi_{s_N|\gamma} \otimes \Psi_{s_{N+1}|\gamma} \otimes\dots\otimes \Psi_{s_M| 
    \gamma}\right)
\end{align}
where $\Psi_{s|\gamma}= \ketbra{\Psi_{s|\gamma}}$ with 
\be
\ket{\Psi_{s|\gamma}} = (K_s(\gamma) \otimes \id )\ket{\Phi^+} =\frac{1}{\sqrt 2}
\begin{cases}
    \ket{00}+ \sqrt{1-\gamma}\ket{11}& s=0 \\
    \sqrt{\gamma} \ket{01} & s=1
\end{cases}.
\ee
Crucially, because of their parity for different values for $s$, $\braket{\Psi_{0|\gamma}}{\Psi_{1|\gamma'}}=0$ 
for all $\gamma$ and $\gamma'$. This allows us to write
\begin{align}
   \left| \sqrt{\mathtt{CJ}\big[\proc{P}[\cE_\gamma^{\times N}]\big]} \sqrt{\mathtt{CJ}
   \big[\cE_\gamma^{\otimes M}\big]} \right|&= \sum_{s_1,\dots,s_M}\!\left( \Psi_{s_1|\gamma} \otimes\dots
   \otimes \Psi_{s_N|\gamma} \otimes \left|\sqrt{\Psi_{s_{N+1}|\hat \gamma}}\sqrt{\Psi_{s_{N+1}| \gamma}}\right| 
   \otimes\dots\otimes \left|\sqrt{\Psi_{s_{M}|\hat \gamma}}\sqrt{\Psi_{s_{M}| \gamma}} \right|\right).
\end{align}
Then, for ${\rm Pr}^{(1)}(s|\gamma) = \tr \Psi_{s|\gamma}= 
\begin{cases}
1-\frac{\gamma}{2} & s=0 \\
\frac{\gamma}{2} & s=1
\end{cases}
$, we get
\begin{align}\nonumber
    F_{\rm CJ}(\proc{P}[\cE_\adp^{\times N}], \cE_\adp^{M}) &= \tr  \left| \sqrt{\mathtt{CJ}\big[\proc{P}
    [\cE_\gamma^{\times N}]\big]} \sqrt{\mathtt{CJ}\big[\cE_\gamma^{\otimes M}\big]} \right| \\ \nonumber
    & = \sum_{s_1,\dots,s_M}\!\left( {\rm Pr}^{(1)}(s_1|\gamma)  \times\dots\times {\rm Pr}^{(1)}(s_N|\gamma)  
    \times \tr \left|\sqrt{\Psi_{s_{N+1}|\hat \gamma}}\sqrt{\Psi_{s_{N+1}| \gamma}}\right| \times\dots\times \tr 
    \left|\sqrt{\Psi_{s_{M}|\hat \gamma}}\sqrt{\Psi_{s_{M}| \gamma}} \right|\right) \\\nonumber
    & = \sum_{s_1,\dots,s_N}\!\left( {\rm Pr}^{(1)}(s_1|\gamma)  \times\dots\times {\rm Pr}^{(1)}(s_N|\gamma)  
    \times \left(\fcj(\cE_{\hat \gamma(t)},\cE_{\gamma})^{M-N}\right) \right)\\
    &= \sum_{t=0}^N\, {\rm Pr}_{\rm Bin}(t|N,\nicefrac{\gamma}{2})\, \fcj(\cE_{\hat \gamma(t)},
    \cE_{\gamma})^{M-N}\, 
\end{align}
with the binomial distribution ${\rm Pr}_{\rm Bin}(t|N,\nicefrac{\gamma}{2})$.

\section{Semidefinite programming techniques}\label{app:sdp}

In this appendix we give a more detailed exposition of the SDP approximation of the optimal cloning task 
presented in Sec.~\ref{sec:sdp} in the main text. Before we begin, we will establish a notation convention. In 
the text we defined the Choi-Jamilkowski isomorphism to be normalised. Here, we denote the non-normalised version 
with an overline, since some of the identities that we will use deal with non-normalised objects. For 
$\ket{\Phi^+}=\nicefrac{1}{\sqrt{d_1}}\, \sum_{i=0}^{d_1-1} \ket{ii}$
\begin{alignat}{2}
    {\mathtt{CJ}}[\cA] &:= ({\rm id} \otimes \cA) \left [\ketbra{\Phi^+}\right] \quad \in \cL(\cH_{\rm 1 }\otimes 
    \cH_{\rm 2 })\\
     \overline{\mathtt{CJ}}[\cA] &:= d_{\rm 1} \!\cdot\!{\mathtt{CJ}}[\cA].
\end{alignat}

Our goal now is to give an explicit form of the feasibility program of \eqnref{eq:pseudoopt} that can be executed on  a classical computer. To avoid unnecessary complications we will start with a simpler 
problem of  maximizing the fidelity $\fcj (\proc{P}[\cE^{\times N}],\cE^{\otimes M} )$ with 
respect to the process $\proc{P}$ for a fixed channel $\cE$. As argued in the main text, 
following ~\cite{Skrzypczyk2023}, this maximization can be cast as the following SDP, similar to the one given 
in the main text,
\begin{subequations}\label{eq:fidsdp}
\begin{align} 
\underset{{\proc{P}} }{\textbf{max}}\, \fcj (\proc{P}[\cE^{\times N}],\cE^{\otimes M} )
  =\; \underset{\proc{P},Y}{\textbf{max}}&\quad \tfrac{1}{2}\tr (Y +Y^\dagger)\\
     \textbf{subject to}&\quad
		\begin{pmatrix}
					 \mathtt{CJ}[\proc{P}[\cE^{\times N}]]& {Y} \\
					Y^\dagger  & \mathtt{CJ}[\cE^{\otimes M}]
				\end{pmatrix} \ge 0 \\
&\quad \textrm{process constraints on } \proc{P}\, .
\end{align}
\end{subequations}

\begin{wrapfigure}{h}{0.38\textwidth}
    \begin{center}
    \vspace{-0.4 cm}
    \includegraphics[width=0.3 \textwidth]{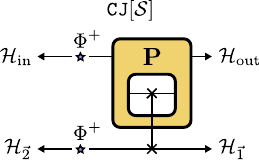}
    \end{center}
     \vspace{-0.2 cm}
    \caption{A circuit preparing~\cite{Taranto2025} the \ChJa state $\mathtt{CJ}[\mathcal{S}]$ representing the 
    process $\proc{P}:\cL(\cH_{\vec 1}\tightto\cH_{\vec 2})\tightto \cL(\cH_{\rm in}\tightto\cH_{\rm out})$. Each 
    source (star) prepares the maximally entangled state $\ket{\Phi^+}$ in $\cH_{\rm in}\otimes\cH_{\rm in}$, 
    respectively $\cH_{\vec{2}}\otimes\cH_{\vec{2}}$.}
    \vspace{-0.5 cm}
    \label{fig:CJ}
\end{wrapfigure}

What remains is to parametrize the variable $\proc{P}$ and specify its constraints.  The process $\proc{P}$ takes 
as input N copies of a CPTP map and returns as output another CPTP map, formally
\begin{align}
  \proc{P} :  \bigotimes_{i=1}^{N} \cL( \cH_{1_i} \tightto \cH_{2_i})  &\to \cL(\cH_{\rm in} \tightto 
  \cH_{\rm out})  \\
   \cE^{\otimes N} & \mapsto\proc{P}[\cE^{\times N}].
\end{align}
To avoid notational clutter we will use the notation $\cL(\cH_{\vec{1}}\,\tightto\, \cH_{\vec{2}})$ to mean 
$\bigotimes_{i=1}^{N} \cL( \cH_{1_i} \tightto \cH_{2_i})$, i.e. tensor products of linear maps 
$\cE^{(i)} \in \cL(\cH_{{1}_i}\tightto \cH_{{2}_i})$. Here, both $\cE^{\otimes N}$ and $\proc{P}[\cE^{\times N}]$ 
can be represented by their \ChJa operators 
\be
\overline{\mathtt{CJ}}[\cE^{\otimes N}]\in \cL(\cH_{\vec{1}} \otimes \cH_{\vec{2}}) \qquad \text{and} \qquad 
\overline{\mathtt{CJ}}[\proc{P} [\cE^{\times N}]]\in \cL (\cH_{\rm in} \otimes \cH_{\rm out}).\nonumber
\ee

In turn, every process $\proc{P}$ has a corresponding \textit{representing map}, which is the induced CP linear 
map on these operators
\begin{align}\label{eq:repP}
   \cS :
    \cL(\cH_{\vec{1}} \otimes \cH_{\vec{2}})
  &\to  \cL (\cH_{\rm in} \otimes \cH_{\rm out}) \\
  \overline{\mathtt{CJ}}[\cE^{\otimes N}] &\mapsto
    \cS( \overline{\mathtt{CJ}}[\cE^{\otimes N}]) := \overline{\mathtt{CJ}}[\proc{P} [\cE^{\times N}]].
\end{align}
By the \ChJa isomorphism, the action of the final channel $\proc{P} [\cE^{\times N}]$ on any input state 
in $\cL(\cH_{\rm in})$ can be obtained from the representing map via
\begin{align}\label{eq: CJ1app}
    \proc{P}[\cE^{\times N}][\,\cdot\,] = \text{tr}_{\rm in} \Big((\id_{\rm out} \otimes \,
    \cdot\,{}^{\rm T}])\,\cS(\overline{\mathtt{CJ}}[\cE^{\otimes N}]) \Big).
\end{align}

The representing map $\cS$ itself can be associated with its \ChJa operator 
$\overline{\mathtt{CJ}}[\cS] \in \cL(\cH_{\rm in}\otimes \cH_{\rm out}\otimes \cH_{\vec{1}}\otimes 
\cH_{\vec{2}})$,
so that its action on operators $R\in \cL(\cH_{\vec{1}}\otimes \cH_{\vec{2}})$ also follows from the \ChJa 
isomorphism
\begin{align}\label{eq:pact}
  \cS(R) &:=\tr_{\vec{1} \,\vec{2}}[(\id_{\rm in}\otimes \id_{\rm out}\otimes R^{\rm T})
    \overline{\mathtt{CJ}}[\cS]]\quad \in \cL(\cH_{\rm in}\otimes \cH_{\rm out}),
\end{align}
where the trace over $\vec{1}$ indicates all spaces $\cH_{1_i}$ and similarly for $\vec{2}$. Similarly to the 
case of channels, the \ChJa state ${\mathtt{CJ}}[\cS]$ representing the process 
$\proc{P}$, can be prepared by embedding the process in a simple circuit composed with sources of maximally 
entangled states and SWAP gates (see Fig.~\ref{fig:CJ} and Ref.~\cite{Taranto2025}). Substituting 
\eqnref{eq:pact} into \eqnref{eq: CJ1app} for the variable $\mathtt{CJ}[\proc{P}
[\cE^{\times N}]]$ our optimization variable appearing in \eqnref{eq:fidsdp}
\begin{align}\nonumber
     \mathtt{CJ}[\proc{P}[\cE^{\times N}]] & 
   =\tfrac{1}{d_{\rm in} d_{\rm out}}\cS(\overline{\mathtt{CJ}}[\cE^{\otimes N}]).\\
    & = \label{eq: rho app} 
   \tfrac{1}{{d_{\rm in}}{d_{\rm out}}}\tr_{\vec{1}\, \vec{2}}\Big(\id_{\rm in
   }\otimes \id_{\rm out }\otimes {\overline{\mathtt{CJ}}[\cE^{\otimes N}]}^{\rm T})\,\overline{\mathtt{CJ}}
   [\cS]\,\Big),
\end{align}
is now in terms of the operator $\overline{\mathtt{CJ}}[\cS]\simeq \cS \simeq \proc{P}$.

Now, it is crucial that the conditions for $\proc{P}$ to be a valid process---parallel (par), sequential (seq) 
or general (non-causal, nc) in Fig.~\ref{fig:parseqnc}---can be translated into simple linear and semidefinite 
constraints on the operator $\overline{\mathtt{CJ}}[\cS]$.  More precisely, one can write
\begin{align}\label{app:process const}
\textrm{process constraints on } \proc{P}\qquad \simeq \qquad  \begin{cases}
    \overline{\mathtt{CJ}}[\cS] \ge 0\\
   \tr \overline{\mathtt{CJ}}[\cS]  = d_{\rm in}d_{\vec{2}}\\
    \widetilde{P}^\bullet[\overline{\mathtt{CJ}}[\cS] ] = \overline{\mathtt{CJ}}[\cS] \qquad \bullet = 
    \{\rm par, seq, nc\},
    \end{cases}
\end{align}
where $\widetilde P^\bullet: \cL(\cH_{\rm in}\otimes \cH_{\rm out}\otimes \cH_{\vec{1}}\otimes 
\cH_{\vec{2}}\tightto \cH_{\rm in}\otimes \cH_{\rm out}\otimes \cH_{\vec{1}}\otimes \cH_{\vec{2}})$ are linear 
projective  maps that can be derived by a method introduced in Theorem~2 of Ref.~\cite{Milz2024characterising} 
(see the example in the next section), and we used the shorthand notation  $d_{\vec{2}}:=\Pi_i \,d_{2_i}$. Hence, 
using \eqnsref{eq: rho app}{app:process const} we can rewrite the optimization of \eqnref{eq:fidsdp} in 
a more explicit form
\begin{subequations}\label{eq:sdpprog}
\begin{alignat}{1} 
  \underset{{\proc{P}} }{\textbf{max}} \,
  \fcj (\proc{P}[\cE^{\times N}] , \cE^{\otimes M} )\;
  =\quad \underset{\overline{\mathtt{CJ}}[\cS], Y}{\textbf{max}}  &\quad\tfrac{1}{2}\tr (Y + Y^\dagger) \\
     \label{eq:matopt2}
     \textbf{subject to}
		&\begin{pmatrix}
		\tfrac{1}{{d_{\rm in}}{d_{\rm out}}}\tr_{\vec{1}\,\vec{2}}[(\id_{\rm in}\otimes \id_{\rm out}
            \otimes {\overline{\mathtt{CJ}}[\cE^{\otimes N}]}^{\rm T})\,\overline{\mathtt{CJ}}[\cS]]\;\; & {Y} \\
		Y^\dagger & \mathtt{CJ}[{\cE^{\otimes M}}]
		\end{pmatrix} \ge 0 \\
            \quad &\quad\tr   \overline{\mathtt{CJ}}[\cS] = d_{\rm in}d_{\vec{2}}\\
            \quad &\quad\overline{\mathtt{CJ}}[\cS] \ge 0\\
            \quad &\quad \widetilde{P}^\bullet[\overline{\mathtt{CJ}}[\cS]] = \overline{\mathtt{CJ}}[\cS]  
            \qquad \text{ for } \bullet =\{\rm par, seq, nc\}\,.
\end{alignat}
\end{subequations}
To execute it on a computer it remains to compute the desired  projective map $\widetilde{P}^\bullet$, which 
depend on the specific choice of $N$. 

Before discussing concrete examples, recall that originally we are not interested in maximizing the fidelity 
$ \fcj (\proc{P}[\cE^{\times N}] , \cE^{\otimes M} )$ for a fixed channel. Instead, we want to verify 
(\eqnref{eq:pseudoopt}) whether for some $\proc{P}$ a certain fidelity $ \fcj (\proc{P}[\cE_i^{\times N}] , 
\cE_i^{\otimes M} )>x$ is feasible for all channels $\cE_i$ from a discrete set $\eucal{N}$ with $|\eucal{N}|=H$. 
Hence, instead of the optimization given in \eqnref{eq:sdpprog} we are actually interested in the following  \emph{feasibility} SDP
\begin{subequations}\label{eq:sdpfeas}
\begin{alignat}{3} 
  \quad &\textbf{find} \quad \overline{\mathtt{CJ}}[\cS]\,,
  \{Y_i\}_{i=1}^H
  &&\\
     &\textbf{subject to}  
		&&\begin{pmatrix}
    \tfrac{1}{d_{\rm in}d_{\rm out}}\tr_{\vec{1}\vec{2}}[(\id_{\rm in}\otimes \id_{\rm out}\otimes 
    {\overline{\mathtt{CJ}}[\cE_i^{\otimes N}]}^{\rm T})\,\overline{\mathtt{CJ}}[\cS]]\;\; & {Y_i}\\
    Y_i^\dagger & \mathtt{CJ}[{\cE^{\otimes M}_i}]\end{pmatrix} \ge 0 
    \quad \forall i\\
    &\quad &&\tr   \overline{\mathtt{CJ}}[\cS] =  d_{\rm in}d_{\vec{2}}\\
    &\quad &&\overline{\mathtt{CJ}}[\cS] \ge 0 \\
    &\quad && \widetilde{P}^\bullet[\overline{\mathtt{CJ}}[\cS]] = \overline{\mathtt{CJ}}[\cS]  \qquad
    \text{ for } \bullet =\{\rm par, seq, nc\}\\
    &\quad && \tfrac{1}{2}\tr (Y_i+ Y_i^\dagger)  \ge x \quad \forall i.
\end{alignat}
\end{subequations}

\subsection{Example of process constraints}

\begin{figure}
    \centering
    \includegraphics[width=0.7\linewidth]{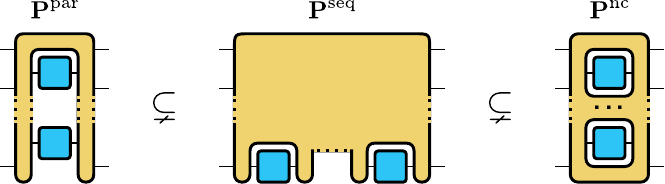}
    \caption{The yellow boxes represent parallel (par), sequential (seq) and general (non-causal, nc) processes.}
    \label{fig:parseqnc}
\end{figure}

As concrete examples, we give the projectors for the parallel ($\rm{par}$), sequential ($\rm{seq}$), and 
general (non-causal, $\rm{nc}$) processes (represented in Fig.~\ref{fig:parseqnc}) when the number of input channels is $N =2$, i.e., for
\begin{align}
\proc{P}: (\cL(\cH_{1_1}\rightarrow \cH_{2_1})
\otimes  
\cL(\cH_{1_2}\rightarrow \cH_{2_2}))\longrightarrow
\cL(\cH_{\rm in}\rightarrow \cH_{\rm out}),
\end{align}
in terms of the \textit{trace-and-replace} operator ${}_{A}{X}:= \tr_A X\otimes \tfrac{\id}{d_A}$, which is 
idempotent $_{A}({}_{A}X) = {}_{A}X$ and commutative ${}_{AB}X = {}_{BA}X$. The projectors read 
\begin{align}
 \widetilde{P}^{{\rm par}}[X] 
      &=
       X -  {}_{\rm out}X  +{}_{2\,\rm{out}}X 
      - {}_{12\,\rm{out}}X 
      + {}_{\rm{in} \,12\,\rm{out}}X \\
          \widetilde{P}^{\rm seq}[X] &: =
          X 
      - {}_{\rm{out}}X
      + {}_{2_2 \,\rm{out}}X
    - {}_{1_22_2\,\rm{out}}X
    + {}_{2_11_12_2\rm{out}}X
    - {}_{1_12_12_12_2\,\rm{out}}X
      + { }_{\rm{in}1_12_11_22_2\,\rm{out}}X\\
           \begin{split}
          \widetilde{P}^{\rm NC }[X] &: 
      = 
                X 
      - {}_{\rm{out}}X
      + {}_{2_2\,\rm{out}}X
    - {}_{2_1 2_2\,\rm{out}}X
    + {}_{1_2\,\rm{out}}X
    - {}_{1_2 2_2\,\rm{out}}X
    + {}_{1_2 2_1 2_2\,\rm{out}}X
    - {}_{1_1 2_1\,\rm{out}}X
    + {}_{1_1 2_1 2_2\,\rm{out}}X\\
   &\qquad\qquad\qquad\qquad\qquad \qquad\qquad\qquad - {}_{1_1 2_1 2_1 2_2\,\rm{out}}X
      + { }_{\rm{in}1_1 2_1 1_2 2_2\,\rm{out}}X.
      \end{split}
\end{align}

\subsection{Constraining measure-and-prepare processes}

Consider the physical realization (Fig.~\ref{fig:CJ}) of the \ChJa state $\mathtt{CJ}[\mathcal{S}]$ representing 
a process $\proc{P}$. It is obvious that for a measure-and-prepare process, the representing  \ChJa state is 
separable across the bi-partition $\cH_{\vec{1}}\otimes \cH_{\vec{2}}\,|\,\cH_{\rm in}\otimes \cH_{\rm out}$--- 
only classical information can be exchanged between these systems. Separability of $\mathtt{CJ}[\mathcal{S}]$ 
implies a positive partial transpose.  Hence all measure-and-prepare processes satisfy the SDP 
constraint 
\be
\proc{P}\,\text{ is M\&P}\quad \implies \quad \overline{\mathtt{CJ}}[\cS] ^{\rm T(\cH_{\rm in},
\cH_{\rm out})}\geq 0, 
\ee
which can be directly plugged in the SDP of Eq.~\eqref{eq:sdpfeas}. This is of course a relaxation of the 
measure-and-prepare set of processes, but can nevertheless be useful to upper-bound their performance in a given 
task.

\subsection{Cloning $1\to 2$ copies of the amplitude-damping map}
Numerically we tested the qubit amplitude-damping channel $\cE_\adp[\,\cdot\,] := K_0(\adp) \cdot 
K_0^\dag(\adp) +K_1(\adp) \cdot K_1^\dag(\adp)$ with Kraus operators given by 
\begin{align}\label{eq:adchan}
K_0 = \begin{pmatrix} 
	1 & 0 \\ 
	0 & \sqrt{1-\gamma}
	\end{pmatrix}, \quad 
K_1 = \begin{pmatrix} 
	0 & \sqrt{\gamma} \\ 
	0 & 0
	\end{pmatrix}.
	\end{align}
Note the natural asymmetry of the problem at the end points: 
\begin{align}
    \gamma = 0 \qquad &\Longrightarrow \quad \cE_0 = { \rm id},\\
    \gamma = 1 \qquad &\Longrightarrow \quad \cE_1 = \cT_{\ketbra{0}{0}}
\end{align}
where $\cT_{\ketbra{0}{0}}$ trashes any input and returns $\ketbra{0}{0}$. In other words, the extremal points 
are a unitary map and the \textit{trash-and-replace} map. These extremal points will prove to be the problematic 
ones in the numerics (see Table~\ref{tab:num}). For the superchannel that clones $1\to 2$ copies of the qubit 
amplitude-damping channel $d_{\rm in}=d_{\rm out} = 4$, the superchannel acts on the following spaces 
\begin{align}
     \proc{P}: (\cL(\cH^{\otimes 2})\to \cL(\cH^{\otimes 2}))    \longrightarrow (\cL(\cH^{\otimes 4})\to 
     \cL(\cH^{\otimes 4})),
\end{align}
and its corresponding respresenting map (the induced linear CP map on \textit{Choi} operators) by
\begin{align}
    \cS: \cL(\cH^{\otimes 2} \otimes \cH^{\otimes 2})
    \longrightarrow 
    \cL(\cH^{\otimes 4} \otimes \cH^{\otimes 4}), 
\end{align}
i.e., a map which takes dimension $d_{} = 4$ matrices on its input to $d_{} = 16$ matrices on its output. The 
Choi of this map is defined as the action on two copies of the input space
\begin{align}
  \overline{\mathtt{CJ}}[\cS] = \id\otimes \cS \big(\sum_{ij}\ket{i} \!\ket{j} \!\bra{i}\!\bra{j}\big) 
  = \sum_{ij}\ketbra{i}{j} \otimes \cS(\ketbra{i}{j}) \quad\in \cL(\cH^{\otimes 2} \otimes \cH^{\otimes 2}) 
  \otimes \cL(\cH^{\otimes 4} \otimes \cH^{\otimes 4}). 
\end{align}
Thus, the \ChJa operator of the representing map $\overline{\mathtt{CJ}}[\cS] $ for $1\to 2$ cloning is a matrix of 
dimension $d=2 \cdot 2\cdot 4\cdot 4 = 64$.
Table~\ref{tab:num} and Fig.~\ref{fig:num} we present the numeric results of $1\to 2$ cloning using the {\tt SeDuMi} solver. 

\begin{figure}[t!]
    \centering
        \includegraphics[width=0.48\textwidth]{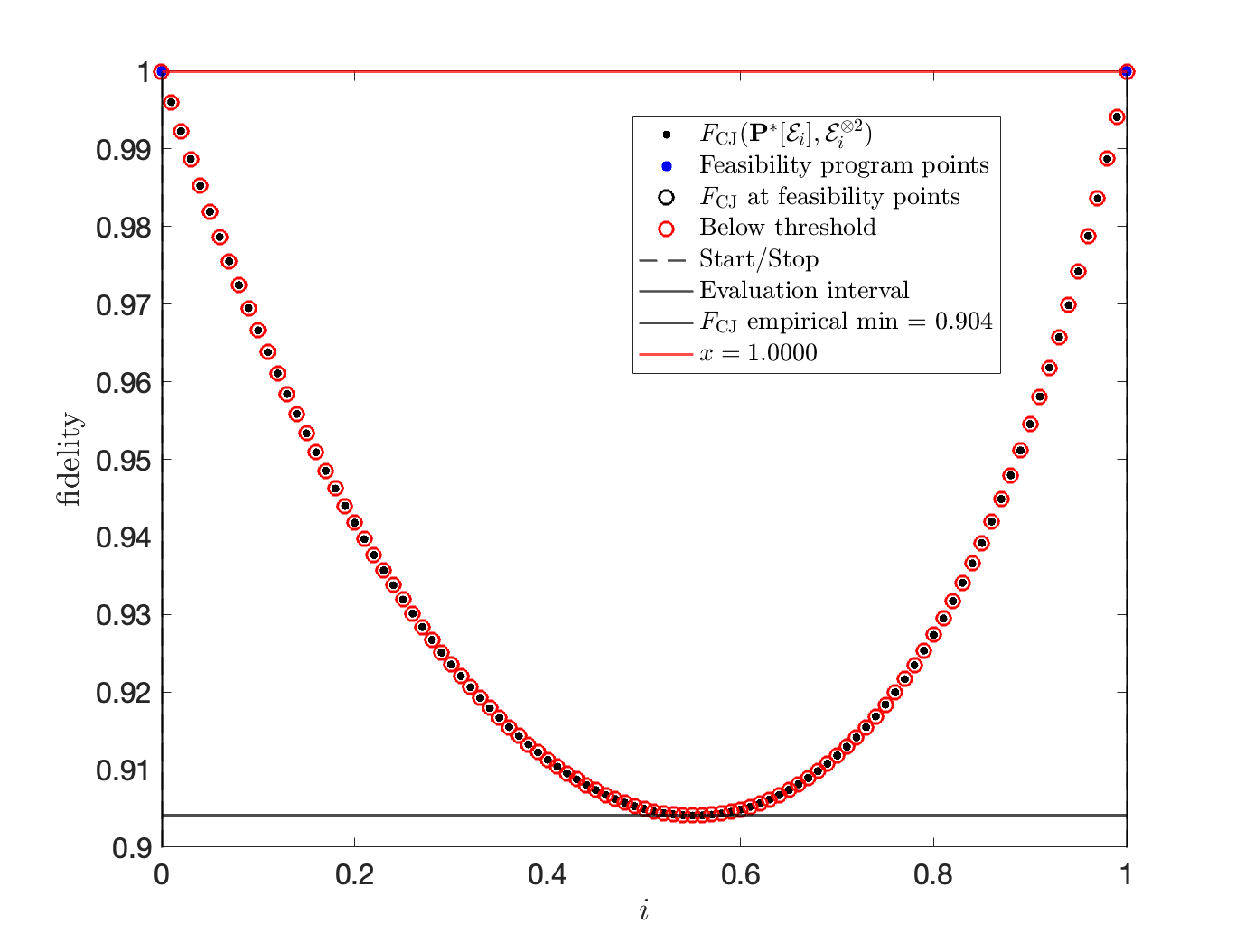}
        \includegraphics[width=0.48\textwidth]{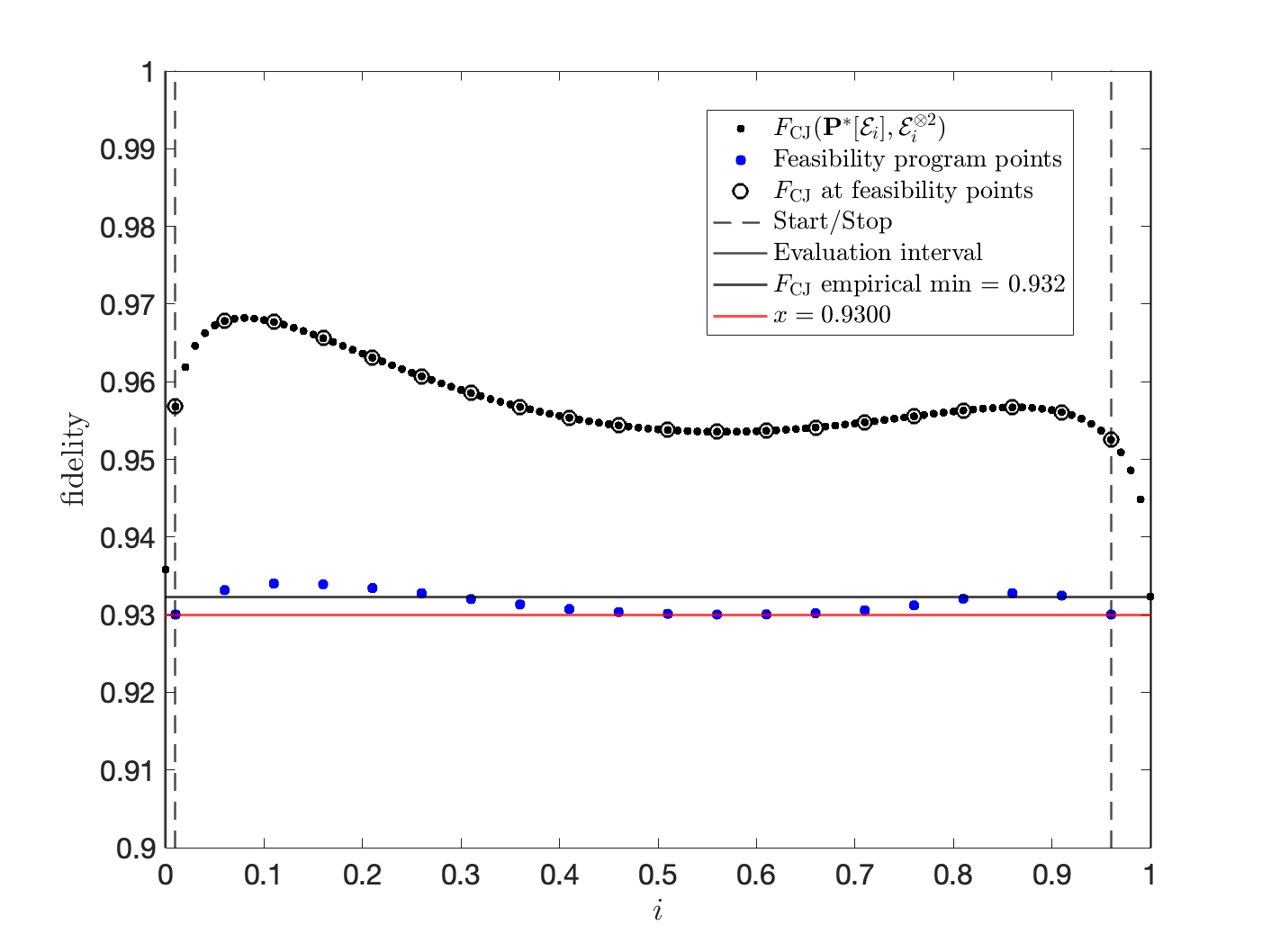}
    \caption{
    Two figures depicting optimal cloning processes for the amplitude damping map $\cE_\gamma$ found by the feasibility SDP in \eqref{eq:sdpfeas}, under different conditions. The CJ fidelity of the cloning process is evaluated on 101 AD maps on the interval $[0, 1]$ to create the plots.
    {(\bf Left)} The optimal process for the net consisting of two points $\cN = \{0, 1\}$. A feasibility threshold of $x= 1.0$ was set and the computer returned the optimal process---one which distinguishes perfectly between $\cE_0$ and $\cE_1$. All other points evaluate to a strictly lower fidelity. 
    $[0, 1]$ {\bf (Right)}
    The optimal process for the net $\cN$ consisting of 21 points in $[0.05, 0.96]$ with a feasibility threshold of $x=0.93$. The feasibility is seen to be a lower bound on the true fidelity and even points outside of the program evaluate to give fidelities that exceed the required threshold. Including the end points in the program results in infeasible runs or worse results, due to numerical instabilities. }
    \label{fig:num}
\end{figure}

\bibliography{super}
\end{document}